\documentclass[11pt,letterpaper]{article}
\ifdefined\pdfminorversion\pdfminorversion=7\fi
\usepackage[margin=1in]{geometry}
\usepackage[T1]{fontenc}
\usepackage{newtxtext}
\usepackage{amsmath,amsthm,mathtools}
\usepackage{newtxmath}
\usepackage{graphicx,booktabs,multirow}
\usepackage{microtype,xcolor,enumitem,array,url}
\usepackage[authoryear,round]{natbib}
\usepackage[font=small,labelfont=bf,labelsep=period]{caption}
\usepackage[unicode=true]{hyperref}
\hypersetup{hidelinks,bookmarksnumbered=true,bookmarksopen=true,bookmarksopenlevel=1,
  pdftitle={Neural Eigenfluid Bases: Spectral Objectives, Gauge, and Topology},
  pdfauthor={Baiming Zhang, Shiying Xiong, Zhecheng Wang}}
\DeclareRobustCommand{\docref}[2]{\hyperref[#2]{#1~\ref*{#2}}}
\DeclareRobustCommand{\doceqref}[2][Eq.]{\hyperref[#2]{#1~\textup{(\ref*{#2})}}}
\newcommand{\doi}[1]{doi: \begingroup\urlstyle{rm}\href{https://doi.org/#1}{\nolinkurl{#1}}\endgroup}
\graphicspath{{figures/}}
\newtheorem{proposition}{Proposition}

\newcommand{\R}{\mathbb R}
\newcommand{\C}{\mathcal C}
\newcommand{\Harm}{\mathcal H}
\newcommand{\ip}[2]{\langle #1,#2\rangle}
\newcommand{\norm}[1]{\lVert #1\rVert}
\newcommand{\curl}{\operatorname{curl}}
\newcommand{\diver}{\operatorname{div}}
\newcommand{\sg}{\operatorname{sg}}
\newcommand{\dd}{\,\mathrm d}
\title{Neural Eigenfluid Bases:\\Spectral Objectives, Gauge, and Topology}
\author{%
\begin{tabular}{ccc}
Baiming Zhang & Shiying Xiong & {Zhecheng Wang\thanks{{Corresponding author.}}}\\[3pt]
\small Zhejiang University & \small Zhejiang University & \small University of Toronto\\[3pt]
\small\href{mailto:baimingzhang@zju.edu.cn}{baimingzhang@zju.edu.cn} &
\small\href{mailto:shiying.xiong@zju.edu.cn}{shiying.xiong@zju.edu.cn} &
\small\href{mailto:zhecheng@cs.toronto.edu}{zhecheng@cs.toronto.edu}
\end{tabular}}

\date{}
\begin{document}
\maketitle
\begin{abstract}
\begingroup
Reduced-order models simplify fluid simulation by representing complex flows with a small number of spatial modes. Building on Laplacian Eigenfluids, we introduce Neural Eigenfluid Bases (NEB), a neural representation whose basis functions can be evaluated anywhere in space. We enforce incompressibility and impermeable boundaries directly through differential lifts and potential boundary conditions. Geometry-local orthogonalization and Rayleigh minimization then learn an ordered basis. We connect learning in potential and velocity spaces, and derive a spectral-variance relation explaining why these objectives weight approximate modes differently, even though their exact positive spectra agree. Gauge fixing reduces the number of spatial derivatives needed for potential learning, while a Dirichlet energy loss learns independent circulation fields around obstacles. With 1024 modes, both objectives keep velocity errors below 1\% in periodic vortex merger. For random flows around obstacles, potential-based learning captures more initial energy and gives lower final errors than numerical Laplacian bases of the same size. The reduced state is about a thousand times smaller than the full reference model. Three-dimensional tests separate spatial reconstruction errors from errors in coefficient evolution. Precomputed Galerkin operators allow millisecond-scale updates when the same basis is reused across simulations. This provides a physically constrained basis for continuous flow reconstruction and efficient reduced-order simulation. Code and models are available at \href{https://github.com/baiming-zhang/NEB}{\textcolor{blue}{https://github.com/baiming-zhang/NEB}}.
\par\endgroup
\end{abstract}

\section{Introduction}
\begingroup
Interactive simulation, initial-condition studies, and geometry design often require solving similar flow problems many times. In airfoil design, engineers may repeatedly adjust shape parameters to identify the dominant flow modes associated with candidate geometries. In games, moving obstacles continually reshape the fluid domain and demand rapid visual feedback. Recomputing numerical bases and flow solutions for each configuration can make this interactive exploration costly. Reduced-order models (ROMs) reduce this cost by representing complex flows with a small set of spatial basis functions. A trained geometry-conditioned neural representation queries these eigenmodes directly. Once the modes and reduced operators are available, the flow evolves in a low-dimensional coefficient space, exploiting the greatly reduced number of degrees of freedom to achieve millisecond-scale updates and support near-real-time flow previews (\docref{Section}{sec:inference-cost}). The choice of basis determines what the reduced model can represent, how physical constraints are handled, and which computations can be reused. A useful fluid basis should therefore capture the main flow structures while supporting dynamics from different initial conditions.

For conventional spectral ROMs, however, the basis is usually computed by solving a Laplacian eigenvalue problem on a specific domain. Even a small change in the boundary or the position of an obstacle can change the eigenfunctions and require a new numerical solve. Analytical solutions are available for a few simple geometries, such as rectangles and circles, but are generally difficult to obtain for more complex shapes. This becomes a practical limitation during early-stage design, where geometry parameters may change repeatedly. Neural networks offer another approach. Instead of solving a separate eigenvalue problem for each geometry, we can learn a mapping from shape space to eigenspace, allowing basis functions to be evaluated at arbitrary spatial locations within a family of geometries. Our goal is to preserve the physical structure and accuracy of Laplacian Eigenfluids while reducing the need for repeated numerical solves. Our experiments show that learned bases can achieve accuracy comparable to numerical Laplacian bases, and sometimes perform better in selected cases. For repeated queries, the cost of training can also be shared across many evaluations.

Fluid model reduction has developed along two main directions: data-driven and operator-based methods. Snapshot-based approaches extract a low-dimensional space from existing simulations. \citet{treuille2006model} demonstrated interactive reduced fluid simulation, including interactions with moving obstacles. \citet{kim2013subspace} combined subspace integration with cubature to efficiently simulate variations of high-resolution flows. These methods use patterns in simulation data to reduce spatial computation and time-stepping costs. Spectral methods instead construct basis functions from differential operators. Laplacian Eigenfluids \citep{dewitt2012fluid} organizes flow structures by eigenvalues and provides a modal interpretation of viscous decay. Model-Reduced Variational Fluid Simulation \citep{liu2015model} further connects spectral reduction with geometric structure and harmonic fields. Unlike bases extracted from specific flow trajectories, operator-based modes are ordered by differential energy and have a clear physical interpretation. These properties provide a natural starting point for our neural representation.

Neural networks have introduced several other ways to reuse computations in fluid simulation. Deep Fluids learns a parameterized velocity generator and latent dynamics from simulation data \citep{kim2019deep}. DeepONet and Fourier Neural Operators learn mappings between function spaces \citep{lu2021deeponet,li2021fourier}, while Geo-FNO extends this approach to general geometries through learned deformations \citep{li2023geometry}. These methods mainly focus on learning how physical inputs map to flow solutions. We take a different approach: rather than learning the flow solution directly, we learn the spatial basis used to represent it. Once the basis and its reduced operators are prepared for a given geometry, they can be reused across different initial conditions by changing only the coefficient state. This also allows us to study spatial representation error and reduced dynamics error separately.

Coordinate-based neural fields provide continuous function values and spatial derivatives. Periodic activation functions help represent differentiable signals \citep{sitzmann2020implicit}, while Fourier features improve the learning of high-frequency variations \citep{tancik2020fourier}. Variational methods can then use physical energy to learn eigenfunctions. The Deep Ritz method minimizes energy functionals \citep{e2018deep}, Spectral Inference Networks learns multiple eigenfunctions through stochastic optimization \citep{pfau2019spectral}, and neural Rayleigh methods apply similar ideas to engineering eigenvalue problems \citep{rowan2025rayleigh}. More recently, Shape Space Spectra extended neural eigenfunction learning to parameterized geometries, using geometry-local sorting and causal orthogonalization to handle changes in spectral ordering \citep{chang2025shapespacespectra}. These works show that neural networks can learn bases defined by differential operators. For incompressible fluids, however, minimizing spectral energy alone is not enough. The basis must also satisfy incompressibility and wall impermeability, account for gauge freedom in the potential representation, and include the independent circulation fields introduced by domain topology.

Physical constraints can be built directly into neural representations. Neural Conservation Laws (NCL) constructs divergence-free fields through differential lifts of antisymmetric potentials \citep{richterpowell2022neural}, and other constrained neural basis methods use similar ideas for incompressible flows \citep{cheng2026decoupled}. Neural Kinematic Bases (NKB) learns velocity basis functions directly, using losses for smoothness, divergence, wall alignment, and orthogonality to support mesh-free flow fitting and animation \citep{liu2025neural}. These studies show that learning fluid bases with neural networks is feasible. However, soft penalties do not strictly enforce all physical constraints, and they do not directly organize the learned functions according to the spectrum of a differential operator.

We aim to reduce how much of the physics the network must learn from penalties. Instead of relying on the loss function to enforce incompressibility, we build it into the representation and use energy minimization to learn the spectral modes. In two dimensions, a scalar streamfunction generates a divergence-free velocity field; in three dimensions, we use a vector potential. Potential boundary conditions enforce wall impermeability, geometry-local orthogonalization separates the modes, and Rayleigh minimization determines their spectral ordering. The training does not require flow snapshots, since the basis is defined by the geometry and differential operator. The neural network provides continuous spatial functions, the physical construction restricts the space of possible solutions, and the energy objective selects the low-frequency modes. We also compare our approach with NKB under the same training budget and Galerkin evolution framework to examine how these structural choices affect accuracy.

Learning in potential space offers another advantage: it can reduce the number of spatial derivatives needed during training. However, this requires a consistent treatment of gauge freedom. For example, adding a constant to a periodic streamfunction in two dimensions does not change the velocity. In three dimensions, certain curl-free components can also be added to a vector potential without changing its curl. Although these changes leave the physical velocity unchanged, they affect the potential mass norm and can therefore distort the Rayleigh objective. We address this by selecting a suitable gauge representative before normalization and optimization. Using the classical compatible positive-spectrum correspondence \citep{arnold2006finite}, we connect the potential-space and velocity-space objectives and derive a spectral-variance relation that explains why they assign different weights to high-frequency components in approximate modes.

Our experiments also reveal a practical difference between two- and three-dimensional learning. In 2D, the potential objective requires fewer spatial derivatives and reduces training time while maintaining good accuracy. In 3D, however, selecting a consistent vector-potential gauge introduces a substantial computational cost, while the accuracy gains over numerical Laplacian bases are relatively small in our tests. We therefore view potential-space learning as an alternative with a clear mathematical basis, rather than an objective that is always preferable. Its usefulness depends on the problem and the cost of gauge selection. We examine auxiliary solvers to understand how gauge accuracy depends on spatial resolution and operator consistency.

Beyond the positive spectral modes, flow in complex domains contains independent circulation fields determined by topology. Hodge theory shows that harmonic velocity fields can carry circulation around obstacles without introducing vorticity \citep{arnold2006finite}. This structure has been studied in spectral variational ROMs \citep{liu2015model}, Fluid Cohomology \citep{yin2023fluid}, and Hodge-based neural operator learning \citep{zheng2026hodge}. In our framework, we learn these harmonic fields through fixed boundary traces and Dirichlet energy minimization, then convert them into velocity fields through differential lifts. Together, the positive modes and harmonic coordinates represent both vortical motion and independent circulation around obstacles.

Circulation is especially important in flow around lifting airfoils. In idealized airfoil flow, it is closely related to the Kutta condition and the generation of lift. A basis containing only positive spectral modes cannot represent the independent circulation needed for such a flow, regardless of how accurately those modes are learned. Harmonic completion therefore provides more than a mathematical correction; it adds degrees of freedom needed to describe certain physical phenomena. However, when the total number of modes is fixed, adding harmonic fields also means removing the same number of positive modes. Although this makes the representation topologically complete, it does not always improve energy capture. The benefit depends on how much of the target flow lies in the vortical and harmonic sectors.

\begin{samepage}
Our main contributions are as follows:
\begin{enumerate}[leftmargin=*,itemsep=3pt,topsep=4pt]
\item \textbf{Physically constrained neural eigenfluid bases.} We introduce a continuous, geometry-conditioned spectral representation that combines differential lifts, impermeable boundary conditions, and geometry-local orthogonalization. The resulting basis satisfies key kinematic constraints by construction. For a selected geometry, reusable Galerkin operators support autonomous coefficient evolution and flow reconstruction at arbitrary spatial locations.
\item \textbf{Potential-space spectral learning with gauge selection.} We formulate a Rayleigh learning objective on selected potential representatives and derive a spectral-variance relation connecting potential-space and velocity-space objectives through their shared positive spectrum. We also examine the accuracy and computational cost of both approaches in two and three dimensions.
\item \textbf{Neural harmonic circulation fields.} We learn independent harmonic coordinates using fixed boundary traces and Dirichlet energy minimization, completing the flow representation in multiply connected domains. Through fixed-capacity comparisons and controlled tests of boundary-lift conditioning, we study how these coordinates affect energy capture and flow evolution.
\end{enumerate}
\end{samepage}

\par\endgroup

\section{Method}
\label{sec:method}
\begingroup
The construction has two stages (\docref{Figure}{fig:method}). During learning, the differential lifts of Neural Conservation Laws \citep{richterpowell2022neural} and potential boundary traces restrict the candidates to divergence-free, impermeable velocities. Geometry-local spectral orthogonalization follows Shape Space Spectra \citep{chang2025shapespacespectra}; harmonic circulation fields complete the span as in spectral fluid reduction \citep{liu2015model}. For a selected geometry, initial projection and preassembled Galerkin operators turn this spatial basis into an autonomous coefficient model, which we advance using the coefficient-space evolution equations of Laplacian Eigenfluids \citep{dewitt2012fluid}. The resulting state can be reconstructed at any requested spatial locations and output times.
\par\endgroup
\begin{figure}[!htb]
\centering\includegraphics[width=\linewidth,height=.72\textheight,keepaspectratio]{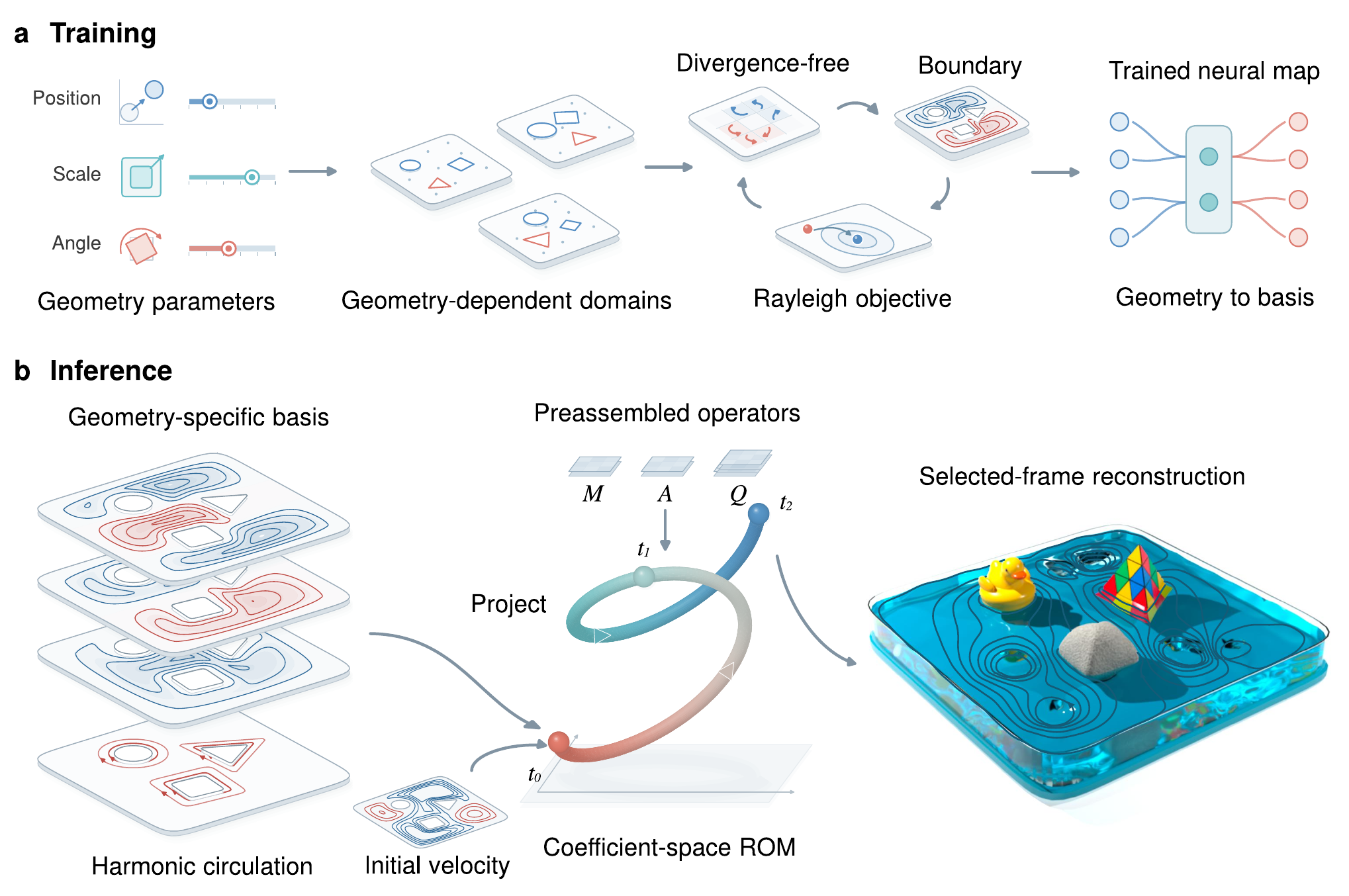}
\caption{\textbf{Shape-dependent training and coefficient-only inference.} (a) Geometry parameters vary obstacle position, scale and orientation within a fixed rounded-square outer boundary. A hard boundary envelope acts on geometry-conditioned potentials. Paired, oppositely oriented red and blue marks depict the antisymmetric tensor in the Neural Conservation Laws (NCL) construction {\citep{richterpowell2022neural}}; its row-wise divergence generates divergence-free velocities. {Gauge fixing and geometry-local causal orthogonalization produce normalized candidates within the divergence-free, impermeable trial space; the Rayleigh contours depict their energy minimization.} Training yields the parameter-to-basis neural map. (b) For a selected geometry, the map supplies positive modes, completed by harmonic circulation fields. Initial projection and geometry-specific preassembled mass, linear and quadratic operators support coefficient-only ROM stepping; spatial fields are reconstructed at selected times. The schematic uses $M$ for the velocity mass matrix, $A$ for the reduced linear operator, $Q$ for the quadratic advection tensor, and $t_0,t_1,t_2$ for selected output times. }
\label{fig:method}
\end{figure}

\subsection{The velocity space: positive modes and harmonic circulation}
Let $\Omega_{\mathbf g}\subset\R^D$ be a connected fluid domain of spatial dimension $D\in\{2,3\}$, indexed by geometry parameters $\mathbf g$, with outward unit normal $\mathbf n$. Here $\R$ is the real number space, $\mathbf x=(x_1,\ldots,x_D)$ is a spatial point, written $(x,y)$ in 2D or $(x,y,z)$ in 3D, and $\mathbf u$ denotes velocity. For scalar or vector fields $a,b$, we use the geometry-local mass inner product
\begin{equation}
\ip{a}{b}_{\mathbf g}=\int_{\Omega_{\mathbf g}}a\cdot b\,\dd\mathbf x,
\qquad \norm a_{\mathbf g}^2=\ip a a_{\mathbf g}.
\label{eq:metric}
\end{equation}
The dot is the Euclidean inner product (ordinary multiplication for scalars); $\dd\mathbf x$ is volume measure. For a fixed domain we write $\Omega=\Omega_{\mathbf g}$ and suppress the geometry subscript on the inner product and norm. The operators $\diver$ and $\curl$ denote spatial divergence and curl.
{An admissible velocity satisfies $\diver\mathbf u=0$ and $\mathbf u\cdot\mathbf n=0$: the wall is impermeable and tangential motion is permitted.} We order this space by vorticity energy in 2D and curl energy in 3D, giving the divergence-free absolute Hodge problem {\citep{arnold2006finite}}. Its zero-energy fields form the harmonic sector; its positive-energy eigenfields form the positive sector. To express both dimensions together, write $\alpha=\mathbf u^\flat$ for the velocity one-form, with $\flat$ the metric dual, $d$ for exterior differentiation, $\delta$ for its Hilbert adjoint, and $\star$ for the Hodge star. The contraction $\iota_{\mathbf n}$ extracts the normal trace. The subscript $A$ denotes absolute boundary conditions; the space $\Harm_A^1$ of harmonic velocity one-forms is
\begin{equation}
\Harm_A^1=\{\alpha:d\alpha=0,\ \delta\alpha=0,\ \iota_{\mathbf n}\alpha=0\}.
\label{eq:harmonic-space}
\end{equation}
These fields have zero Hodge energy but can carry circulation. For positive mode index $i$, let $\Phi_i$ be a normalized velocity eigenfield and $\lambda_i$ its eigenvalue. On admissible velocities orthogonal to the harmonic sector, they satisfy the weak eigenproblem for every admissible test field $v$:
\begin{equation}
\ip{d\Phi_i^\flat}{d v^\flat}=\lambda_i\ip{\Phi_i}{v},
\qquad \lambda_i>0,
\label{eq:target}
\end{equation}
\begingroup
The absolute Hodge realization has $\Phi_i\cdot\mathbf n=0$ and $\iota_{\mathbf n}d\Phi_i^\flat=0$; the latter is zero boundary vorticity in 2D and $\mathbf n\times\curl\Phi_i=0$ in 3D. These are the natural Hodge slip conditions of the spectral problem. The envelope imposes impermeability directly; the energy determines the natural eigenfield condition.
\par\endgroup

\subsection{Admissible neural potentials and the differential lift}
\begingroup
Geometry determines the domain, quadrature and boundary envelope. Following the geometry-conditioned spectral representation of \citet{chang2025shapespacespectra}, the network $f_\theta(\mathbf x,\mathbf g)$ has trainable parameters $\theta$ and outputs $K$ potential candidates $f_{\theta,i}$, indexed by $i=1,\ldots,K$, one per mode. The symbol $\widetilde\mu_i$ denotes the enveloped candidate before gauge selection and normalization: a scalar streamfunction in 2D and a vector potential in 3D. Geometry-family training samples $\mathbf g$ during optimization; fixed-domain training specializes this map to one geometry.
\par\endgroup

On a bounded domain we multiply each candidate by a differentiable geometry envelope $e_{\mathbf g}$ that vanishes on the domain boundary $\partial\Omega_{\mathbf g}$,
\begin{equation}
\widetilde\mu_i(\mathbf x,\mathbf g)=e_{\mathbf g}(\mathbf x)f_{\theta,i}(\mathbf x,\mathbf g),
\qquad e_{\mathbf g}|_{\partial\Omega_{\mathbf g}}=0.
\label{eq:envelope}
\end{equation}
The zero envelope enforces a sufficient relative potential trace: zero streamfunction trace in 2D and zero tangential vector-potential trace in 3D. The periodic benchmark uses a periodic encoding and identity envelope. For a potential $\mu$, write $\psi$ for its scalar streamfunction in 2D and $\mathbf A=(A_1,A_2,A_3)$ for its vector potential in 3D. With $\nabla$ the spatial gradient, $\partial_x,\partial_y$ coordinate derivatives, and $\nabla^\perp$ the rotated gradient, the potential-to-velocity map $\C=\star d$ specializes to
\begin{equation}
\mathbf u=\C\mu=
\begin{cases}
\nabla^\perp\psi=(\partial_y\psi,-\partial_x\psi),&D=2,\\
\curl\mathbf A,&D=3.
\end{cases}
\label{eq:lift}
\end{equation}
Equivalently, the Neural Conservation Laws (NCL) construction represents velocity as the row-wise divergence of an antisymmetric matrix field \citep{richterpowell2022neural}: $u_i=\sum_j\partial_j\mathsf S_{ij}$ with $\mathsf S=-\mathsf S^{\top}$. Here $u_i$ is velocity component $i$, $\partial_j$ differentiates coordinate $j$, the component indices run from $1$ to $D$, and $\top$ denotes transpose. For smooth fields, antisymmetry and commutation of mixed derivatives give
\begin{equation}
\diver\mathbf u=\frac12\sum_{i,j}\partial_i\partial_j
\bigl(\mathsf S_{ij}+\mathsf S_{ji}\bigr)=0.
\label{eq:ncl-antisymmetric}
\end{equation}
\begingroup
In 3D the equivalent antisymmetric potential is the $3\times3$ tensor $\mathsf S$. With $\epsilon_{ijk}$ the Levi--Civita symbol, normalized by $\epsilon_{123}=1$, its entries are $\mathsf S_{ij}=\sum_k\epsilon_{ijk}A_k$. In 2D, $\mathsf S_{12}=\psi$ and $\mathsf S_{21}=-\psi$.
\par\endgroup
These yield the rotated gradient and curl in \doceqref{eq:lift}, respectively. This cancellation is the coordinate form of $d^2=0$. In 2D a constant streamfunction trace has zero tangential derivative and hence zero normal velocity. In 3D the relative trace $\mathbf n\times\mathbf A=0$ gives a curl with zero normal component. Linear combinations preserve these identities, so the parameterization builds the kinematic constraints into the spectral trial space.

\begingroup
Several potentials generate the same velocity. Gauge fixing chooses one representative so that its mass norm, orthogonalization and Rayleigh objective measure the flow-producing coordinates consistently.
\par\endgroup

\subsection{Gauge representatives and the potential metric}
\label{sec:gauge}
\paragraph{Scalar constants in 2D.}
{Periodic streamfunctions admit the constant freedom $\nabla^\perp(\psi+c)=\nabla^\perp\psi$. Here $c$ is an arbitrary scalar constant. We select the centered representative $\psi^\circ$ by subtracting the area-weighted mean $\langle\psi\rangle_\Omega$. Its quadrature approximation uses the node--weight set $Q=\{(\mathbf x_q,w_q)\}_q$, with node index $q$ and positive weights $w_q$:}
\begin{equation}
\psi^\circ=\psi-\langle\psi\rangle_\Omega,
\qquad \langle\psi\rangle_{\Omega,Q}=\frac{\sum_q w_q\psi(\mathbf x_q)}{\sum_q w_q}.
\label{eq:scalar-gauge}
\end{equation}
{This centering removes the constant null coordinate before normalization, ordering, and Gram--Schmidt, and is differentiated during training. It makes the potential norm measure the spatially varying field: an unconstrained constant could otherwise lower the Rayleigh ratio arbitrarily without changing velocity.}

For the bounded scalar Dirichlet problem, $\psi=0$ on the boundary selects the representative. The multi-hole experiment uses this boundary-pinned mass norm. Scalar-field displays subtract one global spatial mean, preserving derivatives and circulation while giving the plots a common gauge. \docref{Appendix}{app:scalar} details the distinction between the boundary-pinned training metric and the centered display convention.

\paragraph{Curl-free directions in 3D.}
{The vector-potential gauge kernel contains admissible gradients and may also contain relative harmonic fields. We first remove the gradients with a Coulomb projection, as also used in vector-potential fluid editing \citep{chen2026obstacle}. Let $\widetilde{\mathbf A}$ be the raw potential, $\phi$ its scalar gauge correction, $\Delta=\nabla\cdot\nabla$ the Laplacian, and $\mathbf A'$ the gradient-corrected field:}
\begin{equation}
\Delta\phi=\diver\widetilde{\mathbf A},\qquad \phi|_{\partial\Omega}=0,
\qquad \mathbf A'=\widetilde{\mathbf A}-\nabla\phi,
\label{eq:poisson-gauge}
\end{equation}
We then obtain the canonical potential $\mathbf A^\circ$ using the identity operator $I$ and the mass-orthogonal projector $P_{\Harm_R^1}$ onto the relative harmonic potential space $\Harm_R^1$ (the subscript $R$ denotes relative boundary conditions):
\begin{equation}
\mathbf A^\circ=(I-P_{\Harm_R^1})\mathbf A'.
\label{eq:vector-gauge}
\end{equation}
The exact representative is orthogonal to the kernel of $\C$, has zero divergence, and preserves curl and the relative boundary trace. A geometry-dependent solve handles all mode right-hand sides together. For the Metaball Potential objective (the canonical curl quotient), the full relative compatible curl complex selects the mass-orthogonal canonical representative, with physical solve actions in FP64. The Potential-AL and Potential-FG controls use augmented-Lagrangian constraints and finite scalar/harmonic projection, respectively, as detailed in \docref{Appendix}{app:old-Pot}. The removed relative harmonic \emph{potential} fields $\Harm_R^1$ belong to the gauge kernel. The absolute harmonic \emph{velocity} fields $\Harm_A^1$ carry physical circulation and enter the completed velocity basis.

\subsection{Rayleigh learning and the effect of objective choice}
Denote the selected potential representative by $\mu^\dagger$: canonical for the periodic/vector quotient and boundary-pinned for the scalar Dirichlet experiment. {Following variational neural eigenfunction learning \citep{e2018deep,pfau2019spectral,rowan2025rayleigh}, the potential Rayleigh ratio $R_p$ and its Hodge velocity counterpart $R_v$ are}
\begin{equation}
R_p(\mu^\dagger)=\frac{\norm{\C\mu^\dagger}^2}{\norm{\mu^\dagger}^2},
\qquad R_v(\mu^\dagger)=\frac{\norm{d(\C\mu^\dagger)}^2}{\norm{\C\mu^\dagger}^2}.
\label{eq:rayleigh}
\end{equation}
For a streamfunction these are $\norm{\nabla\psi^\dagger}^2/\norm{\psi^\dagger}^2$ and $\norm{\Delta\psi}^2/\norm{\nabla\psi}^2$. For a canonical vector potential they are $\norm{\curl\mathbf A^\circ}^2/\norm{\mathbf A^\circ}^2$ and $\norm{\curl\curl\mathbf A^\circ}^2/\norm{\curl\mathbf A^\circ}^2$. We call our family of geometry-based neural spectral methods \textbf{Neural Eigenfluid Bases (NEB)}. \textbf{Velocity} uses a velocity-space Rayleigh objective, optimizing the energy and mass of the lifted velocity; its network can still parameterize a potential before the differential lift. \textbf{Potential} uses a potential-space Rayleigh objective, with a scalar streamfunction in 2D and a mass-orthogonal canonical vector potential in 3D. Its control variants are \textbf{Potential-FG}, with finite scalar/harmonic gauge projection, and \textbf{Potential-AL}, with augmented-Lagrangian constraints. The potential-space objective uses one fewer spatial derivative. We also explored a vorticity-space objective, but its slower training brought little gain in accuracy in our preliminary tests. It may have advantages in particular settings, though we do not pursue it further here. Our baselines are \textbf{Laplacian}, the numerically solved spectral basis \citep{dewitt2012fluid}, and \textbf{NKB}, the direct-velocity Neural Kinematic Basis \citep{liu2025neural}.

The lower-order objective is linked to the velocity spectrum by the differential lift. Write $\C^*$ for the mass-inner-product adjoint of $\C$, and $\mu_i$ for a normalized potential eigenmode. If $\C^*\C\mu_i=\lambda_i\mu_i$ with $\lambda_i>0$ and $\norm\mu_i=1$, then
\begin{equation}
\Phi_i=\frac{\C\mu_i}{\sqrt{\lambda_i}},\qquad
\C\C^*\Phi_i=\lambda_i\Phi_i,\qquad \norm{\Phi_i}=1.
\label{eq:paired-lift}
\end{equation}
\begingroup
Under compatible relative/absolute domains, $\C\C^*$ is the positive Hodge velocity operator. The classical Hodge pairing gives the same exact nonzero spectrum, including multiplicities \citep{arnold2006finite}. This includes mean-zero periodic fields and decaying free-space fields when integration by parts has no boundary contribution. The equivalence concerns eigenvalues and their lifted eigenspaces; for arbitrary approximate candidates the objectives obey \doceqref{eq:variance}. \docref{Appendix}{app:spectrum} gives the operator proof and its free-space specialization.
\par\endgroup

The objectives still differ away from an exact eigenspace. For a potential expansion $\mu=\sum_j a_j\mu_j$, with spectral coefficients $a_j$ and mode index $j$ in an orthonormal potential eigenbasis,
\begin{equation}
R_p=\frac{\sum_j\lambda_j|a_j|^2}{\sum_j|a_j|^2},\qquad
R_v=\frac{\sum_j\lambda_j^2|a_j|^2}{\sum_j\lambda_j|a_j|^2}\ge R_p.
\label{eq:moments}
\end{equation}
Writing $p_j=|a_j|^2/\sum_\ell|a_\ell|^2$ defines probability weights on the eigenvalues. With $\mathbb E_p$ and $\operatorname{Var}_p$ denoting their weighted mean and variance, the difference is
\begin{equation}
R_v-R_p=\frac{\operatorname{Var}_p(\lambda)}{\mathbb E_p[\lambda]}.
\label{eq:variance}
\end{equation}
Velocity learning places extra weight on high-frequency content in an approximate mode. This changes the optimization of a finite neural representation even though the exact spectra agree. The periodic experiment measures the resulting cost--accuracy tradeoff. \docref{Appendix}{app:spectrum} derives the correspondence and the boundary term relating componentwise gradient energy to curl energy in 3D.

\subsection{Geometry-local ordering and causal orthogonalization}
Multiple independent outputs must span distinct low-frequency directions. For quadrature points $(\mathbf x_q,w_q)$ with $w_q>0$, define
\begin{equation}
\ip{a}{b}_Q=\sum_q w_q a(\mathbf x_q)\cdot b(\mathbf x_q).
\end{equation}
At each sampled geometry we select representatives, normalize them in the chosen metric, and sort detached preliminary Rayleigh values. The geometry-local permutation $\pi$ adapts the ordering to that domain. Let $\bar\mu_i$ be preliminary normalized representatives, $\nu_i$ their orthogonalization residuals, and $\widehat\mu_i$ the resulting unit-norm candidates, where $\|a\|_Q^2=\langle a,a\rangle_Q$. The operator $\sg$ stops differentiation through its argument. The training loss $\mathcal L(\theta)$ averages the quadrature Rayleigh ratio $R_{p,Q}$ over the sampled geometry distribution, denoted by $\mathbb E_{\mathbf g}$. For the potential objective, causal orthogonalization follows Shape Space Spectra \citep{chang2025shapespacespectra} and treats lower-ranked modes as fixed when learning the next mode:
\begin{align}
\nu_i&=\bar\mu_{\pi(i)}-\sum_{j<i}\ip{\bar\mu_{\pi(i)}}{\sg(\widehat\mu_j)}_Q\sg(\widehat\mu_j),\\
\widehat\mu_i&=\nu_i/\norm{\nu_i}_Q,
\qquad \mathcal L(\theta)=\mathbb E_{\mathbf g}\sum_{i=1}^K R_{p,Q}(\widehat\mu_i).
\label{eq:training}
\end{align}
A second orthogonalization pass reduces accumulated projection error. The causal gradient prevents a higher mode from changing a lower mode merely to simplify its own residual. Because gauge selection and the boundary constraints define linear spaces, the subsequent linear combinations preserve them.

The velocity objective uses the velocity inner product and corresponding energy, applying the same coefficient transform to every derivative channel. Spatial derivatives and parameter adjoints use an analytic CuPy implementation. At repeated eigenvalues, the identifiable object is the eigenspace; the mode atlas is interpreted up to signs and orthogonal rotations within that space.

\begin{figure}[!t]
\centering\includegraphics[width=0.7\linewidth,height=.72\textheight,keepaspectratio]{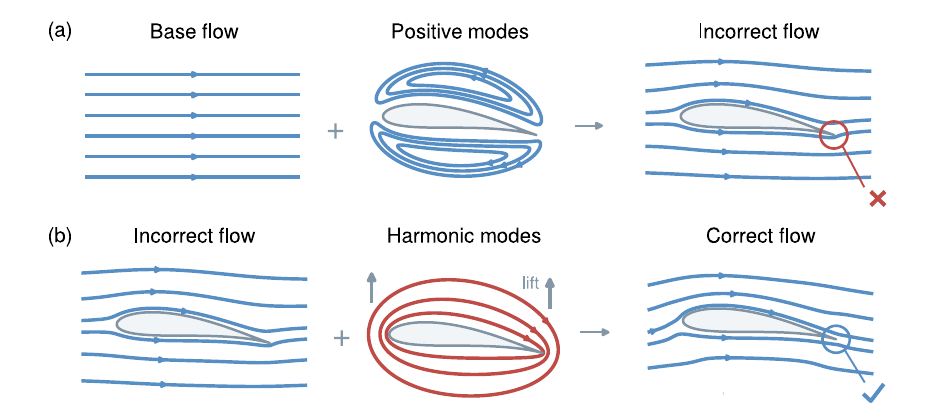}
\caption{\textbf{Harmonic circulation completes the flow representation.} (a) A base flow combined with positive modes gives an incomplete, zero-circulation description of a lifting airfoil, with an incorrect sharp-trailing-edge flow. (b) Adding harmonic circulation permits smooth outflow satisfying the Kutta condition and generates lift in the presence of the base flow. Streamlines schematically illustrate the decomposition in \doceqref{eq:completion}.}
\label{fig:harmonic-airfoil}
\end{figure}

\subsection{Harmonic completion at fixed capacity}
\label{sec:rom}
{Harmonic zero-frequency fields complete the spectral ROM basis in \citet{liu2015model}. We construct their neural counterparts with a dedicated boundary-value objective. Zero-trace positive streamfunctions describe the vortical sector; independent curl-free circulation is supplied by differences between boundary constants. Let $h$ be the number of holes, $\Gamma_0$ the outer boundary, and $\Gamma_m$ the boundary of obstacle $m=1,\ldots,h$. For obstacle $\ell$, its harmonic measure $\chi_\ell$ solves the following problem, where $\delta_{\ell m}$ is the Kronecker delta:}
\begin{equation}
\Delta\chi_\ell=0,\qquad
\chi_\ell|_{\Gamma_m}=\delta_{\ell m},\qquad
\chi_\ell|_{\Gamma_0}=0,
\label{eq:harmonic-measure}
\end{equation}
The velocities $h_\ell=\nabla^\perp\chi_\ell$ are divergence-free, curl-free, and tangent to every boundary. Their span has one independent coordinate per hole. Cardinal boundary values are convenient coordinates; a further invertible circulation matrix can convert them into unit-period coordinates.

{Let $L_\ell$ be a fixed lift of the cardinal boundary traces and $n_{\theta,\ell}$ a neural correction. The parameterization $\chi_\ell=L_\ell+e_{\mathbf g}n_{\theta,\ell}$ gives the correction zero trace. The harmonic training loss $\mathcal L_H$ therefore minimizes Dirichlet energy within this prescribed boundary-value space:}
\begin{equation}
\mathcal L_H(\theta;\mathbf g)=\sum_{\ell=1}^{h}\int_{\Omega_{\mathbf g}}|\nabla\chi_\ell|^2\,\dd\mathbf x,
\label{eq:harmonic-loss}
\end{equation}
{Fixed traces determine the scale of a harmonic measure, so this numerator-only objective directly realizes the Dirichlet principle through neural energy minimization \citep{e2018deep}. \docref{Appendix}{app:harmonic-loss} gives the variational derivation. Raw and scaled lifts parameterize the same boundary-value problem and correction space with different optimization coordinates. Holding the positive checkpoint fixed while changing the lift isolates this conditioning effect.}

With $h$ holes and total dimension $K$, the completed representation contains $K-h$ positive modes and $h$ harmonic modes. Let $\mathbf u_{\rm base}$ supply any prescribed base flow, and $a_i,b_\ell$ be the positive and harmonic coefficients:
\begin{equation}
\mathbf u=\mathbf u_{\rm base}+\sum_{i=1}^{K-h}a_i\Phi_i+\sum_{\ell=1}^{h}b_\ell h_\ell.
\label{eq:completion}
\end{equation}
The prescribed base flow carries nonhomogeneous data when present. Writing $P_K=\operatorname{span}\{\Phi_1,\ldots,\Phi_K\}$ and $H_h=\operatorname{span}\{h_1,\ldots,h_h\}$, the comparison $P_K$ versus $P_{K-h}+H_h$ changes the allocation of coordinates while keeping their total number fixed. We abbreviate these allocations as P and P+H (or PH). They test topological completeness at the same capacity. \docref{Appendix}{app:topology-rom} illustrates this distinction on an annulus and gives the circulation normalization.

\begingroup
Circulation is particularly important in lifting-airfoil flow. In the ideal slip-flow illustration of \docref{Figure}{fig:harmonic-airfoil}, the Kutta condition selects a circulation $\Gamma$ and the classical Kutta--Joukowski lift per unit span has magnitude $\rho U_\infty|\Gamma|$, where $\rho$ is fluid density and $U_\infty$ the freestream speed \citep{batchelor2000fluid}. The harmonic coordinate permits that circulation to be represented alongside the base flow and positive modes. This motivates including the sector whenever the target physics requires it; the multi-hole experiments quantify its allocation at a fixed coordinate budget.
\par\endgroup

\subsection{From the learned span to autonomous dynamics}
{Write the reduced velocity as $u=u_{\rm base}+Uc$, where column $U_i$ is basis field $i$ and $c=(c_1,\ldots,c_K)^{\top}$ is the coefficient state. Following spectral fluid Galerkin reduction \citep{dewitt2012fluid,liu2015model}, let $t$ denote simulation time and $\dot c=\dd c/\dd t$. The full velocity mass matrix $M$ defines the coefficient metric, the vorticity/curl stiffness matrix $S$ supplies viscous dissipation, $\nu$ is kinematic viscosity, and $Q(c,c)$ is the quadratic advective load. Projection of the homogeneous unforced Hodge model gives}
\begin{equation}
M\dot c=Q(c,c)-\nu S c,\qquad
M_{ij}=\ip{U_i}{U_j},\quad S_{ij}=\ip{dU_i^\flat}{dU_j^\flat},
\label{eq:reduced-dynamics}
\end{equation}
\begingroup
For indices $i,j,k=1,\ldots,K$, advection gives the quadratic load $[Q(c,c)]_i=\sum_{j,k}Q_{ijk}c_jc_k$, with $Q_{ijk}=-\langle U_i,(U_j\cdot\nabla)U_k\rangle$.
\par\endgroup
{A base flow or force contributes the corresponding constant and linear terms. Exact orthonormal eigenfields make $M$ the identity and $S$ diagonal. For the learned span, retaining the full operators incorporates its measured metric and energy couplings into the dynamics.}

For initial velocity $u_0$, the initial coefficients solve $M c(0)=\ip{U}{u_0-u_{\rm base}(0)}$, where the vector inner product has entry $\langle U_i,u_0-u_{\rm base}(0)\rangle$. This initializes an autonomous coefficient trajectory. Numerical reference states enter the subsequent error measurement. Learned and numerical bases use the same stated discrete generator within each matched comparison, including the normal-penalty grid generator in 3D.

\paragraph{Offline assembly and coefficient-only integration.}
For a fixed span and a quadratic discrete fluid generator, {preassembled reduced operators \citep{treuille2006model,dewitt2012fluid,liu2015model} move full-field evaluation offline. Here $\mathcal L$ denotes the specified linear fluid generator (viscosity and any linear penalty or base-flow terms), while $\mathcal N$ denotes its bilinear advective part. Writing the generator as $\mathcal L u+\mathcal N(u,u)$ gives the reduced linear matrix $A$ and quadratic tensor $\mathcal Q$ below. The subscript $h$ denotes the discrete spatial quadrature:}
\begin{equation}
M\dot c=A c+\mathcal Q(c,c),\qquad
A_{ij}=\ip{U_i}{\mathcal L U_j}_h,\quad
\mathcal Q_{ijk}=\ip{U_i}{\mathcal N(U_j,U_k)}_h.
\label{eq:offline-rom}
\end{equation}

The same quadrature defines $M$ and the two assembled operators. During inference, each time-integration stage contracts these fixed operators with the current coefficient vector and solves the full mass system using its cached Cholesky factor. The recurrent state is therefore the reduced vector $c(t)$ throughout the trajectory. Spatial reconstruction $u(t)=u_{\rm base}(t)+Uc(t)$ is performed only at requested output times, as illustrated in \docref{Figure}{fig:method}. This is an algebraic reorganization of the chosen Galerkin system. \docref{Appendix}{app:coefficient-acceleration} measures operator equivalence, integration time, offline cost, and flow error separately.

\section{Experiments}
\label{sec:experiments}
We evaluate spectral objectives, harmonic completion, and reduced dynamics through the controlled comparisons in \docref{Table}{tab:experiment-map}.
\begin{table}[!htb]
\centering\small
\caption{{Validation goals and controlled comparisons.}}
\label{tab:experiment-map}
\setlength{\tabcolsep}{5pt}
\renewcommand{\arraystretch}{1.13}
\begin{tabular}{@{}>{\raggedright\arraybackslash}p{.14\linewidth}>{\raggedright\arraybackslash}p{.36\linewidth}>{\raggedright\arraybackslash}p{\dimexpr.50\linewidth-4\tabcolsep\relax}@{}}\toprule
Case & Comparison & {Validation}\\\midrule
Periodic & Learned spectral and kinematic bases; same-capacity Laplacian & {Neural eigenfluid and potential-space feasibility; autonomous accuracy and decay.}\\
Multi-hole & Random small-scale fields and large-scale vortices; fixed-budget allocation and lift conditioning & {Topology-aware circulation completion, energy capture and harmonic conditioning.}\\
Metaball & Velocity-basis capacity; full-grid and assembled coefficient operators & {3D feasibility, capacity-dependent representation and coefficient-space acceleration.}\\\bottomrule
\end{tabular}
\end{table}

{The field comparisons show the numerical ground-truth flow reference (Numerical GT), a same-capacity Laplacian basis, a potential-trained basis, and a velocity-trained basis. Let $u(t)$ be a reduced velocity and $u_{\rm GT}(t)$ the numerical reference at time $t$. Annotations give relative velocity error $e_u(t)=\|u(t)-u_{\rm GT}(t)\|_{L^2}/\|u_{\rm GT}(t)\|_{L^2}$, where $L^2$ is the spatial mass norm over the full assessment domain and plotted percentages are $100e_u$; we also write $e(t)=e_u(t)$. A subscript GT denotes a reference quantity; shared color scales preserve amplitude differences. \docref{Appendix}{app:protocol} gives discretizations and convergence checks. \docref{Figure}{fig:multihole-modes} compares the multiply connected modes, complemented by the periodic and three-dimensional atlases in \docref{Appendix}{app:modes}.}

A fixed-case adaptation of Neural Kinematic Bases {\citep{liu2025neural}} provides an additional direct-velocity comparison under the stated training budget and common Galerkin evolution (\docref{Appendix}{app:nkb}). Experiments are conducted on a workstation with an Intel Core i9-14900K CPU and an NVIDIA GeForce RTX 4080 SUPER GPU.

\subsection{Case 1: Periodic vortex merger}
\label{sec:case1}
The periodic domain admits an analytic Fourier basis and a mean-zero scalar gauge, providing a controlled setting for autonomous flow with learned spatial coordinates. The vortex pair has separation $d=1.35$, with $K=1024$ modes and a final simulation time $T=18\pi$. All methods project the same initial velocity with the full mass matrix and retain every basis coordinate during evolution. {The Laplacian (Fourier) basis supplies the same-capacity numerical control, following Laplacian Eigenfluids \citep{dewitt2012fluid}.}

\paragraph{Autonomous evolution.}
\docref{Figure}{fig:case1-field} shows the potential-trained seed-0 trajectory through $t=0,9\pi,18\pi$. \docref{Figure}{fig:case1-wide} measures velocity and vorticity errors, kinetic-energy decay, and enstrophy decay on the same case. The spatial basis is reused throughout the trajectory.
\begingroup
Both learned spectral objectives sustain sub-percent velocity error. Velocity learning gives the higher accuracy, while potential learning reduces training cost, matching the cost--accuracy tradeoff predicted by the objective analysis (\docref{Figure}{fig:case1-wide}; \docref{Table}{tab:runtime-training}). The potential and NKB trajectories have nearly identical synchronized integration times, yet the spectral span gives substantially smaller flow error (\docref{Table}{tab:runtime-2d-online}). Spectral organization therefore improves autonomous accuracy at comparable recurrent cost under this matched protocol. The absolute mean squared velocity error is reported in \docref{Table}{tab:case1}.
\par\endgroup
\begin{table}[!htb]\centering\small
\caption{Periodic vortex merger, $d=1.35$, $T=18\pi$, $K=1024$. Velocity mean squared error (MSE) averages squared velocity-component errors over all 241 saved $128^2$ observation-grid states. Learned entries then average three seeds. Relative errors are visualized in \docref{Figure}{fig:case1-wide}; spectral training costs appear in \docref{Table}{tab:runtime-training}, and NKB costs in \docref{Table}{tab:runtime-nkb-training}. NKB uses the 2000-update adaptation in \docref{Appendix}{app:nkb}.}
\label{tab:case1}
\begin{tabular}{lr}\toprule
Basis & Velocity MSE\\\midrule
Laplacian & $9.737\times10^{-10}$\\
Velocity & $8.882\times10^{-9}$\\
Potential & $1.619\times10^{-7}$\\
NKB & $2.823\times10^{-3}$\\
\bottomrule\end{tabular}\end{table}
\begin{figure}[!htb]
\centering\includegraphics[width=\linewidth,height=.72\textheight,keepaspectratio]{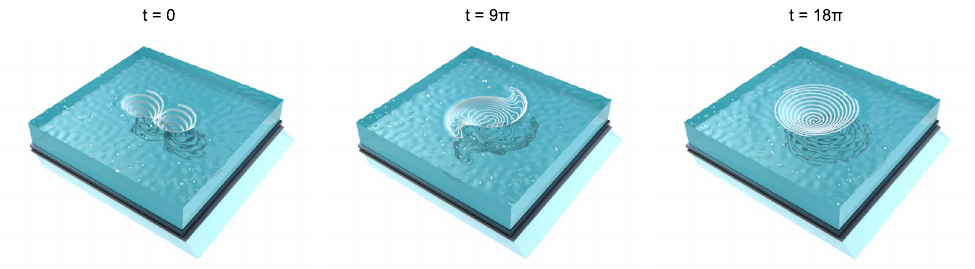}
\caption{\textbf{Potential: periodic vortex merger.} Panels show the $K=1024$, seed-0 Potential trajectory with vortex separation $d=1.35$, from left to right at $t=0,9\pi,18\pi$. Surface relief encodes smoothed vorticity magnitude on a common scale; white accents trace vortex structure. Height is a display coordinate for the two-dimensional flow.}
\label{fig:case1-field}
\end{figure}
\begin{figure}[!htb]
\centering\includegraphics[width=\linewidth,height=.72\textheight,keepaspectratio]{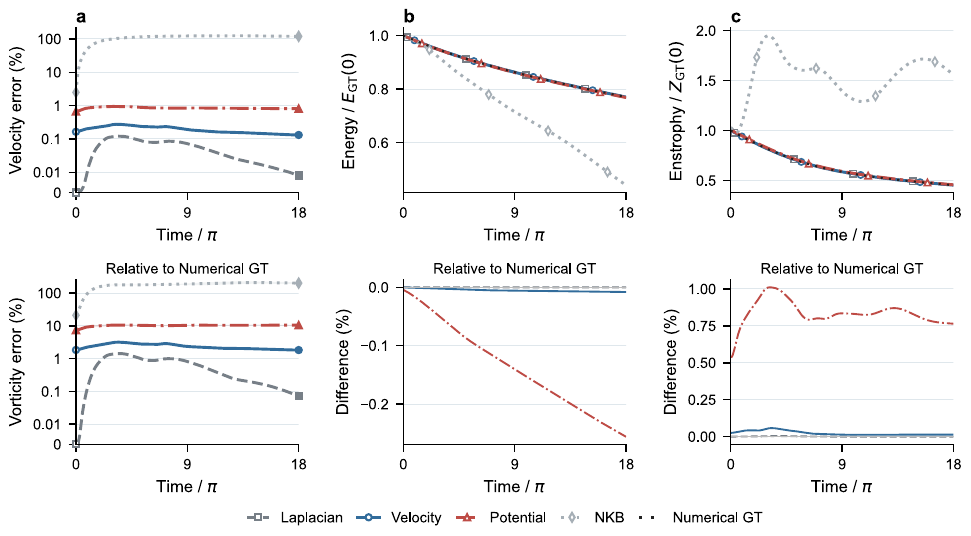}
\caption{\textbf{Autonomous accuracy and decay in vortex merger.} All panels use $K=1024$, $d=1.35$, $T=18\pi$. (a) Relative velocity (top) and vorticity (bottom) errors against Numerical GT, normalized by the corresponding time-dependent observation-grid $L^2$ norm. Learned curves show three-seed means with one sample standard deviation (SD), on a symmetric-log scale with a linear interval below $0.01\%$. Open and filled markers indicate initial and final states. (b) Kinetic energy, $E=\frac12\langle|u|^2\rangle$. (c) Enstrophy, $Z=\frac12\langle\omega^2\rangle$, where $\omega=\partial_xu_y-\partial_yu_x$ is scalar vorticity, $u_x,u_y$ are velocity components, and $\langle\cdot\rangle$ denotes the spatial mean. {All five methods appear in the upper decay panels, using seed-0 learned trajectories and common initial GT normalization. Distinct line styles and sampled markers distinguish nearly overlapping curves. The lower panels of (b,c) show $100(E-E_{\rm GT})/E_{\rm GT}$ and $100(Z-Z_{\rm GT})/Z_{\rm GT}$ for Laplacian, velocity and potential bases.} Curves use all 241 saved $128^2$ states from $256^2$ operator quadrature and maximum step $0.005$.}
\label{fig:case1-wide}
\end{figure}

\subsection[Case 2: Multiply connected flow: energy compression and circulation]{Case 2: Multiply connected flow---energy compression and circulation}
\label{sec:case2}
Each obstacle introduces an independent circulation coordinate. We consider two random flow ensembles: streamfunctions with many small-scale vortical structures, and combinations of a few smooth, large-scale vortices. Together they test energy compression, harmonic completion, and lift conditioning. The small-scale random-field comparison uses total $K=1024$ throughout: $P_K$ and $P_{K-h}+H_h$ exchange positive and harmonic coordinates within the same state budget.

\paragraph{Energy captured per reduced coordinate.}
\begingroup
We sample streamfunctions in physical space from Gaussian lattice values and a quintic spline, then condition each draw to impermeability and zero boundary vorticity before analytic differentiation. Boundary constants are derived from the same draw, so its vortical and circulation content arise together. Frozen seeds $[20260928,\text{topology index},\text{field index}]$ define 20 reproducible fields per topology; indices 0 and 1 supply the six-field all-method dynamic comparison. \docref{Appendix}{app:random-streamfunction} specifies the generator, conditioning, and reference construction.
\par\endgroup

\begingroup
Let $W$ be the positive assessment quadrature matrix, with $\|v\|_W^2=v^{\top}Wv$ for sampled velocity $v$, and let $u_K=Uc_0$ be the full-mass projection of $u_0$, with initial coefficient vector $c_0=c(0)$. The captured kinetic-energy fraction is $\eta=\|u_K\|_W^2/\|u_0\|_W^2$. Potential bases capture more initial kinetic energy than equal-capacity Laplacian controls across all three topologies, with the largest gains in the two- and three-hole domains, while compressing the reference state by roughly three orders of magnitude (\docref{Figure}{fig:random-k1024-energy}). Most energy in these small-scale random fields lies in the vortical sector, represented by the spectrally ordered $P$ modes. Harmonic coordinates lift independent obstacle-boundary constants into circulation-carrying velocities \citep{arnold2006finite,liu2015model}. At fixed $K$, their added harmonic energy competes with the energy of the $h$ displaced positive modes. This energy allocation favors positive coordinates for the small-scale ensemble, while harmonic coordinates provide its independent circulation sector (\docref{Appendix}{app:random-streamfunction}).
\par\endgroup

\begin{figure}[!htb]\centering
\includegraphics[width=\linewidth,height=.72\textheight,keepaspectratio]{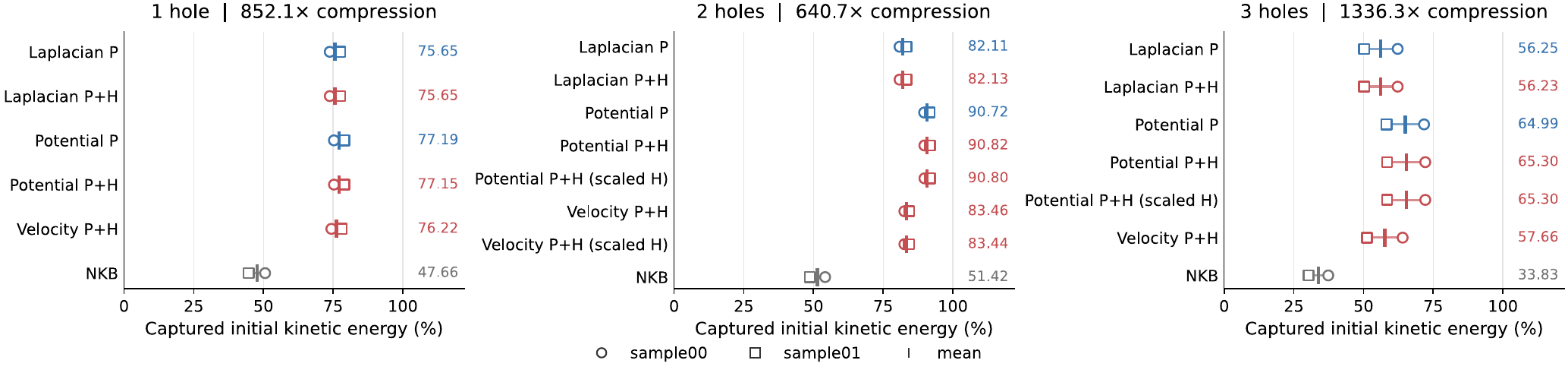}
\caption{Initial energy captured by the original $K=1024$ representations. Circles and squares denote the two predetermined physical fields in each topology; vertical ticks denote their mean. The annotations give mean captured-energy percentages and state-coordinate compression $N_{\rm ref}/K$, where $N_{\rm ref}$ is the number of independent coordinates of the reference state. The complete set of learned branches, raw/scaled harmonic coordinates, numerical controls, and NKBs use the same analytic initial fields.}
\label{fig:random-k1024-energy}
\end{figure}

\paragraph{Autonomous dynamics at the same capacity.}
\begingroup
Each representation evolves from its analytic-field projection using the full Galerkin operators, $\nu=1/750$, and $T=1$. The Potential P span combines higher initial energy capture with lower final trajectory error than the same-capacity Laplacian P control in every topology. \docref{Table}{tab:completion-improvement} quantifies the paired effects of coordinate allocation on energy capture and trajectory accuracy. The nonlinear evolution couples these coordinates within the retained span. With $e_P,e_{PH}$ and $\eta_P,\eta_{PH}$ denoting the errors and energy fractions under the two allocations, its gains are $\Delta e=e_P-e_{PH}$ and $\Delta\eta=\eta_{PH}-\eta_P$, with positive values denoting improvement. \docref{Appendix}{app:random-k1024-six} gives the full method--field curves, reference checks, energy histories, and circulation diagnostics; \docref{Appendix}{app:random-streamfunction} provides the broader numerical capacity and low-dimensional allocation studies.
\par\endgroup
\begingroup
\begin{table}[!htb]\centering\small
\caption{Paired changes under harmonic allocation on the six qualified random fields at total $K=1024$. Each entry averages the two fixed fields in its topology. Initial capture $\eta$ uses the analytic-target full-mass projection; final velocity error uses the qualified numerical trajectory. Gains are $\Delta\eta=\eta_{PH}-\eta_P$ and $\Delta e(1)=e_P(1)-e_{PH}(1)$ in percentage points, so positive values favor completion. Absolute capture and trajectories are shown in \docref{Figures}{fig:random-k1024-energy} and~\ref{fig:random-k1024-errors}.}
\label{tab:completion-improvement}
\setlength{\tabcolsep}{4pt}
\begin{tabular}{rlrr}\toprule
Holes & Basis & $\Delta\eta$ & $\Delta e(1)$\\\midrule
1 & Laplacian & +0.004 & +0.020\\
1 & Potential & -0.045 & +0.092\\
2 & Laplacian & +0.018 & +0.087\\
2 & Potential & +0.102 & -0.243\\
3 & Laplacian & -0.014 & +0.021\\
3 & Potential & +0.305 & +0.942\\
\bottomrule\end{tabular}\end{table}
\par\endgroup

\paragraph{Completion and conditioning in random large-vortex flows.}
\begingroup
We randomly construct flows from a few smooth, large-scale vortices, with random locations, signs, and strengths. These flows combine local vorticity with motion around the obstacles. At $K=32$, harmonic completion consistently lowers final error across all three obstacle configurations and brings the learned completed bases close to the numerical completed controls (\docref{Table}{tab:prescribed32-summary}). The boundary lifts make the independent circulation component accessible within a fixed-size reduced state. The two ensembles thus expose complementary demands on the basis: small-scale random fields emphasize vortical energy capture, while random large-vortex flows make topological completion and lift conditioning visible.
\par\endgroup

{At $K=1024$, raw/scaled pairs share the positive checkpoint, boundary traces, Dirichlet objective, initialization, and training budget. Scaling substantially lowers the multi-hole trajectory errors (\docref{Figure}{fig:case2-ablation}), whereas changing the spectral objective in the raw two-hole control has little effect. The paired comparison identifies harmonic-lift conditioning as the main source of this improvement. \docref{Table}{tab:case2} gives the corresponding positive-only controls.}

\begin{table}[!htb]
\centering\small
\caption{Positive-only controls for the random large-vortex ensemble: final relative velocity error (\%), total $K=1024$, $T=1$, averaged over 20 initial states. Learned positive modes use potential training, seed 1. Completed bases and raw/scaled lift comparisons are visualized in \docref{Figure}{fig:case2-ablation}.}
\label{tab:case2}
\begin{tabular}{rrr}\toprule
Holes & Numerical $P$ & Learned $P$\\\midrule
1 &32.542&32.124\\
2 &34.197&33.963\\
3 &36.685&36.707\\\bottomrule
\end{tabular}
\end{table}
{Once the harmonic fields are consistently conditioned, potential and velocity learning produce similar trajectory errors. \docref{Figure}{fig:case2-field} shows the resulting flow, and \docref{Appendix}{app:fields} gives the matched-objective and conditioning comparisons.}

\begin{figure}[!htb]
\centering\includegraphics[width=\linewidth,height=.72\textheight,keepaspectratio]{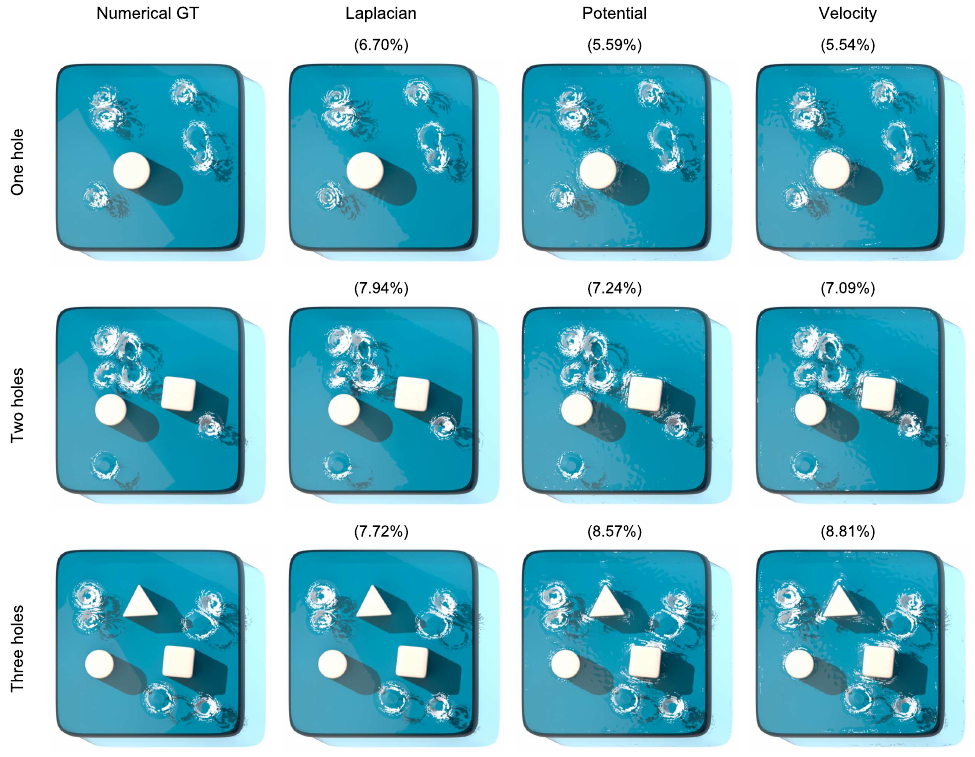}
\caption{\textbf{Random large-vortex flow at $K=1024$.} Rows: one-, two-, and three-hole domains. Columns: numerical reference, Laplacian, potential- and velocity-trained bases at $t=1$, seed 1, state 00. Reduced bases have $1024-h$ positive and $h$ harmonic coordinates; the learned objectives share raw $H_1$ and scaled $H_2,H_3$. Errors use the common level-2 reference. Overhead views share height and foam mappings within each geometry (\docref{Appendix}{app:fields}).}
\label{fig:case2-field}
\end{figure}

\begingroup
\paragraph{Topology-dependent spectral modes.}
\docref{Figure}{fig:multihole-modes} compares the first eight positive modes of the Laplacian and both learned constructions across three obstacle configurations. The aligned panels show each basis's spatial structure and ordering; independent harmonic coordinates complete the span.
\par\endgroup

\begin{figure}[!htb]
\centering\includegraphics[width=\linewidth,height=.72\textheight,keepaspectratio]{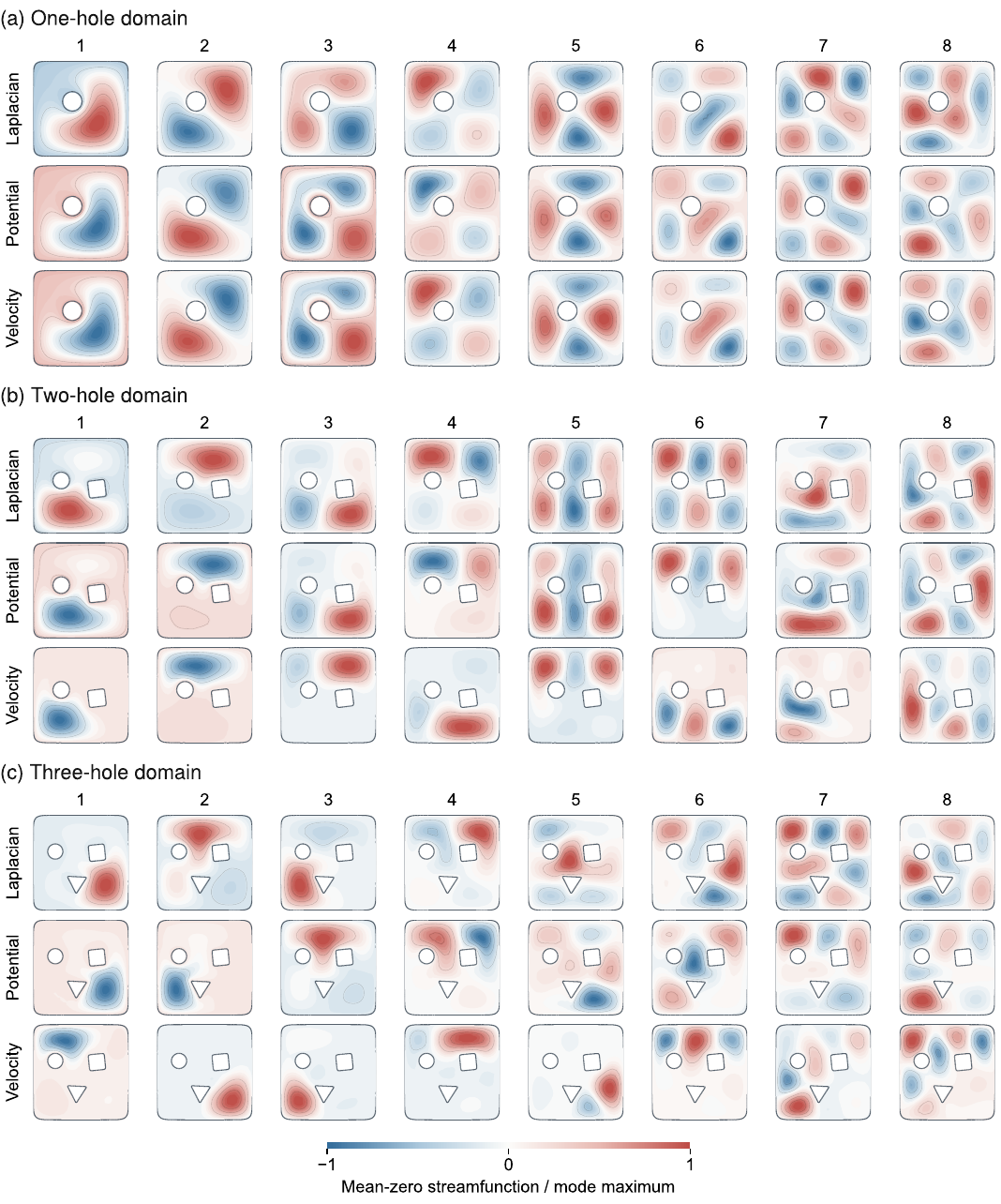}
\caption{\textbf{Positive spectral modes across obstacle topologies.} Panels (a--c) show one-, two-, and three-hole domains. Within each panel, rows compare Laplacian, potential-trained, and velocity-trained fields; columns show modes 1--8. Learned fields follow the positive-mode selection order in the total-$32$ continuous piecewise-linear finite-element (P1) basis, whose $h$ harmonic coordinates are separate. The shared color scale shows area-centered streamfunction divided by its per-mode maximum magnitude. Signs and rotations within repeated eigenspaces preserve the span.}
\label{fig:multihole-modes}\end{figure}

\subsection[Case 3: 3D Metaball: representation and autonomous evolution]{Case 3: 3D Metaball---representation and autonomous evolution}
\label{sec:case3}
The three-dimensional domain tests how a learned span supports autonomous dynamics. We separate the error due to spatial representation from the error in coefficient evolution, then measure the cost of evaluating the same reduced generator entirely in coefficient space.

\paragraph{Autonomous flow at 1024 modes.}
{\docref{Figure}{fig:case3-field} compares a velocity-trained $K=1024$ basis with same-capacity Laplacian and NKB bases against the common independent $128^3$ normal-penalty Navier--Stokes (NS) reference. All three evolve autonomously from full-mass projections of the same initial field. Eigenfluid and Laplacian begin with comparable reconstruction quality and develop similar relative errors during autonomous evolution. The vector-error panels in \docref{Figure}{fig:case3-field} also highlight similar regions of elevated error. Their agreement in error magnitude and spatial structure identifies the shared role of truncated Galerkin dynamics. Both methods compress the flow into the same number of coordinates, with nonlinear interactions resolved within that span. This correspondence connects the learned continuous basis to the established spectral ROM behavior at the same compression level. \docref{Figure}{fig:case3-capacity}(b) examines the joint contributions of spatial capacity and coefficient dynamics.}
\begin{figure}[!htb]
\centering\includegraphics[width=\linewidth,height=.72\textheight,keepaspectratio]{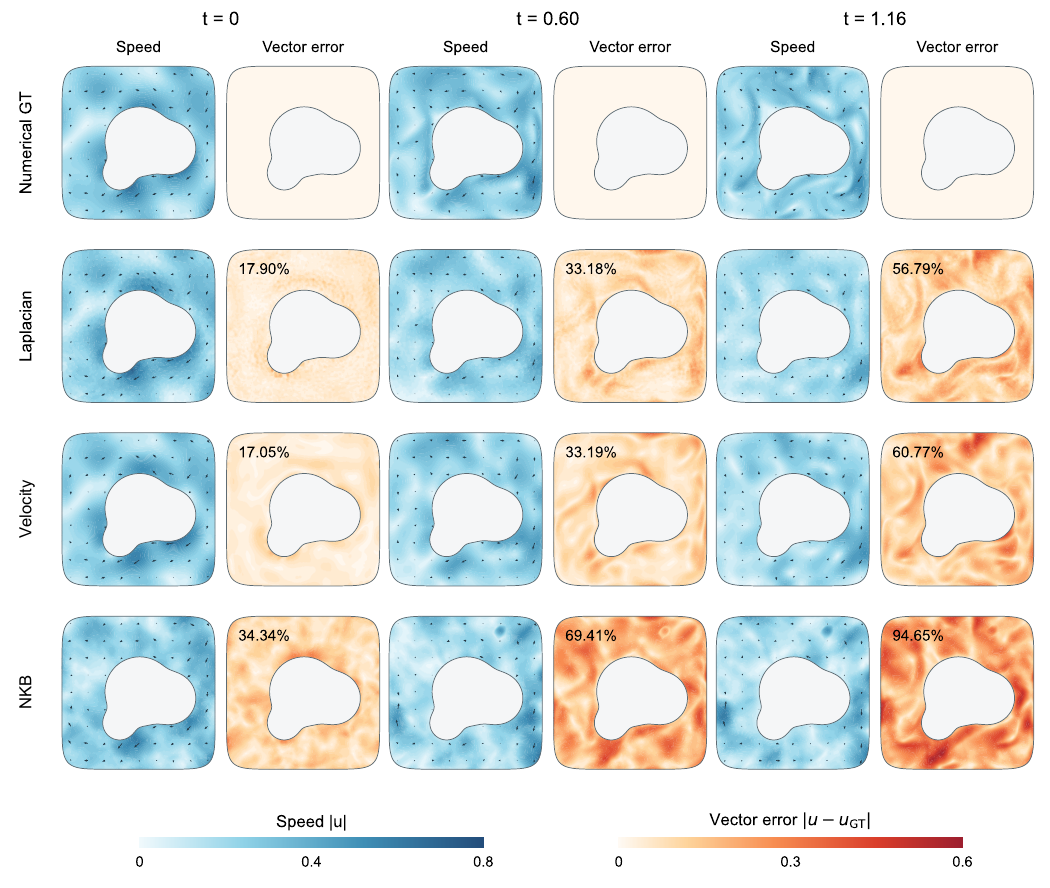}
\caption{\textbf{Three-dimensional autonomous flow at $K=1024$.} Rows show Numerical GT, Laplacian, Velocity, and NKB; columns show $t=0,0.60,1.16$. Each time pairs speed $|u|$ (blue) with pointwise vector error $|u-u_{\rm GT}|$ (orange--red) on the offset $z=0.31640625$ slice, which exposes the fused Metaball silhouette. Filled contours interpolate the saved nodal values and are clipped to the original geometry's smooth boundary. Speed, error, and in-plane arrow lengths use separate fixed scales across all methods and times. Percentages in the upper-left corner of the vector-error panels denote full-volume relative velocity error, evaluated on all $1.6\,\mathrm{M}$ fluid nodes. Here $\Delta t$ is the integration time step; ROM trajectories use $\Delta t=0.000625$ and the common NS reference uses $\Delta t=0.0025$.}
\label{fig:case3-field}
\end{figure}

\paragraph{Spatial capacity.}
\begingroup
Increasing capacity improves both initial reconstruction and autonomous accuracy under the same training budget, full-volume reference, and assessment times (\docref{Figure}{fig:case3-capacity}(b)). At $K=1024$, the learned velocity span has lower initial projection error than the Laplacian control. \docref{Table}{tab:finite-potential-initial} gives the original-field projection errors, and \docref{Appendix}{app:Pot-cost} reports training costs.
\par\endgroup

\paragraph{Representation of randomized 3D large-eddy fields.}
Figure~\ref{fig:case3-Pot-capacity}(a) compares the frozen Metaball spans on 20 independently randomized large-eddy initial states. A case-local Philox stream draws 96 independent standard-Gaussian sine and cosine coefficients for a three-component vector potential over 16 representatives of the integer wavevector pairs $\pm\boldsymbol{k}$ with $1\leq\|\boldsymbol{k}\|\leq2$. Coefficients are scaled by $1/\|\boldsymbol{k}\|$; the raw potential wavelengths range from $0.5$ to $1$ unit-domain length. The analytic curl follows the common $128^3$ Metaball incompressibility projection and boundary conditioning, RMS normalization to $0.32$, and four warm-up substeps. Each frozen basis projects the same states using the full physical mass matrix, all $1.6\,\mathrm{M}$ fluid nodes, and all $K$ retained coordinates. Errors decrease with capacity. At $K=512$, Potential has lower error than Laplacian in 15 of the same 20 states. This paired assessment measures representation of independently sampled large eddies with each training checkpoint held fixed; the autonomous-flow study evaluates evolution on its trajectory initial condition. Assessment states are separate from basis training and selection, and NKB uses geometry alone. The $K=1024$ results for Velocity, Laplacian, and NKB are reported separately in Appendix~\ref{app:randomfield-values}. The modest reconstruction gain of Potential over Laplacian comes at a substantial training cost; the cost--accuracy comparison in Appendix~\ref{app:randomfield-values} favors Velocity or Laplacian at larger capacities in this case. Appendix~\ref{app:old-Pot} documents the Potential-FG and Potential-AL configurations; Appendix~\ref{app:Pot-cost} reports preparation costs.
\begin{figure}[!htb]
\centering\includegraphics[width=\linewidth]{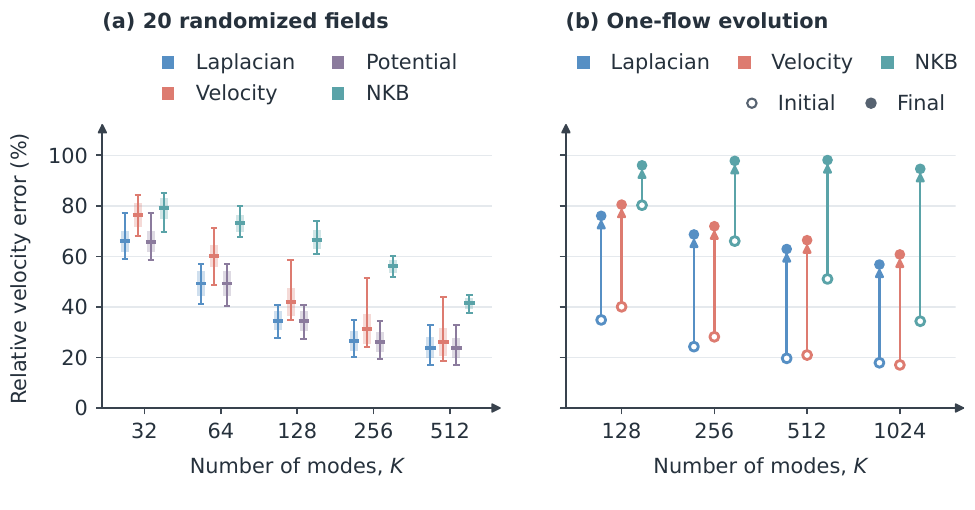}
\caption{\textbf{Spatial representation and autonomous evolution in the three-dimensional Metaball domain.} \textbf{(a)} Frozen-basis reconstruction on 20 independently randomized large-eddy initial fields, with 400 paired measurements across four methods and $K=32,64,128,256,512$. Shaded rectangles show mean $\pm$ sample standard deviation ($n=20$, $\mathrm{ddof}=1$); whiskers show minima and maxima. The $K=1024$ randomized-field results are in Appendix~\ref{app:randomfield-values}. \textbf{(b)} A separate, fixed initial field is used for the autonomous-flow comparison of Laplacian, Velocity, and NKB. Initial full-mass projection error (open circles, $t=0$) and final autonomous error (filled circles, $T=1.16$) at $K=128,256,512,1024$, with arrows connecting the two times. The three methods are offset horizontally in this left-to-right order within each capacity. All capacities use the same method colors. Both panels report relative velocity $L^2$ errors in percent using all retained modes and all $1.6\,\mathrm{M}$ fluid nodes. Autonomous trajectories store 30 states, with ROM step $0.000625$ and common NS reference step $0.0025$; training costs appear in Tables~\ref{tab:Pot-training-basis-cost} and~\ref{tab:runtime-nkb-training} and the error decomposition in Appendix~\ref{app:3d-error}.}
\label{fig:case3-representation-evolution}
\label{fig:case3-Pot-capacity}
\label{fig:case3-capacity}
\end{figure}

\paragraph{Separating the two error sources.}
{Projection onto each frozen span separates autonomous error into orthogonal representation and in-span coefficient terms (\doceqref{eq:error-split}). The extended decomposition shows the same pattern through 1024 modes: a richer span lowers the best-projection error, while the coefficient term remains nearly unchanged at the upper capacities (\docref{Table}{tab:3d-error-attribution}). Spatial refinement and coefficient evolution thus control complementary parts of the reduced model. The decomposition makes the representation gain explicit and connects the common evolution pattern to the retained-coordinate dynamics.}

\paragraph{Exact coefficient-space execution.}
\begingroup
Preassembling the full quadratic Galerkin operator moves spatial evaluation out of recurrent time stepping. The resulting integration is faster than the matched-step NS control at both tested capacities, with assembled and direct reduced right-hand sides agreeing to numerical precision. This acceleration comes from executing the same reduced generator in coefficient space. \docref{Figure}{fig:coefficient-capacity} shows the associated flow errors and timings; \docref{Appendix}{app:runtime} separates one-time preparation from online costs. Repeated initial-condition queries reuse that assembly, motivating the \hyperref[sec:inference-cost]{amortization analysis below}.
\par\endgroup

\paragraph{Continuous reconstruction across query domains.}
{The same frozen neural map can be evaluated at held-out points and on three unseen domains. \docref{Appendix}{app:unseen-geometry} reports these reconstruction measurements together with target-specific numerical controls. This assessment complements the capacity study by examining how the continuous coordinates respond to changes in query location and geometry.}

\subsection{Continuous basis queries and amortized inference}
\label{sec:inference-cost}
The learned object is the map $(\mathbf x,\mathbf g)\mapsto U_\theta(\mathbf x;\mathbf g)$ from spatial coordinates and geometry parameters to a vector-valued basis. Once trained, this map supplies basis values and analytic spatial derivatives at newly requested locations without repeating variational optimization or solving a new eigenproblem. It supports arbitrarily dense spatial resampling: there is no fixed output lattice or discrete set of admissible query points. For a geometry-conditioned model, geometry parameters can likewise be queried continuously within its learned family. The three-dimensional capacity experiment uses a fixed-domain checkpoint; the continuous geometry-family construction is exercised in the planar cases.

\paragraph{Online latency and numerical comparison.}
\begingroup
The map supplies spatial fields once a geometry is selected; operator assembly then prepares their recurrent coefficient dynamics. Synchronized measurements give mean RK4 update times of 3.79 ms at $K=512$ and 26.18 ms at $K=1024$, with field reconstruction charged at requested outputs. \docref{Figure}{fig:coefficient-capacity} places integration costs alongside flow errors. This cost separation makes reuse explicit: trajectories on one geometry share the basis and assembled operators, while initial projection and selected field outputs belong to each query.
\par\endgroup

\paragraph{When repeated inference recovers training cost.}
For $n$ trajectories on one geometry, let $T_{\rm train}$ denote training, $T_{\rm off}$ the reusable preparation cost, and $T_{\rm ROM},T_{\rm NS}$ the per-trajectory integration costs. The smallest integration-only break-even count $n_\star$ is given below, with $\lceil\cdot\rceil$ denoting rounding up to an integer:
\begin{equation}
n_\star=\left\lceil\frac{T_{\rm train}+T_{\rm off}}{T_{\rm NS}-T_{\rm ROM}}\right\rceil,
\qquad T_{\rm NS}>T_{\rm ROM}.
\label{eq:inference-break-even}
\end{equation}
{Repeated initial-condition queries amortize both training and reusable operator preparation. \docref{Table}{tab:inference-amortization} gives the counts and exact timer boundaries. Projection and requested field output enter the per-query cost when included in the workload. For a parameterized family, training is shared across geometry queries, while each geometry's basis generation and operator preparation are charged to that query (\docref{Appendix}{app:inference-amortization}).}

\section{Conclusion}
\begingroup
Neural eigenfluids express a physically structured fluid spectrum through a continuous spatial basis. Differential lifts and potential traces define admissible velocities; geometry-local orthogonalization and Rayleigh minimization organize their positive modes. Canonical gauge selection gives potential-space learning a lower-derivative formulation, with the spectral-variance relation explaining its weighting of approximate modes. A neural Dirichlet construction supplies the independent harmonic circulation coordinates.

The experiments connect this structure to representation and dynamics. Learned spectral spans sustain accurate periodic vortex merger and improve energy capture and autonomous accuracy over equal-capacity Laplacian controls on directly sampled random fields. Completed bases express obstacle-scale circulation, and matched conditioning tests identify how harmonic coordinates affect evolution. In three dimensions, increasing capacity reduces spatial representation error while exposing the contribution of coefficient dynamics to total flow error. Full quadratic assembly moves the same reduced generator into coefficient space, allowing repeated initial states to share their spatial basis and operators. Continuous evaluation, physical spectral organization, and reusable autonomous computation are thus realized within one reduced representation.
\par\endgroup

\pdfbookmark[1]{References}{navigation-references}
\bibliographystyle{plainnat}
\begingroup
\small
\setlength{\bibsep}{1.2pt plus .5pt}
\bibliography{references}
\endgroup
\appendix

\section{Gauge representatives, projection, and controls}
\label{app:theory}
This appendix develops the \hyperref[app:vector-gauge]{gauge projection} and the \hyperref[app:old-Pot]{Potential-AL and Potential-FG controls} used in the construction. The \hyperref[app:spectrum]{spectral correspondence} and \hyperref[app:topology-rom]{harmonic completion} follow, together with \hyperref[app:3d-controls]{three-dimensional comparisons}, \hyperref[app:protocol]{experimental protocols}, \hyperref[app:modes]{mode atlases}, \hyperref[app:nkb]{NKB comparisons}, \hyperref[app:fields]{flow visualizations}, and \hyperref[app:runtime]{computational costs}.

\subsection{Scalar constants and boundary-pinned representatives}
\label{app:scalar}
For a connected domain, $\nabla^\perp(\psi+c)=\nabla^\perp\psi$. The mean-zero projection $P_0\psi=\psi-\langle\psi\rangle_\Omega$ is invariant to a constant shift, preserves all spatial derivatives, and selects a unique representative. Let $w$ be the column vector of positive quadrature weights, $\mathbf1$ the all-ones vector, and $\psi_q=\psi(\mathbf x_q)$. The discrete projector is $P_0=I-\mathbf1 w^{\top}/(\mathbf1^{\top}w)$; backpropagation uses its transpose in Euclidean coordinates. In particular,
\[
\sum_q w_q(P_0\psi)_q=0,\qquad
\norm{P_0\psi}_Q^2=\norm{\psi}_Q^2-\frac{(\sum_qw_q\psi_q)^2}{\sum_qw_q}.
\]
For a periodic scalar potential, this removes the constant kernel before the spectral mass is formed. Without removal, $\norm{\nabla(\psi+c)}^2/\norm{\psi+c}^2\to0$ as $|c|\to\infty$, although the physical velocity is unchanged.

For a homogeneous Dirichlet scalar eigenproblem, $\psi=0$ on the boundary excludes nonzero constant shifts from the admissible space and fixes the representative used in its mass form. Centering this function changes the mass form and is therefore used here for display only: it changes the common boundary value to $-\langle\psi\rangle_\Omega$ while preserving velocity. Finite-element scalar displays use area-weighted means, and periodic displays use uniform-grid means.

On a multiply connected domain, independent boundary constants encode physical harmonic velocities. Subtracting one global constant preserves their differences. It therefore fixes scalar gauge without deleting circulation degrees of freedom. Each connected component requires its own additive convention if the fluid domain itself is disconnected.

\subsection{Vector gauge as a quotient-space projection}
\label{app:vector-gauge}
Let $d_R^k$ denote exterior differentiation of degree-$k$ forms with relative boundary trace; its adjoint is $\delta_R^{k+1}$. The closed potential-to-velocity operator is $\C=\star d_R^{D-2}$. Writing $\ker$ for the nullspace and $\perp$ for mass-orthogonal complement, the canonical potential space is $(\ker\C)^\perp$. In 3D, exact-gradient removal is the orthogonal projection defined using scalar test functions $\eta\in H_0^1(\Omega)$, where $H_0^1$ is the Sobolev space of functions with square-integrable first derivatives and zero boundary trace:
\[
\ip{\nabla\phi}{\nabla\eta}=\ip{\widetilde{\mathbf A}}{\nabla\eta},\qquad \eta\in H_0^1(\Omega).
\]
Integration by parts gives \doceqref{eq:poisson-gauge}. Removing the remaining relative harmonic kernel completes the projection. Both corrections have zero curl, and a Dirichlet scalar has no tangential derivative on the boundary, so curl and the relative tangential trace are preserved.

To see why this operation belongs before the loss, choose a nonzero $\kappa\in\ker\C$. The family $\mu_\alpha=\mu+\alpha\kappa$ represents exactly the same velocity for every real scalar $\alpha$, yet
\begin{equation}
\frac{\norm{\C\mu_\alpha}^2}{\norm{\mu_\alpha}^2}
=\frac{\norm{\C\mu}^2}{\norm{\mu+\alpha\kappa}^2}\longrightarrow0
\quad\text{as }|\alpha|\longrightarrow\infty.
\end{equation}
The unprojected quotient can therefore decrease without improving a physical mode. On the canonical space this ambiguity disappears. Write $L^2\Lambda^k$ for square-integrable degree-$k$ forms, $\operatorname{Ran}$ for operator range, an overline for closure, and $\oplus$ for orthogonal direct sum. The relative harmonic space $\Harm_R^k$ contains the closed, coclosed forms with relative boundary trace. With the orthogonal Hodge decomposition {\citep{arnold2006finite}}
\begin{equation}
L^2\Lambda^{D-2}=\overline{\operatorname{Ran}d_R^{D-3}}
\oplus\Harm_R^{D-2}\oplus\overline{\operatorname{Ran}\delta_R^{D-1}},
\end{equation}
the last summand is $(\ker\C)^\perp$. Let $\Delta_H=d\delta+\delta d$ be the positive Hodge Laplacian and $G_R$ its relative-boundary Green pseudoinverse. The exact projector is $P_{\rm gf}=\delta_R d_R G_R$, and it satisfies
\begin{equation}
P_{\rm gf}(\mu+\kappa)=P_{\rm gf}\mu,\quad
\C P_{\rm gf}\mu=\C\mu,\quad
\delta_R P_{\rm gf}\mu=0.
\end{equation}
Normalization and Gram--Schmidt consequently act on physical equivalence classes. In a finite auxiliary space, the invariance identity applies to the resolved kernel directions; \docref{Appendix}{app:old-Pot} quantifies the residual for the three-dimensional implementation.

\subsection{Potential-AL and Potential-FG controls and capacity}
\label{app:old-Pot}
Potential-AL and Potential-FG are the noncanonical 3D potential controls. The Potential-AL models at $K=32,64,128$ use divergence and wall constraints with an augmented-Lagrangian objective. The seed-0 $K=256$ Potential-FG model uses finite scalar/harmonic gauge projection. Their original-field reconstruction results appear in \docref{Table}{tab:finite-potential-initial}. A separately trained, matched Potential-FG family spans $K=32,64,128,256$, with the results below.

The matched finite-gauge family completes 2000 training epochs and 80 scheduled validations at each capacity. Every model reconstructs the same 20 independent random large-eddy fields with all retained modes and all $1.6\,\mathrm{M}$ fluid points (\docref{Table}{tab:matched-FG-NEW20-values}). Preparation costs, including geometry and streamed basis-output I/O, appear in \docref{Table}{tab:Pot-training-basis-cost}.

\begin{table}[!htb]\centering\small
\caption{Matched Potential-FG reconstruction on the same 20 independent random large-eddy fields, after 2000 epochs and 80 scheduled validations. Relative velocity $L^2$ errors are reported as mean $\pm$ sample standard deviation (\%); all $1.6\,\mathrm{M}$ points and all $K$ modes are retained. The Laplacian entries use the same fields and capacity.}
\label{tab:matched-FG-NEW20-values}
\begin{tabular}{rcc}\toprule
$K$&Matched Potential-FG&Laplacian\\\midrule
32&$99.928\pm0.045$&$65.924\pm4.039$\\
64&$98.690\pm1.315$&$49.278\pm4.782$\\
128&$95.348\pm3.671$&$34.595\pm3.852$\\
256&$91.001\pm4.096$&$26.511\pm3.946$\\
\bottomrule\end{tabular}
\end{table}

\begin{table}[!htb]\centering\small
\caption{Initial relative velocity $L^2$ errors on the original common initial field ($n=1$ per method and capacity), using all retained modes and all $1.6\,\mathrm{M}$ fluid nodes. Potential completes 2000 epochs and 80 scheduled validations through $K=512$. The controls use Potential-AL at $K=32,64,128$ and Potential-FG at $K=256$. The $K=1024$ Potential training run completes 25 optimizer updates; its update count is reported separately from the full-training representation results.}
\label{tab:finite-potential-initial}
\begin{tabular}{rrrrr}\toprule
$K$&Potential&Velocity&Laplacian&Potential-AL/FG\\\midrule
32&0.843112&0.839273&0.845281&0.917776 (AL)\\
64&0.515772&0.622647&0.521114&0.878704 (AL)\\
128&0.343194&0.400527&0.348509&0.869497 (AL)\\
256&0.235699&0.281845&0.242727&0.994937 (FG)\\
512&0.201610&0.209542&0.196530&\textemdash\\
1024&\textemdash&0.170499&0.178991&\textemdash\\
\bottomrule\end{tabular}\end{table}

The finite projection leaves gradient directions in the curl kernel. For the admissible analytic potential $A=\nabla[e(x)^2\sin(2\pi f(x_1+2x_2+3x_3))]$, where $e$ is the boundary envelope, curl vanishes and $A$ is zero on the boundary. At $f=4$, the scalar and harmonic correction leaves $99.86\%$ of its source-point mass, despite a scalar-system residual below $5\times10^{-15}$. Residual gradient content can enlarge the potential quotient's denominator without adding velocity energy. On the frozen 256-mode model, replacing 441 scalar directions with a nonnested 1392-direction space changes retained potential mass by $44.91\%$ and the mean orthogonalized quotient by a factor of $4.62$, while preserving each raw curl. This measures representative sensitivity rather than auxiliary-space convergence. The measured FP32/FP64 analytic-curl discrepancy is $0.0111\%$. For the Potential-FG seed-0 control, these probes connect finite-gauge sensitivity to its velocity representation; they are not diagnostics of the distinct Potential-AL loss. Full-mass best projection is invariant to nonzero column scaling, so reduced curl amplitude alone does not explain the reconstruction error. Resolving gradient content is the necessary qualification for this 3D potential objective.

\subsection{Differentiable auxiliary gauge projection}
The Potential-FG configuration uses 441 effective scalar gauge directions for differentiable representative selection. A pure-gauge injection diagnostic quantifies the resolved gradient content, with 43.80\% amplitude retained after projection. Scalar corrections preserve the analytic curl exactly. The \hyperref[app:gauge-solvers]{solver comparison below} separately examines refinement, harmonic orthogonality, and derivative consistency using a frozen 32-network auxiliary-solver checkpoint after five updates; it is not the full-training Potential or Potential-FG capacity checkpoint.

\subsection{Accuracy and cost of the auxiliary gauge solve}
\label{app:gauge-solvers}

The auxiliary solve selects the potential representative entering the Rayleigh quotient. Its accuracy depends on the space used to represent the Poisson correction as well as on the precision of the linear algebra. We compare these effects on a fixed 32-network checkpoint obtained after five updates with Fourier projection and a method-of-fundamental-solutions (MFS) boundary correction. All methods use the same smooth Metaball geometry, neural parameters, analytic derivatives, and relative harmonic field, isolating representative selection from learning.

\paragraph{Common assessment.}
Let $A_\theta$ denote the neural vector potential, with raw column $A_j$. Its scalar divergence forcing is $f_j=\diver A_j$, and $\phi_j$ denotes its numerical scalar correction. We define $\epsilon_P$ as the maximum relative Poisson residual over columns $j=1,\ldots,32$; the sums below use unit weights on independent points $x_q$. On 8192 identical independent fluid points, we report
\begin{equation}
\epsilon_P=\max_{j\leq32}
\frac{\left(\sum_q|f_j(x_q)-\Delta\phi_j(x_q)|^2\right)^{1/2}}
{\left(\sum_q|f_j(x_q)|^2\right)^{1/2}}.
\label{eq:gauge-solver-residual}
\end{equation}
The residual measures the fraction of divergence forcing left after the scalar correction, since $\diver(A_j-\nabla\phi_j)=f_j-\Delta\phi_j$. Relative harmonic overlap is measured on a disjoint 110592-node volume rule after the column transform, and tangential trace on 8192 independent points per wall. Columnwise diagnostics characterize gradient correction, harmonic orthogonality, and boundary trace across all 32 modes. Setup, forward, reverse, and assessment times are measured separately on an RTX 4080 SUPER using native CuPy. Forward/reverse measurements use a 4 GiB allocation pool; the R-marked synchronized setup and assessment measurements use a 12 GiB pool with the same complete workload.

\paragraph{Fourier particular solution and boundary correction.}
On an auxiliary box of side $L=1.12$, an exterior continuation agrees with the original neural potential throughout the fluid. Its divergence is sampled on $193^3$ nodes. Let $k\in\mathbb Z^3$ be an integer wavevector ($\mathbb Z$ denotes integers), $\kappa_k=2\pi k/L$ its physical frequency, $x_0$ the box origin, and $\widehat f_k$ the Fourier coefficients of the sampled forcing. The scalar particular solution $\phi_p$ is
\begin{equation}
\phi_p(x)=-\sum_{k\ne0}\frac{\widehat f_k}{|\kappa_k|^2}
e^{\mathrm i\kappa_k\cdot(x-x_0)}
+\widehat f_0\frac{|x-(0.5,0.5,0.5)|^2}{6}.
\end{equation}
The mean term preserves nonzero-mean forcing. The represented Laplacian is the Fourier interpolant $f_N$ with $N=193$ nodes per axis. The independent residual $f_j-f_N$ measures the source representation at this resolution. A Laplace method-of-fundamental-solutions (MFS) correction with 6144 exterior sources and 18432 boundary collocation nodes enforces the zero Dirichlet trace in least squares \citep{fairweather1998method}. This harmonic boundary correction complements the volume particular solution. We separately test the training arithmetic, FP64 fast Fourier transforms (FFTs) on the same FP32 samples, and an FP64 source/FFT path. A reuse control combines repeated source FFTs and accumulates spectral cotangents before the source-grid inverse FFT, keeping the represented operator unchanged.

\paragraph{Smooth weak projection.}
For a scalar test function $v\in H_0^1(\Omega)$, integration by parts gives the same continuous Dirichlet problem in the form
\begin{equation}
\int_\Omega\nabla v\cdot\nabla\phi\,\dd x
=\int_\Omega\nabla v\cdot A_\theta\,\dd x,
\qquad v\in H_0^1(\Omega).
\end{equation}
We use $C^2$ tensor cubic B-splines {\citep{deboor1978splines}} multiplied by the original analytic zero-trace envelope. Let $b_i$ be the enveloped spline basis functions, $\zeta_i$ their coefficients in $\phi_h=\sum_i\zeta_i b_i$, $\zeta=(\zeta_i)$ their coefficient vector, $\mathsf K$ the gradient stiffness matrix, and $r=(r_i)$ the load vector. The discretization is
\begin{equation}
\mathsf K_{ij}=\sum_qw_q\nabla b_i(x_q)\cdot\nabla b_j(x_q),\quad
r_i=\sum_qw_q\nabla b_i(x_q)\cdot A_\theta(x_q),\quad \mathsf K\zeta=r.
\end{equation}
The neural right-hand side and its parameter reverse require potential values rather than divergence derivatives. Basis gradients and Laplacians are analytic, including across spline knots. Resolution-12, -20, and -24 spaces use 32768, 110592, and 262144 fitting quadrature points. {Matrix-free conjugate gradients \citep{hestenes1952methods} apply the weighted gradient Gram operator; a standard two-level construction \citep{briggs2000multigrid} uses cubic refinement prolongation, symmetric smoothing, and a factored coarse solve.} Severely cut supports are removed when their quadrature gradient-energy diagonal is at most $10^{-5}$ of the maximum. The selection and quadrature remain fixed during differentiation. The algebraic solve residual measures convergence within this space, while \doceqref{eq:gauge-solver-residual} measures the resolved neural forcing.

\paragraph{Volume potential with multipole acceleration.}
{The alternative particular solution follows volume-potential methods \citep{malhotra2016volume} and uses the Newton potential of the piecewise-constant source density $f_h$ over the auxiliary box $B$, with $y$ the integration point:}
\begin{equation}
\phi_p(x)=-\frac1{4\pi}\int_B\frac{f_h(y)}{|x-y|}\,\dd y.
\end{equation}
The volume representation uses constant cuboid densities in $B=[-0.02,1.02]^3$, with the same interior forcing and exterior continuation. Near cuboids, including self interactions, are integrated analytically. {Degree-six Cartesian multipole-to-multipole, multipole-to-local, and local-to-local translations follow the fast multipole method (FMM) \citep{greengard1987fast} to evaluate the far field, whose local Taylor polynomials are harmonic.} The same MFS correction enforces the boundary trace. A source-adaptation variant refines coarse cells when a neighboring density difference, normalized by the per-column coarse-source root mean square (RMS) and maximized over the 32 columns, exceeds $0.5$. The sampling map is fixed for the adjoint, and evaluation uses a uniform tree. The fixed sampling map permits consistent differentiation through the adaptive piecewise-constant representation, following the volume-potential construction \citep{malhotra2016volume}.

\begin{table}[!htb]\centering\footnotesize\setlength{\tabcolsep}{3pt}
\caption{Gauge-solver accuracy and cost at the fixed checkpoint. \docref{Figure}{fig:gauge-solver-ablation} plots the Poisson residual $\epsilon_P$. Harmonic overlap and tangential trace are percentages; all four timing columns are seconds. Size denotes grid size or retained spline degrees of freedom. Setup and assessment are excluded from forward/reverse. Superscript R denotes a synchronized measurement of the complete workload; forward/reverse columns exclude setup and assessment.}
\label{tab:gauge-solver-ablation}

\begin{tabular}{lrrrrrrr}\toprule
Method&Size&Harm.&Wall tan.&Forward&Reverse&Setup&Assess.\\\midrule
Fourier--MFS & $193^3$ & 0.0177 & 0.6904 & 49.07 & 138.60 & $54.57^{\mathrm R}$ & $23.82^{\mathrm R}$\\
Fourier--MFS, reuse & $193^3$ & 0.0177 & 0.6904 & 48.28 & 136.12 & $62.86^{\mathrm R}$ & $23.32^{\mathrm R}$\\
Spline 12 & 3,106 & 0.0181 & 0.0012 & 6.63 & 6.98 & 1.98 & 5.90\\
Spline 20 & 10,648 & 0.0015 & 0.0012 & 15.40 & 16.60 & 4.75 & 5.91\\
Spline 24 + two-level & 17,086 & 0.0036 & 0.0012 & 71.00 & 78.96 & 11.69 & 5.84\\
Volume FMM 32 & $32^3$ & 0.0242 & 4.8727 & 12.15 & 13.43 & 22.99 & 61.76\\
Volume FMM 64 & $64^3$ & 0.0261 & 3.6399 & 14.23 & 17.66 & 22.82 & 60.63\\
Volume FMM 64 adaptive & $64^3$ & 0.0261 & 3.6399 & 14.53 & 17.96 & 23.71 & 60.51\\\bottomrule\end{tabular}\end{table}

Forward and reverse totals include synchronized component measurements; one-time setup and independent assessment are separate columns. Setup combines spline space/quadrature and optional multigrid construction, or the FMM and fixed harmonic moments; assessment combines independent residual, harmonic, and wall checks. Spline forward timing includes query-point construction. The Fourier and Fourier-reuse measurements share identical neural passes and include source-array export, while spline and volume methods use in-memory sources. Both particular-solution families reuse the geometry-dependent MFS factors.

\begin{table}[!htb]
\centering\small
\caption{Floating-point precision on the same $193^3$ grid. The first two rows use identical source samples. Source time measures one all-32 forward pass, including export.}
\label{tab:gauge-precision}
\begin{tabular}{lrr}\toprule
Arithmetic & $\epsilon_P$ (\%) & Source forward (s)\\\midrule
Training arithmetic & 7.989593 & 37.24 \\
FP32 source, FP64 FFT & 7.989595 & 37.24 \\
FP64 source and FFT & 7.989582 & 266.41 \\
\bottomrule\end{tabular}\end{table}
\paragraph{Representation and arithmetic.}
The precision controls yield nearly identical forcing residuals (\docref{Table}{tab:gauge-precision}). With identical source samples, changing FFT precision produces a relative Laplacian difference of $4.20\times10^{-7}$. Source-space resolution therefore provides the effective control over representative selection at this checkpoint.

Transform reuse preserves forward outputs to the reported precision and changes the source cotangent by a relative $2.50\times10^{-12}$. It reduces the three-evaluation Fourier stage from 9.91 to 8.99 seconds and the source-grid transpose from 12.96 to 10.48 seconds. Total component time decreases from 187.67 to 184.40 seconds (1.74\%); neural source evaluation and parameter differentiation account for 85.5\% of the original total.

\paragraph{Auxiliary-space resolution.}
Refinement systematically improves the spline and volume corrections (\docref{Figure}{fig:gauge-solver-ablation}), while the harmonic projection maintains small overlap across the solver family (\docref{Table}{tab:gauge-solver-ablation}). These complementary diagnostics identify two controls for representative selection: the approximation space resolves gradient content, and the mass projection selects harmonic coordinates. Source adaptation refines 31,432 of 32,768 coarse cells (95.9\%), with 284,224 source evaluations including the coarse pass, compared with 262,144 for the uniform rule. This distribution favors uniform refinement for the present forcing.
\paragraph{Operator consistency.}
FMM evaluation and transpose consistency are tested using $16^3$ cuboids, 48 query points, and two random right-hand sides. The analytic cuboid kernel agrees with nonsingular Gauss integration to relative $7.30\times10^{-11}$. Against direct cuboid integration, increasing the far expansion degree from four to six reduces relative gradient error from $1.38\times10^{-3}$ to $1.90\times10^{-4}$; the degree-six transpose discrepancy is $2.73\times10^{-16}$. These comparisons isolate evaluation of the chosen density from density approximation. For the manufactured scalar $e(x)[1+0.2\sin(2\pi x_1)+0.1x_2x_3]$, the refined spline method gives relative gradient and Poisson errors of $2.85\times10^{-4}$ and $8.82\times10^{-4}$, with linear residual $4.95\times10^{-9}$. Together with \docref{Figure}{fig:gauge-solver-ablation} and \docref{Table}{tab:gauge-solver-ablation}, these checks connect accurate operator evaluation to systematic control through auxiliary-space refinement.

\begin{figure}[!htb]
\centering\includegraphics[width=\linewidth,height=.72\textheight,keepaspectratio]{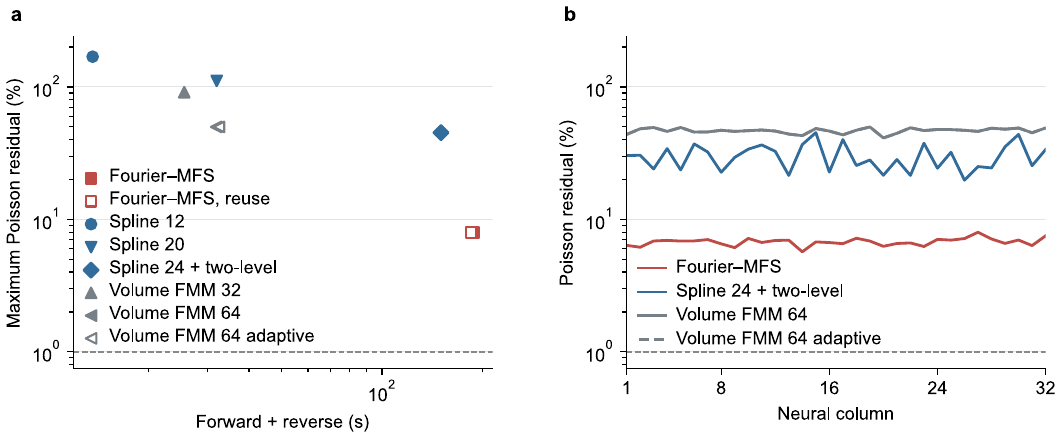}
\caption{\textbf{Gauge-solver accuracy and cost.} Left: all-32 forward-plus-reverse time versus independent Poisson residual. Right: per-column residuals for Fourier--MFS and the finest tested spline and volume spaces. All methods evaluate the same fixed neural checkpoint.}
\label{fig:gauge-solver-ablation}
\end{figure}

\paragraph{Connection to the velocity span.}
At a fixed network, each scalar correction has zero curl, so all solvers retain the same 32-dimensional analytic velocity span. The solver comparison therefore measures representative selection for the potential objective; changes to the learned velocity span arise through subsequent optimization.

\section{Spectral correspondence, objective weighting, and boundaries}
\label{app:spectrum}
\subsection{Pairing the nonzero spectrum}
Take $\C$ to be a densely defined closed operator between the potential and velocity Hilbert spaces, with Hilbert adjoint $\C^*$. {The classical spectral relationship follows from the paired normal operators of the differential complex \citep{arnold2006finite}, with the differential-complex domains determining which physical boundary problem they realize.}

\begin{proposition}[Positive-spectrum correspondence]
\label{prop:positive-spectrum}
If $\C^*\C\mu=\lambda\mu$ with $\lambda>0$ and $\norm\mu=1$, then $u=\C\mu/\sqrt\lambda$ has unit norm and satisfies $\C\C^*u=\lambda u$. Conversely, $\mu=\C^*u/\sqrt\lambda$ maps a positive eigenfield back to the canonical potential space.
\end{proposition}
\begin{proof}
The norm identity is $\norm{\C\mu}^2=\ip\mu{\C^*\C\mu}=\lambda$. Applying $\C$ to the potential eigenproblem gives the velocity equation. Applying $\C^*$ gives the converse. Positive eigenvectors are orthogonal to the corresponding kernels. With the relative/absolute Hodge domains, the range of the lift is the coexact absolute velocity sector; physical harmonic velocities form its orthogonal complement.
\end{proof}

For the absolute Hodge domain, smooth velocities obey $\mathbf u\cdot\mathbf n=0$ and $\mathbf n\times\curl\mathbf u=0$ in 3D (zero boundary vorticity in 2D). These boundary conditions, together with the differential energy, identify the physical realization of the paired normal operators.

\begingroup
\subsection{Periodic and decaying free-space specialization}
For periodic zero-mean potentials, $\C=\nabla^\perp$ in 2D and $\C=\curl$ on divergence-free canonical vector potentials in 3D commute with $-\Delta$. The positive eigenspaces are therefore related by the normalized lift in \hyperref[prop:positive-spectrum]{the proposition above}.

On $\mathbb R^D$, smooth fields with sufficient decay give the same identities without a wall boundary contribution. In Fourier coordinates, let $\xi$ be spatial frequency, $i$ the imaginary unit, a hat the Fourier transform, and $\xi^\perp=(\xi_y,-\xi_x)$ in 2D. The positive symbol is $|\xi|^2$; $\widehat{u}=i\xi^\perp\widehat\psi$ in 2D and $\widehat{u}=i\xi\times\widehat A$ with $\xi\cdot\widehat A=0$ in 3D. For $\xi\ne0$, these maps pair the potential and divergence-free velocity spectral sectors. Free space generally has continuous spectrum, so this statement is a spectral correspondence rather than a discrete list of square-integrable eigenmodes. For either periodic cancellation or decay at infinity, integration by parts also gives $\|\nabla u\|^2=\|\curl u\|^2$ for divergence-free velocity. A bounded curved-wall problem requires its specified boundary domain. Even when the positive spectra agree, $R_p$ and $R_v$ coincide on exact eigenspaces and differ on mixtures according to \doceqref{eq:variance}.
\par\endgroup

\subsection{Weak formulation and its 2D specialization}
For scalar Dirichlet trial functions $\psi,\eta\in H_0^1(\Omega)$,
\begin{equation}
\ip{\C\psi}{\C\eta}=\int_\Omega\nabla\psi\cdot\nabla\eta\,\dd\mathbf x.
\end{equation}
Stationarity of $\frac12\norm{\nabla\psi}^2-\frac\lambda2(\norm\psi^2-1)$ therefore yields
\begin{equation}
\int_\Omega\nabla\psi\cdot\nabla\eta\,\dd\mathbf x
=\lambda\int_\Omega\psi\eta\,\dd\mathbf x,
\qquad \eta\in H_0^1(\Omega).
\end{equation}
The first-derivative form realizes $\C^*\C=-\Delta_D$, where $-\Delta_D$ is the positive Dirichlet Laplacian. If $-\Delta\psi=\lambda\psi$ and $u=\nabla^\perp\psi$, then
\begin{equation}
\diver u=0,\quad u\cdot\mathbf n=0,\quad -\Delta u=\lambda u,
\qquad \omega=\partial_xu_y-\partial_yu_x=\lambda\psi.
\end{equation}
The boundary vorticity vanishes because the streamfunction trace does. This explicitly identifies the compatible 2D Hodge problem. On the periodic domain, the same calculation uses mean-zero periodic functions in place of Dirichlet trial functions.

For compatible form domains, the general calculation is the commutation identity $\Delta_Hd=d\Delta_H$, which follows from $d^2=0$, together with $\Delta_H\star=\star\Delta_H$. It yields $\Delta_H\star d\mu=\lambda\star d\mu$ for a potential eigenform. The absolute harmonic velocities remain outside the positive lifted sector and must be supplied independently.

\subsection{Spectral moments and the choice of velocity energy}
For nonnegative integer $r$, define the spectral moment $m_r=\sum_j\lambda_j^r|a_j|^2$ in an orthonormal canonical eigenbasis. Potential and velocity Rayleigh ratios are $m_1/m_0$ and $m_2/m_1$. In 2D, a direct-vorticity objective gives $m_3/m_2$. Cauchy--Schwarz gives $m_r^2\le m_{r-1}m_{r+1}$, hence
\[
\frac{m_1}{m_0}\le\frac{m_2}{m_1}\le\frac{m_3}{m_2}.
\]
Equality holds when the nonzero coefficients occupy one eigenvalue eigenspace. At repeated eigenvalues the eigenspace, rather than the phase or orientation of individual modes, is identifiable.

The Hodge velocity form is $\norm{\curl\mathbf u}^2+\norm{\diver\mathbf u}^2$, reducing to curl energy on divergence-free fields. The componentwise Dirichlet form $\norm{\nabla\mathbf u}^2$ coincides with it on a periodic domain and differs by boundary terms on a general bounded domain. The 3D velocity experiments use the componentwise Dirichlet quotient during training and the common discrete NS generator during evolution. Their \hyperref[app:boundary-energy]{boundary-energy relation} is given below.

\subsection{The boundary term in gradient and curl energy}
\label{app:boundary-energy}
For a smooth vector field on a Euclidean domain, with $\dd S$ denoting boundary surface measure, integration by parts gives
\begin{equation}
\norm{\curl u}^2+\norm{\diver u}^2-\norm{\nabla u}^2
=\int_{\partial\Omega}\left[(u\cdot n)\diver u-n\cdot((u\cdot\nabla)u)\right]\,\dd S.
\end{equation}
Periodic faces cancel, and a strong componentwise-zero trace also eliminates the right-hand side. For a tangent field on a curved wall, the remaining term depends on the wall geometry. Thus a bounded-domain spectral problem is specified jointly by its energy and boundary domain.

\section{Topology completion and reduced evolution}
\label{app:topology-rom}
\subsection{One scalar gauge constant versus physical boundary differences}
Let $r$ be radial distance, $r_{\rm in},r_{\rm out}$ the inner and outer radii, and $e_\theta$ the azimuthal unit vector. The annulus $r_{\rm in}<r<r_{\rm out}$ provides a concrete example; $\gamma$ is a counterclockwise loop around its hole and $\dd l$ its oriented line element. The streamfunction $\psi=\log r$ has
\begin{equation}
u=\nabla^\perp\psi=-\frac1r e_\theta,\qquad
\diver u=0,\quad \curl u=0,\quad
\oint_{\gamma}u\cdot\dd l=-2\pi
\end{equation}
for a counterclockwise loop around the hole. Subtracting the spatial mean of $\psi$ leaves this circulation unchanged. Setting both boundary values to the same constant would instead remove this harmonic component. The former changes only the representative; the latter changes the physical field.

For $h$ holes, the harmonic measures of \doceqref{eq:harmonic-measure} form convenient independent coordinates. Let $H_\ell=\nabla^\perp\chi_\ell$ and let $\gamma_m$ encircle hole $m$ counterclockwise. Define the period matrix $P_{m\ell}=\oint_{\gamma_m}H_\ell\cdot\dd l$. If unit-circulation coordinates are desired, set $\widetilde H_\ell=\sum_jH_j(P^{-1})_{j\ell}$. Then $\oint_{\gamma_m}\widetilde H_\ell\cdot\dd l=\delta_{m\ell}$. This invertible change of coordinates preserves the span. The reduced mass matrix accounts for the actual harmonic basis used in computation.

\subsection{Harmonic training with a numerator-only Rayleigh energy}
\label{app:harmonic-loss}
The harmonic sector already appears in the spectral ROM of \citet{liu2015model}, whose \href{https://doi.org/10.1145/2816795.2818130}{Section~3.1} includes the zero-frequency vector fields and whose \href{https://doi.org/10.1145/2816795.2818130}{Appendix~A} discusses both kernel eigenvectors and sparse harmonic solves. Our construction learns this sector through geometry-conditioned potentials. Each harmonic measure is learned in the affine admissible space $\mathcal A_\ell$ below, where $\chi$ is a trial scalar field and $H^1(\Omega)$ is the Sobolev space without a zero-trace restriction:
\begin{equation}
\mathcal A_\ell=\{\chi\in H^1(\Omega):\chi|_{\Gamma_0}=0,\ \chi|_{\Gamma_m}=\delta_{\ell m}\}
=L_\ell+H_0^1(\Omega).
\end{equation}
The prescribed boundary differences fix the amplitude and circulation coordinate. The Dirichlet principle therefore supplies the unnormalized potential Rayleigh numerator in \doceqref{eq:harmonic-loss}, rather than a spectral quotient \citep{e2018deep}. For Dirichlet energy $E(\chi)=\int_\Omega|\nabla\chi|^2$ and any admissible variation $v\in H_0^1(\Omega)$ with real variation parameter $\varepsilon$,
\begin{equation}
\frac{\dd}{\dd\varepsilon}E(\chi+\varepsilon v)\bigg|_{\varepsilon=0}
=2\int_\Omega\nabla\chi\cdot\nabla v
=-2\int_\Omega(\Delta\chi)v=0.
\end{equation}
Thus the minimizer solves $\Delta\chi=0$ with the cardinal traces, giving a curl-free velocity $u_H=\nabla^\perp\chi$. The traces exclude the trivial zero field, and the energy is strictly convex in the zero-trace correction. If $\chi^*$ is the harmonic extension, then
\begin{equation}
E(\chi^*+v)=E(\chi^*)+\norm{\nabla v}^2,\qquad v\in H_0^1(\Omega),
\end{equation}
which identifies the unique energy minimizer.

A potential-mass denominator would change this boundary-value problem. Setting $M(\chi)=\int_\Omega\chi^2$ and $R(\chi)=E(\chi)/M(\chi)$ yields the stationarity condition
\begin{equation}
\int_\Omega\nabla\chi\cdot\nabla v
=R(\chi)\int_\Omega\chi v,\qquad v\in H_0^1(\Omega),
\end{equation}
or $-\Delta\chi=R(\chi)\chi$, instead of the required harmonic equation. The positive-mode quotient determines an eigenvalue with free amplitude; the harmonic loss determines the minimum-energy extension of fixed, nonhomogeneous traces. Zero Hodge energy of the resulting velocity refers to its curl, while the training energy $E(\chi)=\norm{u_H}^2$ is nonzero for a nontrivial circulation coordinate. The implementation averages $|\nabla\chi_\ell|^2$ over uniformly sampled interior points for each output and differentiates these per-output energies; on a fixed geometry their sum differs from \doceqref{eq:harmonic-loss} by the constant factor $1/|\Omega|$, where $|\Omega|$ is the domain volume. Validation additionally averages over the $h$ outputs. Subsequent mass or circulation normalization is a change of basis after learning and leaves the harmonic span unchanged.

\subsection{Conditioning the harmonic coordinates}
In the neural parameterization $L_\ell+e_{\mathbf g}n_{\theta,\ell}$, the lift fixes the cardinal traces before training. Rescaling its interior construction changes optimization gradients while preserving the boundary-value problem. The raw/scaled comparison holds the positive modes, cardinal traces, correction envelope, initialization protocol, and epoch budget fixed, isolating conditioning at finite optimization accuracy.

\subsection{Full reduced operators and physical initialization}
{Following the spectral Galerkin fluid operators of \citet{dewitt2012fluid,liu2015model}, for a static geometry and basis $U$ define}
\begin{equation}
Q_{ijk}=-\ip{U_i}{(U_j\cdot\nabla)U_k},\qquad
M_{ij}=\ip{U_i}{U_j},\qquad S_{ij}=\ip{dU_i^\flat}{dU_j^\flat}.
\end{equation}
\doceqref[Equation]{eq:reduced-dynamics} is the homogeneous unforced form. Substituting $u=u_{\rm base}+Uc$ adds constant and linear terms, the base-flow time derivative, and projected external forcing. A full-mass initial projection minimizes $\norm{u_{\rm base}+Uc-u_0}$ using the assembled physical velocity metric.

For exact normalized positive Hodge eigenfields, $S$ is diagonal with entries $\lambda_i$, so viscosity alone gives $a_i(t)=e^{-\nu\lambda_i t}a_i(0)$. Ideal harmonic modes occupy the zero Hodge sector. The actual experiments retain assembled operators; in 3D they project the stated grid generator, including its normal penalty. The exact-mode identity gives the spectral interpretation of viscous decay; the assembled operators realize the measured dynamics.

\subsection{Invariant subspaces and the mode atlas}
An eigenvalue multiplicity identifies an eigenspace whose basis can rotate. For mass-orthonormal subspace matrices $V$ and $W$, the singular values of $V^{\top}M W$ are cosines of their principal angles, providing a basis-independent comparison. The atlases retain each method's ordering and display normalization. Across topologies, the harmonic dimension changes with the first Betti number; this changes the physical space itself.

\section{Three-dimensional capacity, computation, and transfer}
\label{app:3d-controls}
\subsection{Coefficient-only acceleration of the 3D Galerkin system}
\label{app:coefficient-acceleration}
For the original $128^3$ normal-penalty generator, let $B_i$ denote a sampled learned velocity field, extended by zero outside the fluid nodes, and let $P_h$ be the original self-adjoint Leray, dealiasing, and zero-mean projector. Here $\Omega_h$ is the set of fluid grid nodes, $h=1/128$ the Cartesian grid spacing, and $\nabla_h,\Delta_h$ the grid spectral gradient and Laplacian. The spatial weight $\chi$ and positive penalty timescale $\eta$ specify the original normal penalty along wall normal $n$. The operators in \doceqref{eq:offline-rom} are
\begin{align}
M_{ij}&=h^3\sum_{x\in\Omega_h}B_i(x)\cdot B_j(x),\\
A_{ij}&=\left\langle B_i,P_h\left[\nu\Delta_h B_j-\frac{\chi}{\eta}(B_j\cdot n)n\right]\right\rangle_h,\\
\mathcal Q_{ijk}&=h^3\sum_{x\in\Omega_h}(P_hB_i)(x)\cdot\left[B_j(x)\times(\nabla_h\times B_k)(x)\right].
\end{align}
The full-grid projector and derivatives act before restriction to the fluid nodes. The adjoint projector in the last line moves the output projection onto the test field, permitting blocked matrix products without changing the nonlinear operator. Input pairs are stored as $j\leq k$, summing $\mathcal Q_{ijk}+\mathcal Q_{ikj}$ for $j<k$. This stores $K^2(K+1)/2$ FP64 entries, occupying 0.501 and 4.004 GiB at $K=512$ and $1024$, respectively; both remain on the GPU during integration. All $1.6\,\mathrm{M}$ fluid nodes and all modes are retained; the sampled basis cache remains at its original FP32 precision, with operator construction and coefficient arithmetic in FP64.

At $K=512$, native CuPy evaluation of two production right-hand sides agrees with full-grid projection to $1.17\times10^{-14}$ relative error; one RK4 step differs by $4.53\times10^{-17}$ in relative mass norm. The complete coarse-step trajectory agrees with the full-grid Galerkin trajectory to $1.25\times10^{-14}$. These comparisons verify that tensor assembly preserves the selected reduced dynamics.

On an NVIDIA RTX 4080 SUPER, the assembled $K=512$ ROM is $17.74\times$ faster than the matched-step numerical NS solver over 1856 steps to $T=1.16$. This ratio compares synchronized online integration. The NS solver retains its split viscosity and normal-penalty integrator; the ROM integrates the projected unsplit generator using RK4. \docref{Appendix}{app:runtime} separates training, field preparation, assembly, initial projection, and requested output from integration.

Halving the ROM time step changes the trajectory by at most $2.67\times10^{-12}$ in relative mass norm, confirming temporal convergence of the assembled dynamics. Refining the NS reference from $0.0025$ to $0.000625$ changes its final velocity by $2.12\%$. \docref{Figure}{fig:coefficient-acceleration} presents the flow error and kinetic energy alongside the measured online and preparation costs, characterizing the accuracy--cost tradeoff of coefficient-space execution.
\begin{figure}[!htb]
\centering\includegraphics[width=\linewidth,height=.72\textheight,keepaspectratio]{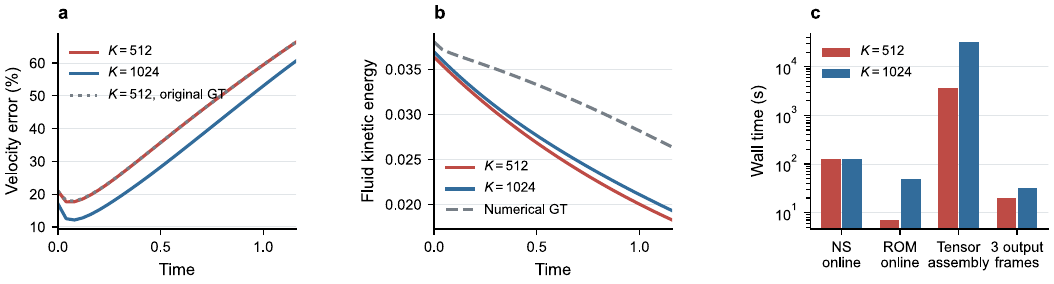}
\caption{\textbf{Accuracy, energy, and cost of coefficient-space dynamics.} Red and blue denote $K=512$ and $1024$. (a) Autonomous velocity errors against the shared refined NS reference on all $1.6\,\mathrm{M}$ fluid nodes; the dotted gray curve assesses $K=512$ against the original reference. (b) Fluid kinetic energy. (c) Synchronized NS/ROM online integration, offline tensor assembly, and three-frame output times. Both 2000-epoch models use full-mass initial projection, all retained coordinates, and $\Delta t=0.000625$. Time-step checks and online speed ratios are given in the adjacent text.}
\label{fig:coefficient-acceleration}
\label{fig:coefficient-capacity}
\end{figure}

At $K=1024$, coefficient-space execution gives a synchronized online speedup of $2.61\times$ over the matched-step NS control. \docref{Figure}{fig:coefficient-capacity} compares its error, energy, and timing with $K=512$. All trained coordinates are retained.

\paragraph{Operator agreement and stage costs.}
At $K=1024$, assembled and direct full-grid right-hand sides agree to $3.81\times10^{-14}$ relative error; the one-step RK4 discrepancy is $6.48\times10^{-17}$ in relative mass norm. Halving the ROM step changes the trajectory by at most $1.49\times10^{-11}$. {\docref{Appendix}{app:runtime}} collects field preparation, tensor assembly, loading, full-grid checks, projection, and output costs alongside the NKB and Laplacian comparisons. The agreement measurements show that assembly preserves the reduced generator, connecting the online cost directly to the learned span and projected dynamics assessed above.

\subsection[Randomized-field representation at K1024]{Randomized-field representation at $K=1024$}
\label{app:randomfield-values}

\begin{table}[!htb]\centering\small
\caption{Randomized-field representation at $K=1024$. Each frozen basis projects the same 20 independent fields used in Figure~\ref{fig:case3-Pot-capacity}(a), with all $1.6\,\mathrm{M}$ fluid points and all retained modes. Entries are relative velocity $L^2$ error in percent, reported as mean $\pm$ sample standard deviation.}
\label{tab:randomfield-K1024}
\begin{tabular}{lc}\toprule
Method&Relative error (\%)\\\midrule
Laplacian&$23.00\pm4.11$\\
Velocity&$23.03\pm4.46$\\
NKB&$27.77\pm1.82$\\
\bottomrule\end{tabular}
\end{table}
At $K=1024$, Velocity and Laplacian have nearly identical mean reconstruction error, while NKB has higher error. Potential is assessed through $K=512$ on this ensemble. Its $K=512$ training takes 140.20 hours, compared with 3.99 hours for Velocity; the Potential and Laplacian mean errors differ by 0.17 percentage points. This cost--accuracy comparison motivates Velocity or Laplacian bases at larger capacities in this three-dimensional case.

\subsection{Training and boundary-to-basis costs}
\label{app:Pot-cost}
\begin{table}[!htb]\centering\scriptsize
\caption{Three-dimensional preparation costs in seconds. Training totals cover 2000 updates, scheduled validation and checkpoint I/O. Basis generation runs from the original parameterized boundary to all $1.6\,\mathrm{M}$ points and all $K$ fields, including streamed-output I/O; Laplacian also includes coefficient construction, interstage waits, setup, paging and background load. A dash means that the stage is not applicable. Superscript $*$ marks an estimate, not a completed run; projected runs assume uninterrupted execution. ``incl.'' means recovery is already included in the training total and was not timed separately; $\dagger$ marks these inclusive counters. Estimates and the measured short $K=1024$ endpoint are explained below. Matched Potential-FG rows use the same training protocol.}
\begin{tabular}{rlrrr}\toprule
$K$&Method&Normal training&State recovery&Boundary to all basis\\\midrule
32&Laplacian&\textemdash&\textemdash&3467.56\\
32&Velocity&3042.38&0.00&267.91\\
32&Potential&8170.08&200.31&255.50\\
32&Potential-AL&6316.24&0.00&259.52\\
32&Potential-FG (matched)&7332.09&0.00&291.39\\
64&Laplacian&\textemdash&\textemdash&2545.51\\
64&Velocity&$3930.13^{\dagger}$&incl.&307.32\\
64&Potential&14404.07&0.00&269.28\\
64&Potential-AL&$10513.53^{\dagger}$&incl.&273.65\\
64&Potential-FG (matched)&11924.79&0.00&346.78\\
128&Laplacian&\textemdash&\textemdash&3576.89\\
128&Velocity&$5361.59^{\dagger}$&incl.&295.34\\
128&Potential&26569.77&0.00&281.03\\
128&Potential-AL&18887.93&4036.51&281.98\\
128&Potential-FG (matched)&20862.81&0.00&451.21\\
256&Laplacian&\textemdash&\textemdash&5548.86\\
256&Velocity&$8622.13^{\dagger}$&incl.&310.61\\
256&Potential&59879.95&717.31&291.56\\
256&Potential-FG (seed 0)&$23238.19^{\dagger}$&incl.&469.61\\
256&Potential-FG (matched)&39108.35&0.00&665.15\\
512&Laplacian&\textemdash&\textemdash&9246.67\\
512&Velocity&$14366.09^{\dagger}$&incl.&488.29\\
512&Potential&504713.77&106.03&348.12\\
512&Potential-AL&$55800^{*}$&$0^{*}$&$370^{*}$\\
512&Potential-FG (matched)&$66000^{*}$&$0^{*}$&$1090^{*}$\\
1024&Laplacian&\textemdash&\textemdash&22049.01\\
1024&Velocity&$27976.58^{\dagger}$&incl.&747.68\\
1024&Potential&$866800^{*}$&$0^{*}$&$440^{*}$\\
1024&Potential-AL&$96400^{*}$&$0^{*}$&$480^{*}$\\
1024&Potential-FG (matched)&$115300^{*}$&$0^{*}$&$1940^{*}$\\
\bottomrule\end{tabular}
\label{tab:Pot-training-basis-cost}
\end{table}
Initial-field evaluation is separate from training and basis generation: Potential takes 655.76 seconds at $K=256$ and 1182.42 seconds at $K=512$. At $K=512$, the saved-basis assessment of 20 independent random fields takes 123.32 seconds. Potential training at $K=1024$ completes 25 updates, with zero scheduled validations, the original 2000-epoch learning-rate horizon, and a full model/Adam endpoint checkpoint.

At $K=1024$, Laplacian construction from raw geometry (marching-cubes resolution 80, mesh size $0.025$) through the full spectrum and coefficient archive takes 18769.83 seconds. The sampled-basis stage takes 477.10 seconds and writes all three FP64 velocity components for $1.6\,\mathrm{M}$ fluid nodes and 1024 modes (39,315,554,304 payload bytes). The raw-geometry-to-output interval is 22049.01 seconds, including interstage waits, setup, paging, I/O, and background load. CPU construction uses one BLAS thread, a 7 GiB resident limit, and an 8 GiB resident-peak guard; GPU basis output reuses the completed analytic geometry prefix.

\paragraph{Measured short run and cost estimates.}
The completed $K=1024$ Potential endpoint contains 25 optimizer updates and no scheduled validation; its 13267.46-second elapsed interval includes failed attempts, repairs, waits and atomic publication of the full model/Adam state. This short-run elapsed interval is separate from the projected full-training cost in Table~\ref{tab:Pot-training-basis-cost}.

The $K=1024$ Potential training estimate uses 22 measured updates (epochs 3--24; mean 430.68 seconds) and 166 ordinary updates from the completed $K=512$ run (mean 250.78 seconds). Their ratio scales the measured $K=512$ 2000-update total of 504713.77 seconds, retaining its validation and checkpoint allowance. The resulting estimate is about 866800 seconds (240.8 hours); validation and checkpoint costs are assumed to scale by the same factor as the update workload. Potential-AL training costs are extrapolated from measured $K=32,64,128$ runs and matched Potential-FG costs from $K=32,64,128,256$ runs using a least-squares power law in $K$. Basis-generation estimates use an affine fit in $K$ to the measured lower-capacity full-point intervals for the same method. Estimated training and basis times are rounded to 100 and 10 seconds, respectively. These projections assume the same implementation, hardware, sampling and output boundary, with no recovery event. They do not imply trained high-capacity checkpoints or measured reconstruction accuracy.

\subsection{Training geometry and frequency}
The full 2000-epoch Potential and three-dimensional velocity runs use the original sampling and optimization computations: 6144 query points and eight 4096-point validation geometries every 25 epochs (80 scheduled validations). The $K=1024$ Potential training run retains the original sampling and learning-rate horizon and completes 25 optimizer updates with zero scheduled validations. Potential retains the established canonical objective and two-pass causal Gram--Schmidt; its canonical physical actions use FP64, while ANN parameters and analytic derivatives use FP32. The three-dimensional velocity capacity family uses a fixed training domain, first-layer frequency 40, and cosine learning-rate endpoints scaled by 0.1 relative to the original velocity configuration. Its reported checkpoints complete 2000 epochs with the original sampling and scheduled validation. The Potential-AL models at $K=32,64,128$ use seed 20260823, frequencies 40/36, and cosine learning-rate endpoints $2.8\times10^{-5}$ and $10^{-7}$. The Potential-FG $K=256$ model instead uses seed 0, frequencies 160/36, and endpoints $2.8\times10^{-4}$ and $10^{-6}$; it is not a matched cross-method control. \hyperref[app:unseen-geometry]{Geometry-conditioned evaluation on other domains} is reported separately below.

\subsection{Capacity and unseen-geometry reconstruction}
\label{app:unseen-geometry}
\begin{table}[!htb]
\centering\small
\caption{Held-out velocity reconstruction on disjoint 16,384-point assessment sets. Fixed-domain learned checkpoints complete 2000 epochs. Numerical bases use the same capacity, fit and assessment point sets, and target fields. Geometries 2 and 3 use radial Piola transport; geometry 4 uses a globally qualified quadratic approximation of its original smooth boundary and consistent physical finite-element matrices.}
\begin{tabular}{rlrr}\toprule
$K$ & Geometry & Learned & Numerical\\\midrule
128 &Training geometry, independent points&.398468&.346917\\
256 &Training geometry, independent points&.281388&.241861\\
256 &Unseen geometry 2&.751690&.222570\\
256 &Unseen geometry 3&.744673&.254720\\
256 &Unseen geometry 4&.770144&.230048\\\bottomrule
\end{tabular}
\end{table}
The geometry-2 numerical basis uses a fitted mesh with 80,890 vertices and 354,108 tetrahedra, the relative-potential finite-element method (FEM) with consistent mass, and radial Piola transport onto the smooth target domain. Geometry 3 uses the same construction on its own fitted mesh. Each reconstruction uses all 256 positive modes and the identical 16,384 fit points, disjoint assessment points, and target velocities. Solution and reconstruction take 687.17 and 1061.51 seconds for geometries 2 and 3, after mesh and point-map preparation. Geometry 4 retains all 320,815 tetrahedra and uses a quadratic physical N\'ed\'elec chart: positive cell determinants, the entire 95,296-face boundary homotopy, and the boundary degree qualify its global map. This chart approximates the original smooth geometry. Its consistent stiffness/mass eigensolve obtains 264 Ritz vectors with maximum relative residual $5.04\times10^{-11}$; the reconstruction retains 256 positive modes. Assembly, solve, and evaluation total 1054.85 seconds. The disjoint error is .230048, while the fit-point error is .227943. All three learned transfer measurements use the unchanged frozen checkpoint.

\section{Experimental protocols and measurements}
\label{app:protocol}
\subsection{Shared initialization, error, and timing definitions}
Each autonomous trajectory begins with a full-mass projection of the prescribed initial velocity and evolves all basis coordinates. Writing $u_{\rm ref}=u_{\rm GT}$, reference states supply subsequent error measurements, $\norm{u-u_{\rm ref}}/\norm{u_{\rm ref}}$, on a common assessment domain. Numerical controls use the same mode count wherever reported. Training uses analytic derivatives and native CuPy; operator contraction and reduced evolution use 64-bit floating-point arithmetic (FP64) where specified; FP32 denotes 32-bit floating-point arithmetic. Timings are measured wall times with training, assembly, integration, and output reported as separate stages.

\begin{table}[!htb]
\centering\small
\caption{Primary comparison protocols. DNS denotes direct numerical simulation. The multi-hole P1 and analytic high-capacity stages are separate discretizations.}
\begin{tabular}{llll}\toprule
Case & Capacity & Reference / assessment & Autonomous horizon\\\midrule
Periodic &1024&$256^2$ DNS / $128^2$ observations&241 frames, $T=18\pi$\\
Multi-hole P1 &32&Common circulation-aware P1 mesh&51 frames, $T=1$\\
Multi-hole analytic &1024&Finest level-2 reference, common points&51 frames, $T=1$\\
Metaball &128--1024&$128^3$, $1.6\,\mathrm{M}$ fluid nodes&30 frames, $T=1.16$\\\bottomrule
\end{tabular}
\end{table}

\subsection{Periodic flow protocol}
The periodic domain is $[0,2\pi)^2$. The vortex separation is $d=1.35$, core size $\sigma=\pi/10$, circulation $\pi/5$ per vortex, and viscosity $\nu=0.0010742914$. The numerical reference, same-capacity Laplacian control, and learned trajectories share 241 observation times through $T=18\pi$. Full mass and curl-stiffness matrices and matrix-free rotational loads use $256^2$ quadrature points; errors are measured on the same uniform $128^2$ observation grid. With $k=(k_x,k_y)$ the integer Fourier wavevector, Fourier uses 512 conjugate wave pairs ordered by $(|k|^2,k_x,k_y)$, with complete sine/cosine pairs and zero constant-velocity coefficients. The learned spectral models retain frozen epoch-2000 checkpoints and analytic spatial derivatives. Initial projection uses the physical velocity alone; subsequent reference states do not enter the evolution. \docref{Figure}{fig:case1-field} shows saved states 1, 121, and 241. Energy and enstrophy use the same initial reference normalization for all methods.
All nine learned full-mass matrices have rank 1024 without ridge or mode deletion. Each frozen span is advanced at maximum steps $0.005$ and $0.0025$; the largest relative mass-norm trajectory difference is $4.412\times10^{-12}$. An independent implementation check against the inherited exponential time-differencing fourth-order Runge--Kutta (ETDRK4) equations {\citep{cox2002exponential,kassam2005fourth}} precedes each trajectory.

\subsection{Topology allocation and harmonic conditioning}
The random large-vortex ensemble combines three to seven smooth Wendland $C^4$ vortices, where $C^k$ denotes fields with continuous derivatives through order $k$. Centers, signs, and amplitudes are random; support radii lie in $[0.65,0.95]$. Random obstacle-boundary streamfunction constants complete each initial field. Their direction is sampled from a normalized Gaussian vector, and the energy of its harmonic lift relative to the zero-trace vortical component is drawn uniformly from $[0,0.30]$. A common rescaling sets peak speed to $0.5$. All bases project the same resulting velocity field. The low-capacity comparison uses three training seeds, 20 such initial states per geometry, and the same circulation-aware P1 NS reference for every basis. Each total-$32$ completed basis contains $32-h$ positive fields and $h$ harmonic fields. The displayed fields use seed 1, state 00. The high-capacity analytic comparison retains 1024 coordinates, the same 20 states, and both time steps $0.002$ and $0.001$; the finest common assessment contains approximately 4.19 million points per domain.

\begin{table}[!htb]\centering\small
\caption{Harmonic completion for random large-vortex flows at total $K=32$. Final relative velocity error (\%) over the same 20 initial states per geometry; learned entries average three training seeds. Completed bases replace $h$ positive modes with $h$ harmonic modes. These are the low-capacity results summarized in \docref{Section}{sec:case2}.}
\label{tab:prescribed32-summary}
\begin{tabular}{rrrr}\toprule
{Holes}&{Potential $P$}&{Potential $PH$}&{Numerical $PH$}\\\midrule
{1}&{53.23}&{40.04}&{39.76}\\
{2}&{52.99}&{40.49}&{40.13}\\
{3}&{54.59}&{40.49}&{39.72}\\\bottomrule
\end{tabular}\end{table}

The high-capacity raw/scaled harmonic pairs fix the positive checkpoint, seed, optimizer initialization, trace, correction envelope, Dirichlet objective, and 2000-epoch budget. For two holes, initial/final errors change from $12.709\%/28.486\%$ to $6.591\%/6.692\%$; for three holes, from $16.069\%/25.831\%$ to $8.503\%/9.066\%$. The single-hole raw completed basis gives $5.739\%$ final error. Maximum final time-step differences are $7.00\times10^{-7}$ for single-hole raw, $3.68\times10^{-7}$ for two-hole scaled, and $3.40\times10^{-7}$ for three-hole scaled. The common level-2 reference fixes the spatial assessment for these temporal comparisons.

For basis streamfunctions $\psi_i$ with vorticities $\omega_i=-\Delta\psi_i$, velocity-basis integration uses the FP64 weak tensor $Q_{ijk}=\sum_q w_q[(\partial_x\psi_i)(\partial_y\psi_j)-(\partial_y\psi_i)(\partial_x\psi_j)]\omega_k$, assembled from all $1{,}048{,}576$ quadrature points. Its $1024^3$ entries occupy 8 GiB on the RTX 4080 SUPER, alongside the coefficient states and linear operators. For the one-hole case, assembly takes 57.9 minutes including checkpoint I/O. Subsequent steps average 0.393 seconds, compared with 153.8 seconds for streamed full-field contraction of the same Galerkin operator: approximately $391\times$ faster stepping, with assembly accounted for separately. Assembled and factored nonlinear contractions agree to $3.37\times10^{-14}$ relative error over the 60 initial coefficient columns.

\subsection{Direct random-streamfunction ensemble}
\label{app:random-streamfunction}
The ensemble contains 20 fixed fields per topology, indexed 0--19. The \hyperref[tab:random-historical-capacity]{numerical capacity study} and \hyperref[tab:random-paired]{learned/numerical $K=32$ allocation ablation} use this ensemble. The all-method $K=1024$ dynamic comparison uses fields 0 and 1 in each topology, with six independently qualified references. Each reported comparison uses identical physical target fields across methods and completion arms.

\subsubsection{Sampling and boundary conditioning}
We draw independent standard-normal scalar values on a $25\times25$ lattice spanning $[-1,11]\times[-6,6]$, with spacing $0.5$, and interpolate them with an exact tensor-product quintic B-spline {\citep{deboor1978splines}}. The seed hierarchy is $[20260928,\text{topology index},\text{field index}]$, with topology indices 0, 1, and 2. The spline is $C^4$ before boundary conditioning, and the final streamfunction is at least $C^2$.

Let $s$ be the raw spline and $q_j$ a smooth nonnegative defining function for boundary component $j$. The boundary indices $i,j,k$ run from $0$ to $h$. The outer component has $q_0=1-((x-5)/4)^{12}-(y/4)^{12}$; the obstacles use the original circle, rounded-square, and smoothed-triangle defining functions. For each component, $c_j$ is the arc-length-weighted mean of $s$ over 1536 boundary points. We construct bounded defining functions $\sigma_j$, partition weights $w_j$, logarithmic functions $\ell_j$, squared distance estimates $d_j^2$, and localized blending functions $b_j$. These give the boundary-adjusted field $p$, vorticity-correction weight $a$, and final streamfunction $\psi$:
\begin{align}
\sigma_j&=q_j/\sqrt{1+q_j^2},&
w_j&=\frac{\prod_{i\ne j}\sigma_i}{\sum_k\prod_{i\ne k}\sigma_i},\\
\ell_j&=\log(1+q_j),&
d_j^2&=\frac{\ell_j^2}{|\nabla\ell_j|^2+(\ell_j/4)^2+10^{-30}},\\
b_j&=w_j^3\exp\!\left[-\left(d_j^2/0.35^2\right)^2\right],&
p&=s+\sum_j b_j(c_j-s),\\
a&=\tfrac12\sum_j d_j^2b_j,&
\psi&=p-a\Delta p.
\end{align}
This construction retains a separate constant on each boundary component and enforces the original impermeability and zero-boundary-vorticity conditions. Velocity and vorticity follow from analytic derivatives, $u=(\partial_y\psi,-\partial_x\psi)$ and $\omega=-\Delta\psi$. One global scale sets the maximum sampled speed to $0.5$ on the original $321\times321$ mask and boundary loops; this scale is fixed under all numerical refinements. The distribution, raw lattice values, boundary constants, and scales define each reproducible draw; harmonic content is measured from the resulting velocity field.

\subsubsection{Captured energy, reference coordinates, and independent samples}
Assessment uses arrays of $1{,}048{,}576$, $4{,}194{,}304$, and $4{,}194{,}304$ points. Their common all-method masks contain $1{,}048{,}092$, $4{,}192{,}289$, and $4{,}192{,}265$ points. These are quadrature/assessment locations. The \hyperref[tab:random-historical-capacity]{state dimensions below} are counted from continuous piecewise-cubic (P3) reference finite-element spaces.

Using the common velocity assessment matrix $W$, let $M=U^{\top}WU$ be the mass matrix, $b=U^{\top}Wu_0$ the initial load vector, and $n=u_0^{\top}Wu_0$ the squared target norm. The full-mass initial projection gives
\begin{equation}
Mc_0=b,\qquad
\eta=\frac{c_0^{\top}Mc_0}{n}=\frac{\tfrac12\|Uc_0\|_W^2}{\tfrac12\|u_0\|_W^2}.
\end{equation}
Here $\eta$ is the fraction of initial kinetic energy represented by the projected field. Projection orthogonality gives $\eta=1-e_0^2$ for each field, where $e_0=\|Uc_0-u_0\|_W/\|u_0\|_W$. We compute the quadratic energy separately for every field and training seed before averaging. An overbar on scalar statistics denotes the ensemble mean. The mean is therefore $\overline\eta=1-\overline{e_0^2}$; no transformation of the ensemble-mean error is used. All coordinates are retained without a ridge term. Across the 840 saved $K=32$ and 120 numerical $K=1024$ projections, direct quadratic-energy and load-energy ratios agree within $1.3\times10^{-13}$.

The one-, two-, and three-hole P3 meshes have $881{,}793$, $669{,}533$, and $1{,}390{,}768$ scalar nodes. Their boundary nodes number $9{,}216$, $13{,}458$, and $22{,}428$. With constant streamfunction trace on each component and the outer constant fixed, writing $N_{\rm P3}$ for the scalar node count and $N_\partial$ for the boundary node count, the velocity-state dimension is
\begin{equation}
N_{\rm ref}=N_{\rm P3}-N_{\partial}+h.
\end{equation}
Equivalently, the reference evolves interior vorticity values and $h$ circulations; streamfunction is reconstructed from that state. This gives $N_{\rm ref}=872{,}578$, $656{,}077$, and $1{,}368{,}343$. The counts in \docref{Table}{tab:random-historical-capacity} define the state-coordinate compression ratio $N_{\rm ref}/K$, visualized in \docref{Figure}{fig:random-k1024-energy}. The same mesh dimensions apply to both reduced capacities in each domain. Initial representation checks cover all 60 P3 fields. The six predetermined fields in \docref{Appendix}{app:random-k1024-six} additionally have full-trajectory refinement measurements.

\begin{table}[!htb]\centering\small
\caption{Initial kinetic-energy capture and reference state dimensions for 20 direct random-streamfunction fields per topology. P3 nodes count the reference mesh; $N_{\rm ref}$ counts independent interior scalar coordinates plus $h$ circulations after boundary constraints and outer-gauge fixing. $K$ is the reduced state dimension, and $\eta_P,\eta_{PH}$ are mean captured-energy percentages for numerical $P_K$ and $P_{K-h}+H_h$. Energy fractions are evaluated per field before averaging. The common assessment measure, basis discretizations, and initial reference qualification are specified in \docref{Appendix}{app:random-streamfunction}.}
\label{tab:random-historical-capacity}
\setlength{\tabcolsep}{4pt}
\begin{tabular}{rrrrrr}\toprule
Holes & P3 nodes & $N_{\rm ref}$ & $K$ & $\eta_P$ (\%) & $\eta_{PH}$ (\%)\\\midrule
1&881,793&872,578&32&2.27&2.13\\
1&881,793&872,578&1024&75.32&75.34\\
2&669,533&656,077&32&2.37&2.16\\
2&669,533&656,077&1024&79.66&79.73\\
3&1,390,768&1,368,343&32&1.79&1.60\\
3&1,390,768&1,368,343&1024&54.48&54.52\\
\bottomrule\end{tabular}\end{table}

\subsubsection{Optional combined energy--compression summary}
\label{app:basis-quality}
For a fixed reference discretization, captured energy and compression in \docref{Figure}{fig:random-k1024-energy} define the energy--compression product $Q_B$. Let $E_{\rm projected}$ and $E_{\rm target}$ denote initial projected and target kinetic energies; then
\begin{equation}
Q_B=\eta\frac{N_{\rm ref}}{K}
=\frac{E_{\rm projected}/K}{E_{\rm target}/N_{\rm ref}}.
\label{eq:basis-energy-quality}
\end{equation}
For field $i$, $Q_{B,i}=(N_{\rm ref}/K)\eta_i$, so $\overline Q_B=(N_{\rm ref}/K)\overline\eta$. This is a derived summary of energy capture and compression. At fixed $N_{\rm ref}$ and $K$ it orders methods exactly as $\eta$ does. Its reference normalization changes under mesh refinement. We therefore report the independently interpretable energy fractions and compression ratios in the \hyperref[fig:random-k1024-energy]{main text}, and retain the \hyperref[tab:basis-quality]{combined values here} for a common reference convention.
\begin{table}[!htb]\centering\small
\caption{Derived energy--compression products for the numerical representations in \docref{Table}{tab:random-historical-capacity}. Energy fractions are evaluated per field before averaging; all 20 draws per topology are retained.}
\label{tab:basis-quality}
\begin{tabular}{rrrr}\toprule
Holes&$K$&$\overline Q_{B,P}$&$\overline Q_{B,PH}$\\\midrule
1&32&618.62&581.29\\
1&1024&641.84&641.99\\
2&32&485.62&442.04\\
2&1024&510.41&510.82\\
3&32&767.13&685.73\\
3&1024&728.05&728.60\\
\bottomrule\end{tabular}\end{table}

The $K=32$ bases use their original common P1 embeddings. The numerical $K=1024$ bases use their original level-2 P1 representations. Both are assessed against the same continuous analytic initial fields on the unchanged common masks. The capacity table compares these existing discretized representations; their basis meshes differ. All 20 fields are retained at each capacity.

The Laplacian basis contributes one numerical realization; potential and velocity training each contribute three existing seeds at $K=32$. For each learned method, captured-energy fractions are averaged across seeds within a field, then over the 20 independent fields. \docref{Table}{tab:random-paired} reports mean capture and the paired gain $\Delta\eta=\eta_{PH}-\eta_P$ in percentage points. This equal-total-$K$ comparison exchanges $h$ positive coordinates for $h$ harmonic coordinates.

\begin{table}[!htb]\centering\small
\caption{Initial kinetic-energy capture at total $K=32$, in percent. $\Delta\eta$ is the mean captured-energy change in percentage points. Positive values favor completion. Learned energy fractions average three seeds within each of the 20 independent fields; the last column counts fields with higher $PH$ capture.}
\label{tab:random-paired}
\begin{tabular}{rlrrrr}\toprule
Holes & Basis & $\eta_P$ (\%) & $\eta_{PH}$ (\%) & $\Delta\eta$ & Improved\\\midrule
1&Laplacian&2.269&2.132&-0.137&2/20\\
1&Potential&2.252&2.178&-0.075&5/20\\
1&Velocity&2.269&2.145&-0.124&2/20\\
2&Laplacian&2.369&2.156&-0.213&3/20\\
2&Potential&2.370&2.174&-0.197&3/20\\
2&Velocity&2.345&2.216&-0.129&3/20\\
3&Laplacian&1.794&1.604&-0.190&3/20\\
3&Potential&1.795&1.633&-0.163&2/20\\
3&Velocity&2.009&1.814&-0.195&2/20\\
\bottomrule\end{tabular}\end{table}

\begin{samepage}
\paragraph{Harmonic-energy diagnostic.}
\begingroup
The 60 directly sampled fields predominantly occupy the vortical sector. Projection onto the numerical harmonic span gives median initial energy fractions of $0.01016\%$, $0.06411\%$, and $0.04385\%$ for one, two, and three holes, with maxima of $0.11744\%$, $0.17651\%$, and $0.27178\%$. Level-2/level-3 harmonic refinement changes these fractions by at most $0.00214$ percentage points. At fixed $K$, mass-orthogonal sectors trade the harmonic projection energy gained against the energy of the $h$ displaced positive modes. The coordinate allocation is therefore governed by the relative occupation of these two sectors. Positive modes are ordered by spectral Rayleigh values; a field's captured energy depends on its modal coefficients. Appending $H_h$ to unchanged $P_K$ creates a nested span, whose best-projection error is nonincreasing.
\par\endgroup
\par
\end{samepage}

\subsubsection{Autonomous trajectories for the first one- and two-hole draws}
\label{app:random-dynamics}
The $K=32$ comparison uses the first fixed draw (index 0) in each of the one- and two-hole domains. Each method starts from the analytic-field full-mass projection and evolves autonomously with the reused Galerkin operators, $\nu=1/750$, $T=1$, and 51 saved frames. The two step sizes are $0.002$ and $0.001$. Subsequent reference values enter error assessment, not the coefficient evolution. \docref{Table}{tab:random-dynamics} summarizes the fine-step errors. {With $\bar e=\int_0^1e(t)\,\dd t$ the time-averaged relative error, we define the improvement $\Delta(t)=e_P(t)-e_{PH}(t)$ and $\bar\Delta=\bar e_P-\bar e_{PH}$; positive values mean lower error after completion.}

The numerical reference uses continuous cubic finite elements and the fine time step. Maximum coarse/fine relative velocity differences are $2.53042\times10^{-6}$ and $5.03966\times10^{-5}$. Full-trajectory quadratic/cubic comparisons give $0.84444\%$ and $0.72543\%$ differences; independent boundary-geometry refinement gives $0.02050\%$ and $0.07681\%$ across all 51 frames on both original masks. These checks retain the same analytic initial field. Direct-field and moment-based error evaluation agree to below $6.10\times10^{-16}$; the independent polynomial point-adapter check is below $3.33\times10^{-13}$. This table reports the two $K=32$ dynamic cases; the \hyperref[app:random-k1024-six]{complete six-field $K=1024$ assessment} is reported next.

\begin{table}[!htb]\centering\small
\caption{Random-field dynamics for index 0 in each of two domains. $\bar e=\int_0^1e_u(t)\,\mathrm dt$ is evaluated by the trapezoidal rule over 51 frames. Errors are percentages and differences are percentage points. Learned entries average the three training seeds of the same physical field.}
\label{tab:random-dynamics}
\begin{tabular}{rlrrrr}\toprule
Holes&Basis&$\bar e_P$&$\bar e_{PH}$&$\bar\Delta$&$\Delta(T)$\\\midrule
1&Laplacian&99.28325&99.29712&-0.01387&-0.01336\\
1&Potential&99.26077&99.26889&-0.00812&-0.00733\\
1&Velocity&99.27255&99.29087&-0.01832&-0.01817\\
2&Laplacian&98.75970&98.79192&-0.03221&-0.02997\\
2&Potential&98.79390&98.84504&-0.05114&-0.05093\\
2&Velocity&98.81309&98.91302&-0.09993&-0.10682\\
\bottomrule\end{tabular}\end{table}

\subsubsection{Six-field \texorpdfstring{$K=1024$}{K=1024} comparison}
\label{app:random-k1024-six}
All methods use the fixed indices 0 and 1, full-rank mass projection of the analytic initial field, and all 1024 coordinates. Existing learned checkpoints, numerical operators, and NKB models are reused. Reduced trajectories were advanced as 20-field batches; this assessment selects the two predetermined fields without changing their coefficients. Coarse/fine steps are $0.002/0.001$ with 51 saved frames on $[0,1]$ and $\nu=1/750$. No later reference state enters autonomous integration. Initial-energy fractions use the assessment mass, which is distinct from the operator quadrature mass.

Reference refinement tests compare the same physical fields under time-step, polynomial-space, and boundary-geometry changes on both original masks. For the three-hole fields, the largest time-step differences are $0.02414\%$ and $0.03835\%$; polynomial-space differences are $0.49879\%$ and $0.92725\%$; geometry differences are $0.10612\%$ and $0.17219\%$. The initial physical fields and evaluation masks remain fixed under all refinements. The reported relative velocity errors use the original common all-method masks; analytic-only supplementary points are reported separately where the representation is defined.

\docref{Figure}{fig:random-k1024-energy} gives initial capture and compression, and \docref{Figure}{fig:random-k1024-errors} shows all autonomous error curves. The paired allocation gains are reported in \docref{Table}{tab:completion-improvement}.

\begin{figure}[!htb]\centering
\includegraphics[width=\linewidth,height=.72\textheight,keepaspectratio]{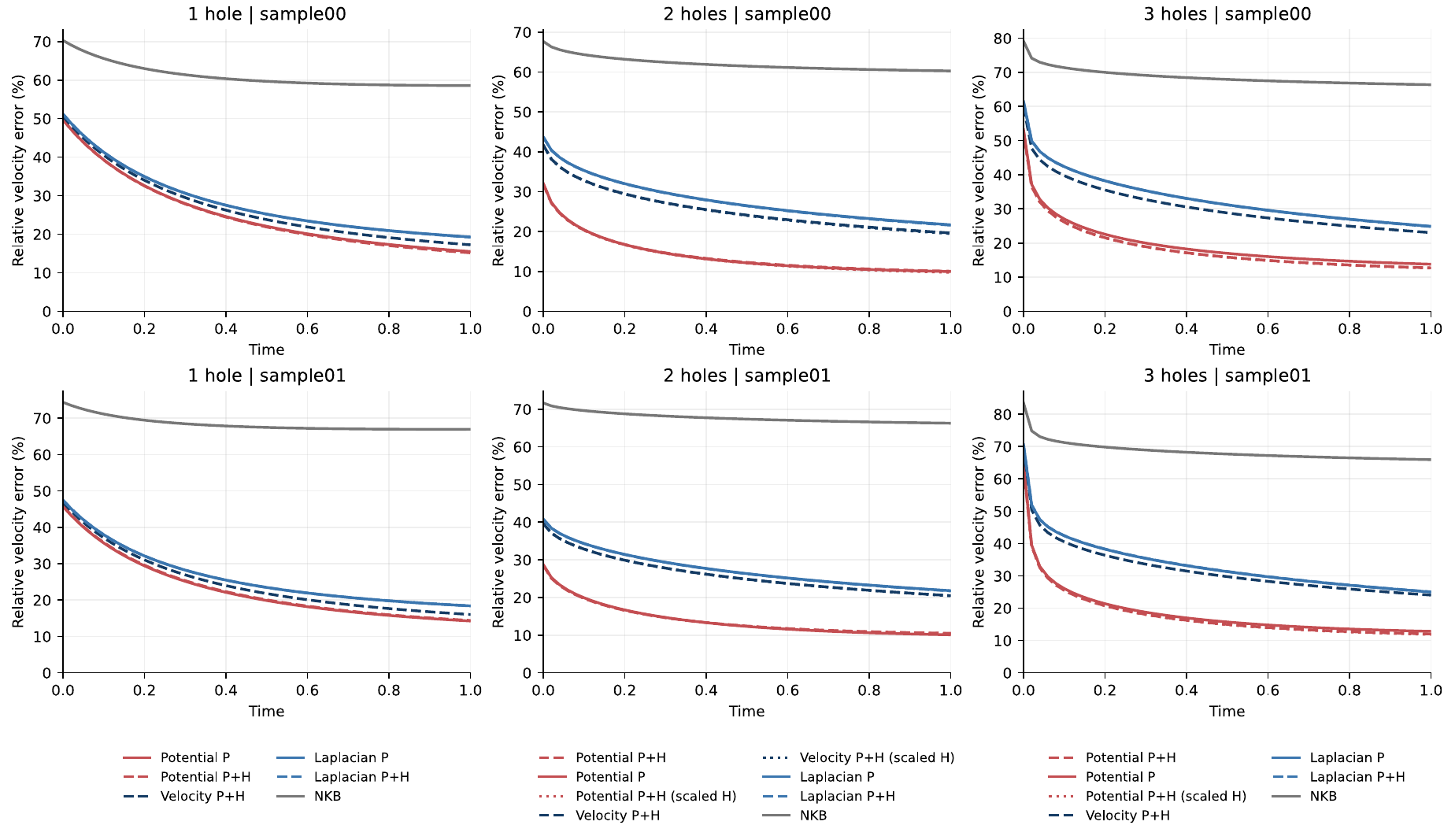}
\caption{All six autonomous random-field comparisons at $K=1024$ against qualified references. Each panel is one predetermined physical field, and every curve uses all 51 saved times and the original common mask. The numerical control has the same total coordinate count and the corresponding P or P+H allocation.}
\label{fig:random-k1024-errors}
\end{figure}

For a saved coefficient vector, let $M_{\rm assessment}=U^{\top}WU$ denote the full mass matrix on assessment nodes. Its quadratic energy is $c^{\top}M_{\rm assessment}c$ (twice the kinetic energy). Explicit harmonic components are reported both through their own coefficients and as the mass-orthogonal projection of the full velocity onto each method's harmonic span. Let $H$ index the harmonic coordinates, $M_{HH}$ be their mass block, and $\ell_H$ the harmonic entries of the full velocity load. The latter quadratic energy is $\ell_H^{\top}M_{HH}^{-1}\ell_H$, with $\ell_H=(M_{\rm assessment}c)_H$, retaining P/H cross terms. Learned harmonic fields use their original learned representation throughout.

\begin{figure}[htbp]\centering
\includegraphics[width=\linewidth,height=.72\textheight,keepaspectratio]{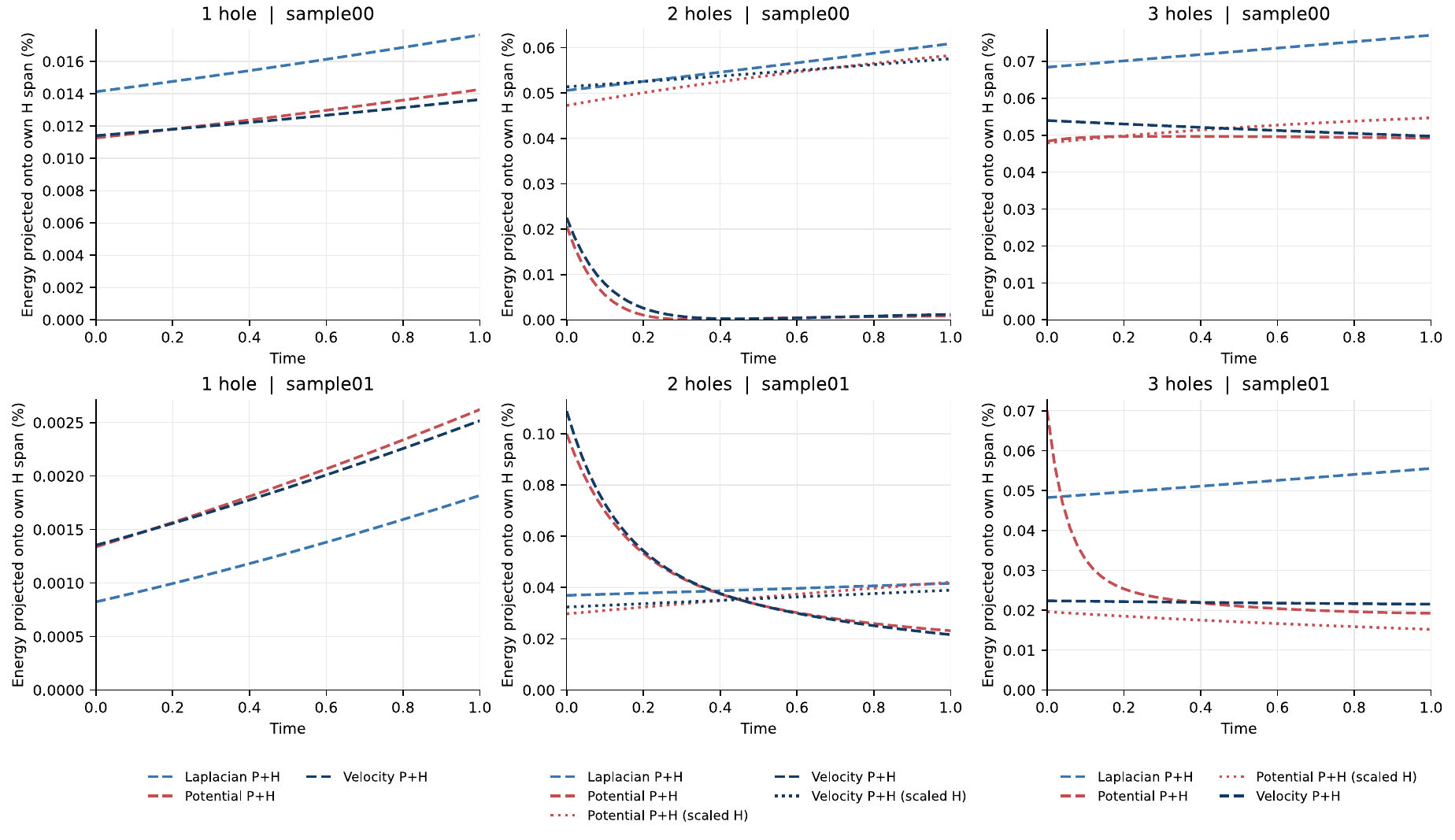}
\caption{Energy projected onto each representation's explicit harmonic span as a fraction of its instantaneous total energy. The projection measures occupation of each representation's explicit circulation coordinates.}
\label{fig:random-k1024-harmonic}
\end{figure}

Supplementary boundary diagnostics use exact edge integration of the original P1 polygon velocities for numerical bases, and nested 1024/2048-point analytic-boundary evaluations for neural bases. They report normal speed and counterclockwise circulation around each enclosed component over all 51 frames. Curl-form neural velocities have cancelling analytic mixed derivatives; NKB divergence is evaluated from its analytic direct-velocity Jacobian at 1024 deterministic original common-mask points. These component diagnostics complement the full-domain error measurements. The released results include all per-component circulation curves, quadrature-refinement differences, and energy histories.

The timing ledger separates reference generation and qualification, cache/operator recovery, projection, batched evolution, and error assessment. Existing 20-field batch costs remain associated with those batches. Stage and process timers retain their respective workload boundaries. Supplementary boundary diagnostics required 390.59 seconds across their measured stages, independently of reference generation and final full-domain error assessment.

\subsection{Three-dimensional dynamics and error attribution}
\label{app:3d-error}
The $128^3$ reference uses rotational NS with viscosity $0.00025$ and normal penalization. The reference solver uses {classical RK4 advection \citep{butcher2008numerical}} with split viscosity and normal-penalty updates at step $0.0025$. Reduced trial fields are zero-extended in the solid and integrated with the projected, unsplit grid generator at $0.00125$ and $0.000625$. Numerical fields use the mapped fitted N\'ed\'elec basis {\citep{nedelec1980mixed}}; learned velocities use analytic curls. All fluid nodes enter the mass matrix, projections, and error measurements. The \hyperref[app:coefficient-acceleration]{coefficient timing study below} additionally evaluates the NS reference at the matched step $0.000625$.

The velocity capacity study retains 6144 training points per epoch, eight validation geometries with 4096 points each every 25 epochs, and all 2000 epochs. Training takes 5361.59 and 8622.13 seconds at $K=128,256$. The learned final/initial energy ratios are $0.496797,0.505110$; numerical ratios are $0.410710,0.396369$. Maximum relative mass-norm time-step differences are below $5.3\times10^{-13}$.

For frozen $U$, let $c_\star(t)$ minimize $\norm{Uc-u_{\rm ref}(t)}$. Projection orthogonality gives
\begin{equation}
\frac{\norm{Uc(t)-u_{\rm ref}}^2}{\norm{u_{\rm ref}}^2}
=\frac{\norm{Uc_\star-u_{\rm ref}}^2}{\norm{u_{\rm ref}}^2}
+\frac{\norm{U(c(t)-c_\star)}^2}{\norm{u_{\rm ref}}^2}.
\label{eq:error-split}
\end{equation}
The representation term decreases with capacity, whereas the in-span term stays close to one half of the reference norm at $K=512,1024$. Its share of the total squared error consequently increases. The decomposition is evaluated after autonomous integration using the saved trajectories and common reference loads. At $K=128,256$, the numerical best-projection errors are $0.659617,0.547576$ and in-span errors are $0.378547,0.414556$, respectively. These controls show the same separation between improving the spatial span and changing its autonomous evolution.

\begin{table}[!htb]\centering\small
\caption{Final orthogonal error decomposition for the velocity-trained Eigenfluid spans, $T=1.16$. Entries are relative velocity errors (\%), normalized by the common full-fluid NS state. Representation and in-span terms add in square, as in \doceqref{eq:error-split}. The $K=512,1024$ entries extend the decomposition using saved full-mass matrices, all 30 reference loads, and autonomous coefficients; their final total errors also agree with direct saved-field differences.}
\label{tab:3d-error-attribution}
\begin{tabular}{rrr}\toprule
$K$ & Best representation & In-span evolution\\\midrule
128 & 67.0459 & 44.5447\\
256 & 55.5999 & 45.6442\\
512 & 43.9777 & 49.7599\\
1024 & 35.0357 & 49.6553\\
\bottomrule\end{tabular}\end{table}

\section{First eight modes: a basis atlas}
\label{app:modes}
{The periodic and three-dimensional atlases complement the multiply connected comparison in \docref{Figure}{fig:multihole-modes}. Rows identify the construction and columns show its first eight displayed modes, with ordering and representation specified in each caption. Signs and orthogonal rotations within repeated eigenspaces preserve the represented span. Mean-zero scalar displays and per-mode amplitude normalization expose spatial structure while leaving the ROM coordinates unchanged.}
\begin{figure}[!htb]
\centering\includegraphics[width=\linewidth,height=.72\textheight,keepaspectratio]{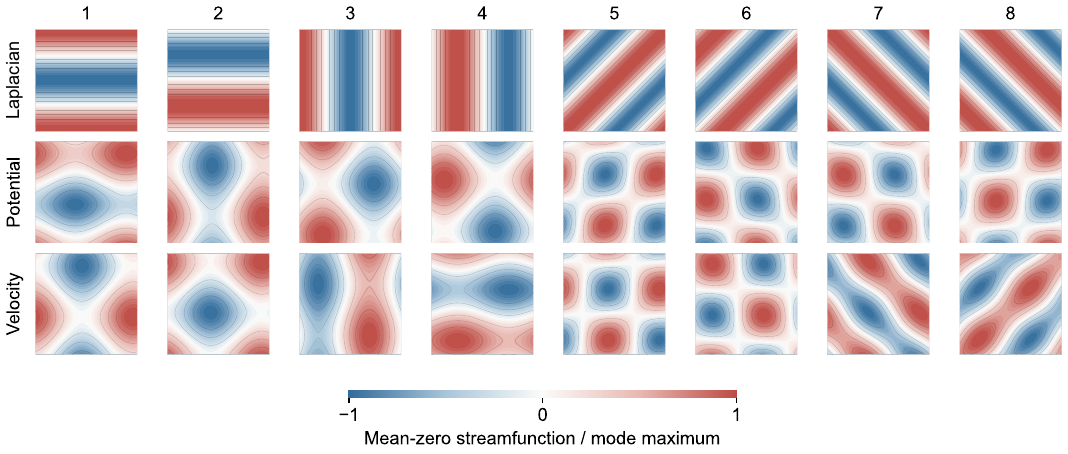}
\caption{\textbf{First eight periodic modes.} Rows compare analytic Fourier modes and generalized Ritz modes from the potential- and velocity-trained $K=1024$ spans; columns show modes 1--8. Potential Ritz fields use centered potential mass and gradient energy, while velocity Ritz fields use velocity mass and curl energy. Color shows mean-zero streamfunction with per-mode normalization; all 1024 channels contribute to each learned Ritz field.}
\label{fig:modes1}
\end{figure}

\begin{figure}[!htb]
\centering\includegraphics[width=\linewidth,height=.72\textheight,keepaspectratio]{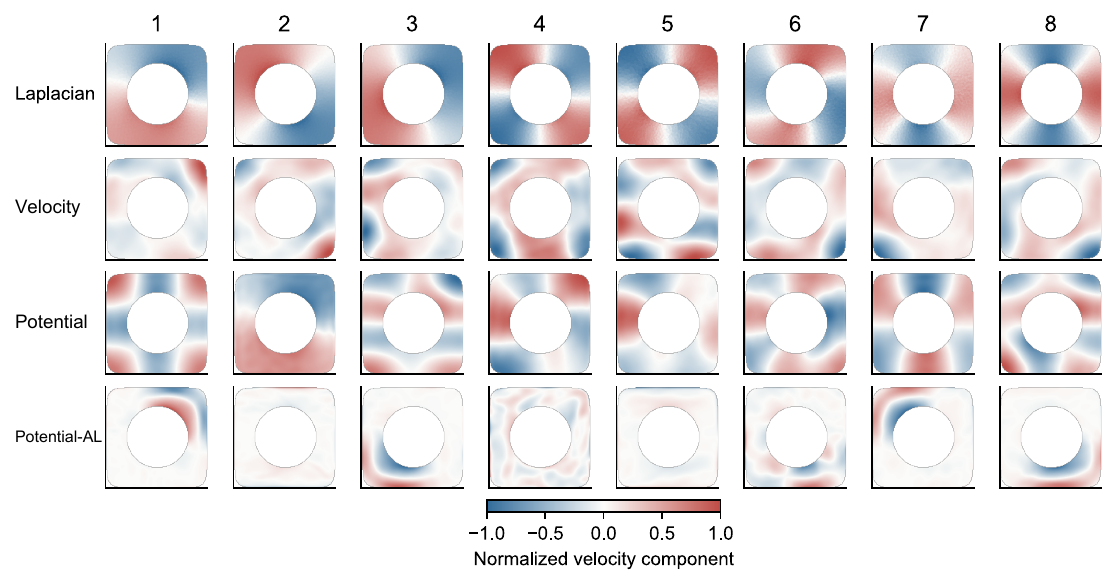}
\caption{\textbf{First eight velocity fields at common $K=32$.} Rows compare Laplacian, Velocity, Potential, and Potential-AL; columns show modes 1--8. All methods use the same $z=0.50390625$ slice and per-mode maximum-magnitude normalization of $u_z$. Laplacian fields follow numerical eigenvalue order; learned fields follow their saved channel order, without a display-only Ritz rotation.}
\label{fig:modes3}
\end{figure}

\section{Neural kinematic basis comparison}
\label{app:nkb}
We compare the learned spectral spans with a direct-velocity neural kinematic basis (NKB), following the camera-ready formulation of \citet{liu2025neural}. Each NKB model is trained independently at the reported capacity on the original case geometry. Its initial coefficients are obtained from the same physical initial velocity through a full-mass least-squares projection. The trajectory comparisons evolve all coordinates autonomously; subsequent reference snapshots are used only for assessment.

\paragraph{Architecture and optimization.}
Our native-CuPy reimplementation uses seven hidden layers of width 256 with LeakyReLU slope $0.01$, followed by an exponential linear unit (ELU) velocity output; LeakyReLU is the leaky rectified linear unit. It optimizes the six NKB penalties with weights $0.01$, $5$, $100$, $30$, $100$, and $100$ for spatial variation, divergence, orthogonality, boundary alignment, mean length, and small length, respectively. The target mean basis length is $0.37$, and the small-length threshold is $0.05$. We realize the orthogonality penalty as the mean squared off-diagonal entry of the global velocity Gram matrix, a nonnegative implementation of the paper's printed signed sum. Spatial Jacobians and their parameter adjoints are evaluated analytically. The periodic model uses sine/cosine coordinates; bounded domains retain their original implicit surfaces and use a soft boundary band of width $0.05$ times the domain side length. Direct velocity outputs use the stated soft kinematic penalties; Eigenfluid uses differential lifts, spectral objectives, and explicit harmonic coordinates.

Each model receives 2000 Adam updates {\citep{kingma2015adam}}, using 4096 points per update in 2D and 6144 in 3D. The learning rate is $5\times10^{-4}\,0.96^{10(s-1)/2000}$ at update $s$, mapping the original ten-epoch decay over this fixed-case budget. Validation uses 4096 points every 25 updates, including the eight original validation geometries in 3D. The common 2000-update protocol compares fixed-case spatial bases, with autonomous assessment at the trajectory capacities. The original NKB geometry-collection experiment uses one million updates; here its shared network architecture and penalties are trained independently for each reported case.

\paragraph{Common evolution and assessment.}
We use the full velocity mass and all output modes, without mode truncation or ridge regularization. The 3D comparison projects the unchanged $128^3$ NS generator, including its Leray projection and normal penalty, onto the NKB span. The periodic and multi-hole comparisons use the same algebraic curl-form operators and integration schemes as the spectral-basis experiments. NKB controls incompressibility and boundary tangency through its penalties. The common Galerkin assessment measures how this spatial representation supports the same reduced evolution used for the spectral bases. The original NKB solver couples its basis to semi-Lagrangian advection and reprojection.

Periodic errors use the common $128^2$ observation grid and all 241 frames through $T=18\pi$, with vortex separation $d=1.35$. NKB operator quadrature uses $256^2$ points. The multi-hole assessment uses all original common-domain points, the level-2 reference, 20 initial states, and 51 frames to $T=1$. The 3D assessment retains all $1.6\,\mathrm{M}$ fluid nodes and 30 frames to $T=1.16$. Training, offline assembly, and online integration are separate costs. \docref{Table}{tab:nkb-matched} gives the 2D initial projection errors, \docref{Figure}{fig:nkb-completed} compares autonomous errors with same-capacity controls, and \docref{Table}{tab:runtime-nkb-training} records NKB training costs.

\begin{table}[!htb]
\centering\small
\caption{Initial full-mass projection errors for the 2D NKB adaptation at $K=1024$ (\%). Periodic entries are mean $\pm$ sample SD over three seeds; multi-hole entries average 20 initial states. Trajectory comparisons are shown in \docref{Figure}{fig:nkb-completed}, 3D initial errors in \docref{Figure}{fig:case3-capacity}(b), and NKB training costs in \docref{Table}{tab:runtime-nkb-training}.}
\label{tab:nkb-matched}
\begin{tabular}{lr}\toprule
Case&Initial velocity error (\%)\\\midrule
Periodic &$2.493\pm0.138$\\
One hole&20.760\\
Two holes&22.566\\
Three holes&22.201\\
\bottomrule\end{tabular}
\end{table}

\paragraph{Kinematic residuals and numerical refinement.}
The periodic NKB initial states have a three-seed mean divergence-to-curl norm ratio of $0.195$. Their full-mass condition numbers range from $1.236\times10^{8}$ to $1.899\times10^{8}$. The maximum double-time-step trajectory difference is $4.412\times10^{-12}$. For seed 0, the $128^2$ to $256^2$ operator-quadrature trajectory sensitivity on this same case is $46.100\%$; these measurements separately characterize spatial quadrature and temporal resolution.

For the one-, two-, and three-hole domains, respectively, the mean initial divergence-to-curl ratios are $(0.721,0.770,0.722)$ and the sampled boundary-normal root mean square (RMS) relative to domain velocity RMS is $(0.170,0.216,0.231)$. Maximum trajectory differences under time-step refinement are $(7.03\times10^{-6},1.26\times10^{-5},4.15\times10^{-6})$; those under operator-quadrature refinement are $(1.24\times10^{-2},3.94\times10^{-3},2.99\times10^{-3})$.

\paragraph{Three-dimensional trajectory and kinematic diagnostics.}
\docref{Tables}{tab:nkb-3d-checks} and~\ref{tab:nkb-3d-kinematics} supplement the final errors with trajectory means, full-mass conditioning, and initial/final kinematic residuals. The trajectory mean averages all 30 stored frames, including $t=0$ and $T=1.16$. The mass condition number $\kappa(M)$ is the largest-to-smallest eigenvalue ratio of the symmetric FP64 mass matrix; all $K$ coordinates are retained.

\begin{table}[!htb]
\centering\small
\caption{Completed 3D NKB trajectories and numerical checks. Mean error is relative velocity $L^2$ in percent. Right-hand-side (RHS) discrepancy compares the assembled coefficient RHS with the projected full-grid RHS at $0.2c_0$ and $c_0$, taking the larger relative Euclidean error. The final column is the maximum relative mass-norm trajectory difference between $\Delta t=0.00125$ and $0.000625$. These checks measure operator agreement and time-step sensitivity on the fixed $128^3$ grid.}
\label{tab:nkb-3d-checks}
\begin{tabular}{rrrrr}\toprule
$K$&Mean error (\%)&$\kappa(M)$&RHS discrepancy&Time-step difference\\\midrule
128&90.170&$5.51\times10^{5}$&$2.37\times10^{-10}$&$9.39\times10^{-12}$\\
256&85.572&$5.58\times10^{6}$&$2.04\times10^{-9}$&$2.14\times10^{-11}$\\
512&78.442&$1.17\times10^{8}$&$5.33\times10^{-8}$&$3.86\times10^{-11}$\\
1024&66.826&$4.33\times10^{9}$&$9.92\times10^{-8}$&$5.00\times10^{-11}$\\
\bottomrule\end{tabular}
\end{table}

For the grid-reconstructed velocity (zero in solid cells), let $\mathbf k$ be the grid Fourier wavevector and $\widehat u$ its Fourier transform, with $\|\cdot\|_2$ the Euclidean norm over Fourier samples. Let $\langle\cdot\rangle_{\mathcal B}$ denote the mean over the original boundary band $\mathcal B$, and $\mathcal P_h$ the original grid Leray projection with componentwise dealiasing and the zero-mean constraint. The divergence ratio $D$, boundary-normal ratio $B$, and grid-correction ratio $P$ are
\begin{align*}
D&=\frac{\|i\mathbf k\cdot\widehat u\|_2}{\||\mathbf k|\widehat u\|_2},&
B&=\left(\frac{\langle(u\cdot n)^2\rangle_{\mathcal B}}{\langle|u|^2\rangle_{\mathcal B}}\right)^{1/2},&
P&=\frac{\|\mathcal P_hu-u\|_2}{\|u\|_2}.
\end{align*}
Here $D$ uses the full velocity-gradient norm, whereas the 2D diagnostic uses the curl norm; $B$ is a band-sampled RMS ratio.

\begin{table}[!htb]
\centering\small
\caption{Initial and final 3D NKB kinematic residuals on the original $128^3$ grid. All entries are dimensionless ratios, using the definitions above.}
\label{tab:nkb-3d-kinematics}
\begin{tabular}{rrrrrrr}\toprule
&\multicolumn{2}{c}{Divergence $D$}&\multicolumn{2}{c}{Boundary normal $B$}&\multicolumn{2}{c}{Grid correction $P$}\\
\cmidrule(lr){2-3}\cmidrule(lr){4-5}\cmidrule(lr){6-7}
$K$&Initial&Final&Initial&Final&Initial&Final\\\midrule
128&0.447&0.328&0.395&0.212&0.401&0.312\\
256&0.472&0.385&0.431&0.272&0.373&0.519\\
512&0.430&0.391&0.368&0.222&0.304&0.470\\
1024&0.364&0.385&0.271&0.175&0.224&0.435\\
\bottomrule\end{tabular}
\end{table}

\begin{figure}[!htb]
\centering\includegraphics[width=\linewidth,height=.72\textheight,keepaspectratio]{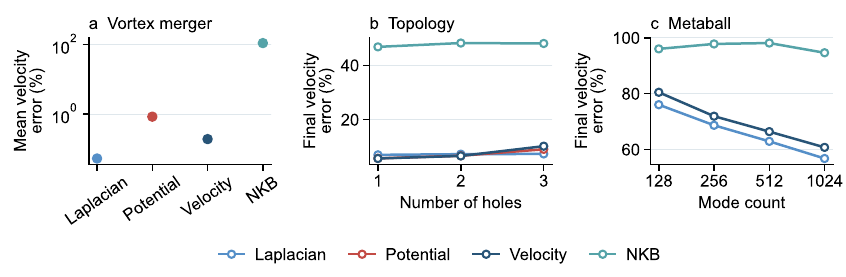}
\caption{\textbf{NKB and Eigenfluid under matched case assessment.} (a) Vortex-merger mean trajectory error over 241 frames, $K=1024$, three learned seeds with one sample SD. (b) Final errors for the three topologies, averaged over 20 states. (c) Metaball capacity comparison against the same original reference. NKB uses the stated 2000-update, common-Galerkin adaptation. Numerical denotes the same-capacity Laplacian control.}
\label{fig:nkb-completed}
\end{figure}

\section{Additional paired flow visualizations}
\label{app:fields}
\subsection{Random large-vortex ablations and matched flow views}
These figures report the random large-vortex ensemble used for the $K=1024$ lift-conditioning comparisons in \docref{Section}{sec:case2}. The direct random-streamfunction ensemble is reported separately in \docref{Appendix}{app:random-streamfunction}.

\begin{figure}[!htb]
\centering\includegraphics[width=\linewidth,height=.72\textheight,keepaspectratio]{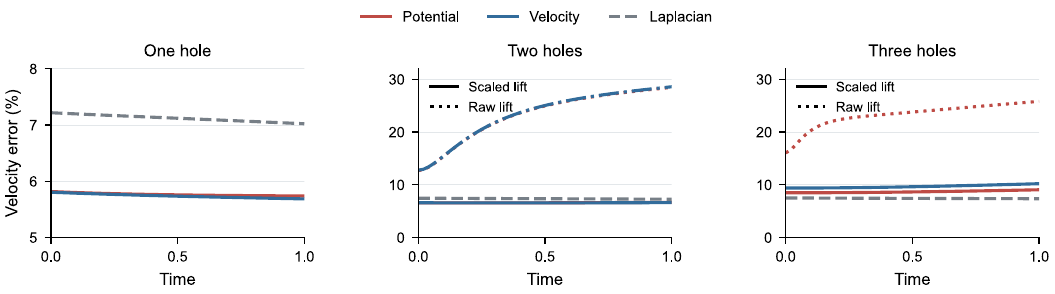}
\caption{\textbf{Harmonic completion, lift conditioning, and objective choice.} One-, two-, and three-hole trajectories use total-$1024$ bases and average the same 20 random large-vortex states. Red/blue denote potential/velocity learning, dashed gray numerical completion, dotted red raw lifts, and dash-dot blue the raw-$H_2$ velocity control. Matched learned bases share raw $H_1$, scaled $H_2$, and scaled $H_3$. Curve endpoints give the final-error comparison; positive-basis training costs appear in \docref{Table}{tab:runtime-training}.}
\label{fig:case2-ablation}
\end{figure}

\subsection{Periodic matched-basis comparison}

\paragraph{Periodic vortex-merger relief.}
\docref{Figure}{fig:case1-field} uses saved potential-trained $K=1024$, seed-0 states at $t=0,9\pi,18\pi$ from the $d=1.35$ trajectory. On the periodic $384^2$ display grid, $S$ smooths $|\omega|$ with Gaussian standard deviation 4.5 pixels; all three times use the fixed original scale $q=0.7395528286634208$ and the common display height $h=0.20-1.40[\operatorname{clip}(S|\omega|/q,0,1)]^{0.70}$. Here $\operatorname{clip}(v,a,b)$ clamps $v$ to $[a,b]$. One subdivision smooths the display mesh. At the middle time, white foam follows 12 signed-smoothed vorticity contour levels, normalized by the same fixed $q$, from 0.065 to 0.69; this exposes the elongated S-shaped merger filaments and pointed outer lobes. At the initial and final times, continuous spiral material accents are anchored to measured vortex cores and cover both deep and shallow relief. These material cues emphasize vortex structure on the data-derived height surface. The oblique camera, lighting, transparent water material, and white ground are shared. Saved velocity, vorticity, and ROM coefficients are unchanged.

\begin{figure}[!htb]
\centering\includegraphics[width=\linewidth,height=.72\textheight,keepaspectratio]{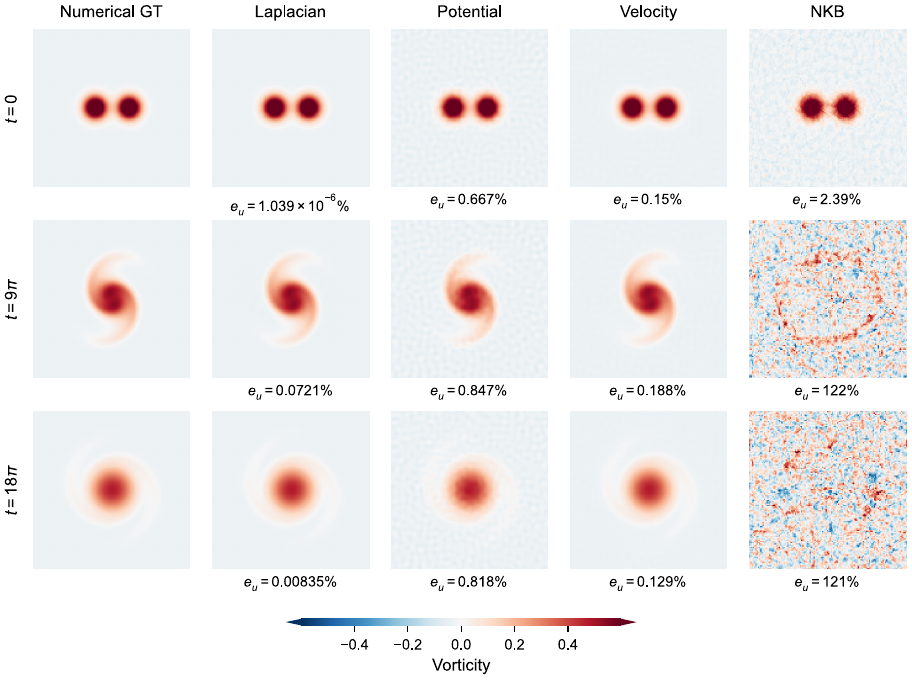}
\caption{\textbf{Vortex merger at $K=1024$.} Columns compare Numerical GT, Laplacian, potential-trained, velocity-trained, and NKB bases; rows show $t=0,9\pi,18\pi$. Learned fields use seed 0 and maximum step $0.005$. Signed vorticity uses the common GT 99.5th-percentile scale; colorbar extensions indicate saturation. Annotations report velocity error over the full $128^2$ observation grid. All methods use the same $d=1.35$ initial velocity and evolve autonomously.}
\label{fig:case1-fourway}
\end{figure}

\subsection{Three-dimensional flow comparison}
\docref{Figure}{fig:3d-context} compares instantaneous velocity streamtubes for the numerical reference, Laplacian basis, and velocity-trained Eigenfluid basis. Shared camera, seeding, material, and speed mapping expose the three-dimensional flow structure behind the common-slice comparison in \docref{Figure}{fig:case3-field}.
\begin{figure}[!htb]
\centering\includegraphics[width=\linewidth,height=.72\textheight,keepaspectratio]{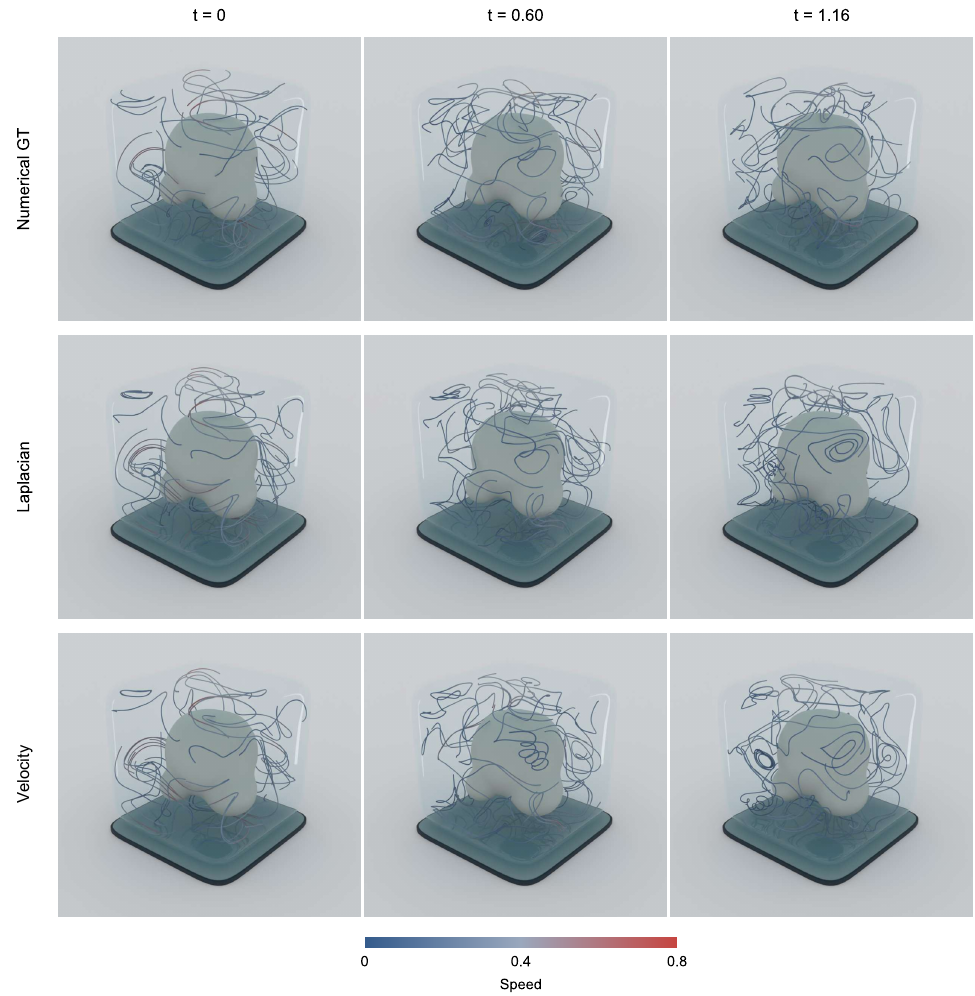}
\caption{\textbf{Three-dimensional flow structure.} Numerical GT, Laplacian, and velocity-trained Eigenfluid $K=1024$ streamtubes at $t=0,0.60,1.16$, colored by speed. The transparent shell follows the rounded domain boundary. Camera, seeding, material, and the $[0,0.8]$ speed range are shared. Centerline smoothing is a display operation.}
\label{fig:3d-context}
\end{figure}

\subsection[Three-hole K1024 time sequence]{Three-hole $K=1024$ time sequence}
\docref{Figure}{fig:case2-h3-water-time} follows the refined reference and all three completed bases through three saved times. One display mapping across methods and times preserves temporal changes in the surface relief.
\begin{figure}[!htb]
\centering\includegraphics[width=\linewidth,height=.72\textheight,keepaspectratio]{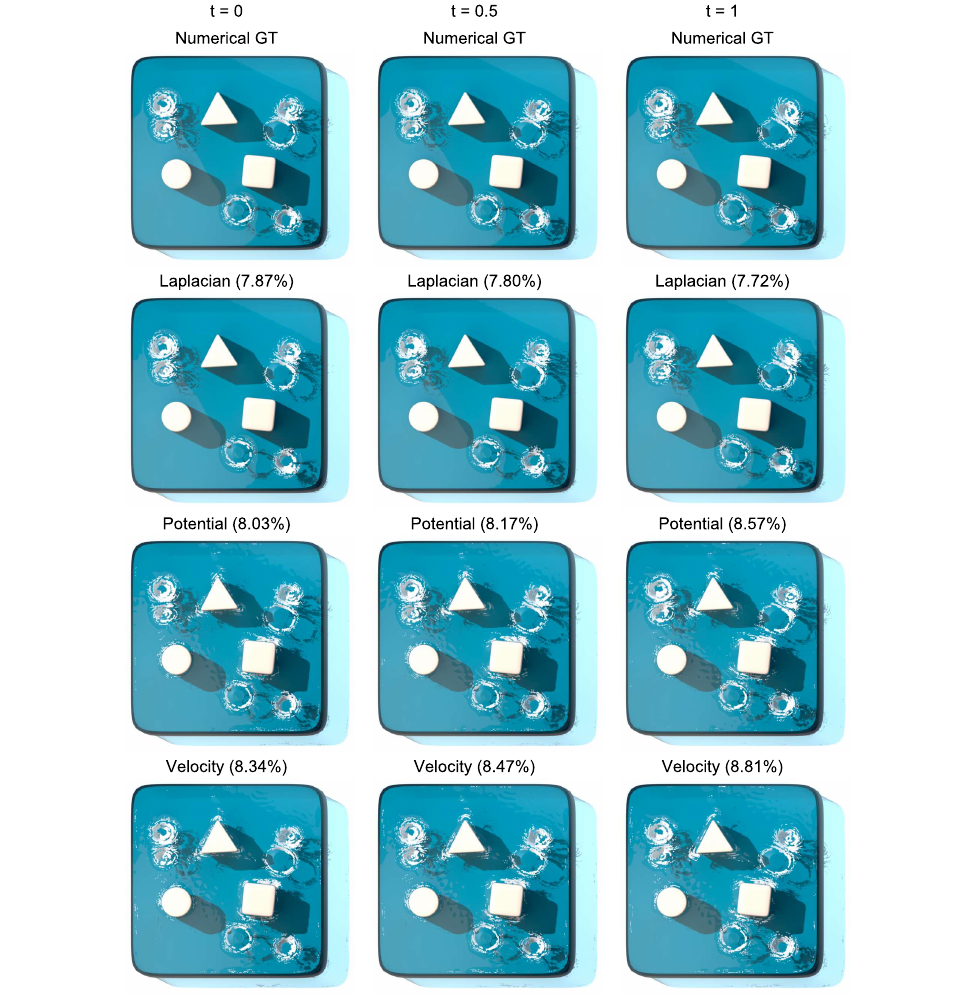}
\caption{\textbf{Three-hole evolution at $K=1024$.} Rows compare Numerical GT, Laplacian, potential-trained, and velocity-trained bases; columns show $t=0,0.5,1$, seed 1, state 00. Each basis contains 1021 positive and three harmonic coordinates; the learned bases share scaled $H_3$. All panels share the height mapping, contour levels, camera, and material. Annotations give full-domain velocity error against the level-2 reference.}
\label{fig:case2-h3-water-time}
\end{figure}

\paragraph{Water-surface display mapping.}
The $K=1024$ renderings in \docref{Figures}{fig:case2-field} and~\ref{fig:case2-h3-water-time} embed the saved two-dimensional fields in a height surface. Analytic derivatives of the learned fields are evaluated from the original checkpoints and contracted with the saved coefficients. The refined GT vorticity is recovered from its stored streamfunction using the original interior finite-element mass--stiffness relation. On a $481\times481$ display grid, let $S$ be a masked Gaussian smoother with standard deviation 8 pixels. Let $Q_p[f]$ be the pooled empirical quantile of field $f$ at probability $p$, $q$ a combined vorticity--speed display signal, and $z_{\rm disp}$ the display height. Define
\begin{align*}
q&=0.8\,\frac{|S(\omega)|}{Q_{0.90}[|S(\omega)|]}
  +0.2\,\frac{S(\|u\|^2)}{Q_{0.90}[S(\|u\|^2)]},\\
z_{\mathrm{disp}}&=0.30-0.75\,\operatorname{clip}\!\left(
\frac{q-Q_{0.02}[q]}{Q_{0.98}[q]-Q_{0.02}[q]},0,1\right).
\end{align*}
Quantiles pool all four methods and all three times within each geometry; the main views and time sequence reuse that common mapping. Filtering signed vorticity before taking its magnitude emphasizes coherent large-scale structure. The speed term also reveals curl-free harmonic circulation. Height is a display coordinate; the saved physical fields and ROM coordinates remain unchanged. White foam follows shared streamfunction contour levels after subtracting each field's nodal mean on the reference mesh; this contour convention differs from the area-weighted centering of the planar scalar plots. The obstacle footprints come from the stored reference-level-2 mesh. Camera and water material are fixed within each plate. \docref{Figures}{fig:case2-field} and~\ref{fig:case2-h3-water-time} use matched orthographic overhead views with a white ground and one subdivision of the display surface. Its foam particles follow the saved contours, with density weighted by $|S\omega|$; within each geometry, all four methods use the same absolute density thresholds taken from the potential-trained $t=1$ field (58th and 95th percentiles), particle-size rule, and random seeds. Foam provides a common material cue for the displayed flow structure.

\section{Computational cost}
\label{app:runtime}
We distinguish basis training, offline preparation, online evolution, and field output.
The tables collect the recorded costs for the case comparisons and their numerical checks.
Times are measured wall times on the RTX 4080 SUPER workstation; each caption specifies
the workload and timer boundary. Superscript R marks synchronized complete-workload measurements with all points and modes. Other entries use the recorded stage timers. Entries marked n.r. have no separate timer, excl. denotes a branch outside the stated comparison, and m.f. denotes matrix-free evaluation. Composite workloads are listed as bundles.
The analytic Fourier basis is available directly; numerical Laplacian eigensolves
belong to offline preparation. Existing Laplacian modes are reused in the
reported controls.

\subsection{Training cost and its components}
Completed 2000-update spectral training costs appear in \docref{Table}{tab:runtime-training} for 2D and \docref{Table}{tab:Pot-training-basis-cost} for 3D. Topology timers cover positive-basis training; the harmonic lift shared within each Potential/Velocity pair is separate. NKB directly predicts velocity without a harmonic training stage; its update, validation, and total timers appear in \docref{Table}{tab:runtime-nkb-training}. The seed-0 Potential-FG run at $K=128$ takes 195.80 minutes; its $K=256$ counterpart and the matched FG, AL, and canonical Potential families are included in the 3D preparation table.

\begin{table}[!htb]\centering\footnotesize\setlength{\tabcolsep}{4pt}
\caption{Two-dimensional spectral training wall time in minutes for complete 2000-update runs. Periodic entries are mean $\pm$ sample SD over three seeds. Topology totals exclude the separately trained harmonic lift shared within each Potential/Velocity pair. Three-dimensional spectral totals are in \docref{Table}{tab:Pot-training-basis-cost}; NKB totals and components are in \docref{Table}{tab:runtime-nkb-training}.}
\label{tab:runtime-training}
\label{tab:k1024-cost}
\begin{tabular}{lrrr}\toprule
Case&$K$&Potential&Velocity\\\midrule
Periodic & 1024 & $106.428\pm3.196$ & $165.760\pm1.636$\\
1 hole & 1024 & 232.61 & 423.97\\
2 holes & 1024 & 235.48 & 376.76\\
3 holes & 1024 & 233.88 & 388.71\\
\bottomrule\end{tabular}\end{table}
\begin{table}[!htb]\centering\footnotesize\setlength{\tabcolsep}{4pt}
\caption{Instrumented 3D Velocity-training components, in minutes, summed over 2000 committed updates and 80 scheduled validations. Fwd./Rev. include analytic derivatives and the reverse pass includes recomputation. The full wall total also includes initialization, progress logging, cleanup, and other uninstrumented overhead; the listed components identify the principal computational stages. Checkpoints are saved every update through $K=256$ and every 25 updates at higher capacities.}
\label{tab:runtime-training-components}

\begin{tabular}{rrrrrrrr}\toprule
$K$&Fwd.&Rev.&Geometry&Ortho.&Validate&Save&Opt./sample\\\midrule
128 & 10.52 & 35.97 & 25.14 & 0.21 & 8.35 & 4.67 & 0.22\\
256 & 20.85 & 72.05 & 25.22 & 0.62 & 10.76 & 9.26 & 0.22\\
512 & 42.17 & 145.88 & 25.61 & 2.20 & 15.95 & 1.82 & 0.35\\
1024 & 89.94 & 298.41 & 29.88 & 8.86 & 28.61 & 3.77 & 0.67\\
\bottomrule\end{tabular}\end{table}
\paragraph{Architecture and derivative workload.}
The 3D velocity implementation assigns an independent width-128, five-hidden-layer sinusoidal representation network (SIREN) {\citep{sitzmann2020implicit}} to each mode: at $K=1024$ it has 70,126,592 trainable parameters. The shared NKB network has 1,188,352 parameters (\docref{Appendix}{app:nkb}). For Velocity, the neural vector potential $A_\theta$ gives $u=\curl A_\theta$ and the velocity-gradient Rayleigh energy require second spatial derivatives of $A_\theta$ and their parameter adjoints. NKB directly predicts $u$ and its losses use first spatial derivatives. At $K=1024$, neural forward and reverse account for 83.3\% of full training wall time; gauge projection is disabled. The cost breakdown connects the derivative workload and network architecture to the measured training time.

\begin{table}[!htb]\centering\footnotesize\setlength{\tabcolsep}{4pt}
\caption{NKB training timers, in seconds. Updates sum the 2000 measured optimization iterations, including sampling, analytic network forward/reverse, losses, Adam, and checkpoint writes within those iterations. Validation sums the 80 scheduled events. Total additionally includes preflight, initialization, and other run overhead. Topology seed is 1 and the 3D seed is 20260823; update totals measure the combined neural forward/reverse and optimization workload.}
\label{tab:runtime-nkb-training}

\begin{tabular}{lrrrrr}\toprule
Case&$K$&Seed&Updates&Validate&Total\\\midrule
Periodic & 1024 & 0 & 151.75 & 3.17 & 159.13\\
Periodic & 1024 & 1 & 151.59 & 3.19 & 158.98\\
Periodic & 1024 & 2 & 151.45 & 3.17 & 158.78\\
1 hole & 1024 & 1 & 153.04 & 3.28 & 160.96\\
2 holes & 1024 & 1 & 154.53 & 3.33 & 162.52\\
3 holes & 1024 & 1 & 159.76 & 3.52 & 167.91\\
3D & 128 & 20260823 & 41.68 & 4.72 & 50.63\\
3D & 256 & 20260823 & 61.90 & 7.01 & 73.30\\
3D & 512 & 20260823 & 121.50 & 13.69 & 139.66\\
3D & 1024 & 20260823 & 325.27 & 37.50 & 367.58\\
\bottomrule\end{tabular}\end{table}

\subsection{Offline basis, projection, and operator preparation}
For bounded domains, basis inference evaluates the frozen NKB model on every assessment node and all modes.
It is the median of three synchronized passes including coordinate/input transfer
and checksum reduction, excluding derivatives and output storage. Projection/preparation
contains the recorded full-mass initial projection and reference/error-moment preparation.
Derived-field preparation contains the derivative caches and linear operators.
Each independently timed pass characterizes its stated workload; production queries reuse the prepared basis and operators.

\begin{table}[!htb]\centering\footnotesize\setlength{\tabcolsep}{4pt}
\caption{Offline costs in seconds. Basis gives all-node field-evaluation times; periodic entries include analytic derivatives. Projection includes the recorded reference preparation where bundled. Derived means field/linear-operator preparation, and $Q$ denotes accepted assembly of the quadratic tensor $\mathcal Q$. The incl. periodic derivative fields are included in basis evaluation, and their nonlinear loads are matrix-free. n.r. marks components without a separate measurement. The $K=512$ derived cost covers all 512 modes, basis-cache transposition, and linear-operator preparation. These columns exclude neural training and online integration.}
\label{tab:runtime-offline}

\begin{tabular}{lrlrrrr}\toprule
Case&$K$&Basis&Basis eval.&Projection&Derived&$Q$\\\midrule
3D & 512 & Velocity & $454.46^{\mathrm R}$ & 18.50 & 3036.39 & 3645.51\\
3D & 1024 & Velocity & $679.90^{\mathrm R}$ & 41.25 & 10381.08 & 32434.41\\
Periodic s0 & 1024 & Potential & 143.21 & 1.73 & \textit{incl.} & \textit{m.f.}\\
Periodic s1 & 1024 & Potential & 142.83 & 1.73 & \textit{incl.} & \textit{m.f.}\\
Periodic s2 & 1024 & Potential & 143.98 & 1.63 & \textit{incl.} & \textit{m.f.}\\
Periodic s0 & 1024 & Velocity & 144.22 & 1.63 & \textit{incl.} & \textit{m.f.}\\
Periodic s1 & 1024 & Velocity & 149.66 & 1.72 & \textit{incl.} & \textit{m.f.}\\
Periodic s2 & 1024 & Velocity & 141.75 & 1.61 & \textit{incl.} & \textit{m.f.}\\
Periodic s0 & 1024 & NKB & 1.36 & 1.68 & \textit{incl.} & \textit{m.f.}\\
Periodic s1 & 1024 & NKB & 1.10 & 1.64 & \textit{incl.} & \textit{m.f.}\\
Periodic s2 & 1024 & NKB & 1.09 & 1.63 & \textit{incl.} & \textit{m.f.}\\
1 hole & 1024 & NKB & 3.55 & 31.34 & $22.58^{\mathrm R}$ & 3247.81\\
2 holes & 1024 & NKB & 14.02 & 118.76 & $90.01^{\mathrm R}$ & 12682.87\\
3 holes & 1024 & NKB & 13.88 & 118.83 & $90.00^{\mathrm R}$ & 12685.85\\
3D & 128 & NKB & 2.99 & 30.18 & 139.57 & 83.94\\
3D & 256 & NKB & 3.50 & 95.56 & 311.72 & 576.10\\
3D & 512 & NKB & 4.67 & 190.53 & 728.16 & 3841.67\\
3D & 1024 & NKB & 7.20 & 322.78 & 1939.69 & 32446.09\\
\bottomrule\end{tabular}\end{table}
For periodic flow, basis evaluation includes the analytic velocity and curl fields on all $256^2$ quadrature nodes; projection includes the full mass/curl-stiffness assembly and generalized eigensolve. The rotational nonlinear load is evaluated matrix-free.
The synchronized 3D Velocity basis pass evaluates analytic curl on all $1.6\,\mathrm{M}$ nodes in the original FP32 spatial arithmetic, with FP64 checksum accumulation; coordinate/input transfer and checksum reduction are included, and storage is excluded. Its checksum agrees with the original full-mass basis energy to $1.10\times10^{-14}$ relative error. The topology basis passes evaluate every field and derivative and assemble the full mass, derivative, and divergence moments on 1,048,576 nodes for one hole and 4,194,304 for two and three holes. Each saved matrix agrees with its original operator to the reported precision.

\subsection{Online evolution, output, and assessment}
\docref{Table}{tab:nkb-runtime} separates coefficient evolution from output costs in 3D.
Laplacian controls use the original grid-based reduced generator, whereas the listed
Velocity and NKB trajectories use assembled coefficient operators. The NS reference
uses its original split viscosity/normal-penalty integrator. The common physical interval and fine time step support comparison of the recorded execution strategies.
Flow errors accompany these timings in the capacity plots. Output measures field
reconstruction, transfer, storage, and recorded diagnostics.

\begin{table}[!htb]\centering\footnotesize\setlength{\tabcolsep}{4pt}
\caption{3D online costs in seconds, $T=1.16$. Coarse/fine steps are $0.00125/0.000625$. ROM output timers cover three selected fields; NS output timers cover its 30 saved fields. Laplacian online timers include the field operations required at every integration step, excluding assessment outside the stepping timer. Repeated NS rows are the measured reference runs for the two capacity experiments; n.a. indicates that NS has no retained-mode count.}
\label{tab:nkb-runtime}

\begin{tabular}{rllrrr}\toprule
$K$&Method&Execution&Coarse&Fine&Output\\\midrule
128 & NKB & Coefficient & 3.57 & 6.84 & $22.30^{\mathrm R}$\\
256 & NKB & Coefficient & 4.90 & 9.78 & $22.26^{\mathrm R}$\\
512 & Velocity & Coefficient & 3.85 & 7.03 & 20.41\\
512 & Laplacian & Grid ROM & 119.45 & 238.71 & $0.26^{\mathrm R}$\\
512 & NKB & Coefficient & 8.29 & 16.58 & 7.98\\
1024 & Velocity & Coefficient & 24.30 & 48.59 & 32.72\\
1024 & Laplacian & Grid ROM & 130.04 & 260.17 & $0.26^{\mathrm R}$\\
1024 & NKB & Coefficient & 766.10 & 1532.83 & 14.56\\
\textit{n.a.} & NS (K512 comparison) & Grid NS & $63.71^{\mathrm R}$ & 124.75 & 7.01\\
\textit{n.a.} & NS (K1024 comparison) & Grid NS & $63.66^{\mathrm R}$ & 126.68 & 7.04\\
\bottomrule\end{tabular}\end{table}
\begin{table}[!htb]\centering\footnotesize\setlength{\tabcolsep}{4pt}
\caption{2D measured execution intervals in seconds, $K=1024$. Each periodic method--seed pair is listed once. Periodic runs reach $T=18\pi$ with 241 saved fields; topology runs reach $T=1$ with 20 initial states. Periodic integration timers include synchronized coefficient capture and exclude field reconstruction and file writing. The matched potential, velocity and NKB runs use $256^2$ operator points and $128^2$ observations. Topology entries time their original 20-state batches.}
\label{tab:runtime-2d-online}

\begin{tabular}{llrll}\toprule
Case&Method&$\Delta t$&Time (s)&Timer scope\\\midrule
Periodic s0 & Potential & $0.005$ & 182.62 & Integration + coefficient capture\\
Periodic s1 & Potential & $0.005$ & 182.63 & Integration + coefficient capture\\
Periodic s2 & Potential & $0.005$ & 182.64 & Integration + coefficient capture\\
Periodic s0 & Velocity & $0.005$ & 185.97 & Integration + coefficient capture\\
Periodic s1 & Velocity & $0.005$ & 183.62 & Integration + coefficient capture\\
Periodic s2 & Velocity & $0.005$ & 181.99 & Integration + coefficient capture\\
Periodic s0 & NKB & $0.005$ & 182.99 & Integration + coefficient capture\\
Periodic s1 & NKB & $0.005$ & 182.81 & Integration + coefficient capture\\
Periodic s2 & NKB & $0.005$ & 182.59 & Integration + coefficient capture\\
1 hole & NKB & $0.001$ & 408.46 & 20-state coefficient batch\\
2 holes & NKB & $0.001$ & 402.47 & 20-state coefficient batch\\
3 holes & NKB & $0.001$ & 418.86 & 20-state coefficient batch\\
\bottomrule\end{tabular}\end{table}
\begin{table}[!htb]\centering\footnotesize\setlength{\tabcolsep}{4pt}
\caption{High-capacity Velocity pipeline wall times, in seconds. Process-stage intervals are grouped by purpose for each 20-state topology experiment. They include process startup, I/O, and process monitoring, and are kept separate from the synchronized NKB coefficient timers in \docref{Table}{tab:runtime-2d-online}. Shared harmonic training is excluded.}
\label{tab:runtime-topology-pipeline}

\begin{tabular}{lrrr}\toprule
Stage&1 hole&2 holes&3 holes\\\midrule
Boundary fields and assessment & 60.05 & 70.07 & 70.04\\
Initial projection & 20.01 & 40.03 & 40.02\\
Linear/operator preparation & 874.27 & 977.54 & 978.13\\
Nonlinear tensor & 3484.71 & 3504.59 & 3501.97\\
ROM execution & 520.06 & 610.08 & 610.08\\
Reference fields and error assessment & 570.13 & 2250.30 & 2280.28\\
\bottomrule\end{tabular}\end{table}
The one-hole kernel comparison in \docref{Appendix}{app:protocol} isolates coefficient stepping from streamed field contraction. Its kernel-level intervals differ from the process-stage workloads in \docref{Table}{tab:runtime-topology-pipeline}.

\subsection{Numerical checks and auxiliary costs}
\begin{table}[!htb]\centering\footnotesize\setlength{\tabcolsep}{4pt}
\caption{Additional 3D measured costs in seconds. Velocity check includes full-grid RHS and one-step comparisons; refined projection uses the matched-step NS reference. Velocity/NKB evaluation bundles include both trajectories, checks, and output and overlap the preceding component timers. NKB capacities time evaluation after operator preparation. Laplacian refresh includes both time steps and fresh FP64 projection. Bundle totals and individual components describe their respective timer boundaries.}
\label{tab:runtime-3d-checks}

\begin{tabular}{rrrrrrr}\toprule
$K$&Velocity load&Velocity check&Refined proj.&Velocity eval.&NKB eval.&Lap. refresh\\\midrule
512 & 0.71 & 183.30 & 19.37 & 386.33 & 69.50 & 376.58\\
1024 & 3.36 & 438.64 & 35.32 & 761.22 & 2337.10 & 407.83\\
\bottomrule\end{tabular}\end{table}
At $K=128,256$, joint Velocity/Laplacian evaluation costs 547.33 and 3449.86 s, covering both methods, both time steps, and their field outputs. The corresponding NKB evaluation bundles take 25.91 and 30.12 s. These totals characterize the complete evaluation workloads at the two capacities.

\paragraph{NKB operator precision.}
The accepted $K=1024$ tensor assembly uses FP64 short-dot products and compensated accumulation. Every online RK4 stage likewise uses compensated products and summation. Their measured costs appear in Tables~\ref{tab:runtime-offline} and~\ref{tab:nkb-runtime}. Independent full-grid RHS comparisons verify the assembled operator to the reported precision (\docref{Table}{tab:nkb-3d-checks}).

\paragraph{Gauge-solver timings.}
\docref{Table}{tab:gauge-solver-ablation} groups setup, forward, reverse, and independent
assessment at the frozen checkpoint, with residuals visualized in \docref{Figure}{fig:gauge-solver-ablation}.
\docref{Table}{tab:gauge-precision} adds source-evaluation times for the precision controls.
The workload evaluates 32 frozen columns with reused geometry-dependent MFS factors;
the listed timers cover the source evaluation, correction, differentiation, and assessment stages.

\subsection{Inference accounting and reuse thresholds}
\label{app:inference-amortization}
Basis-query timing follows the continuous evaluation boundary in \docref{Section}{sec:inference-cost}; numerical construction additionally charges geometry-specific mesh preparation, point mapping, factorization, and the spectral solve. \docref{Table}{tab:runtime-offline} reports all-node evaluation costs, while \docref{Appendix}{app:unseen-geometry} gives fresh numerical constructions on unseen domains. The reuse counts below concern trajectories that share spatial bases and operators.

\begin{table}[!htb]\centering\small
\caption{Reuse thresholds computed from original wall-time receipts, without rounding inputs. A reuse is one $T=1.16$, 1856-step trajectory on the same geometry, with the same time step and existing operators. The baseline is the numerical NS integration timer. Training-only counts set $T_{\rm off}=0$; adding recorded assembly or derived-field preparation charges those stages once. Derived-field costs include the recovered $K=512$ preparation timers. Per-query projection/output and separately unrecorded stages lie outside these integration-only thresholds.}
\label{tab:inference-amortization}
\begin{tabular}{rrrr}\toprule
$K$&Training only&+ $Q$ assembly&+ derived fields\\\midrule
512&123&154&179\\
1024&359&774&907\\
\bottomrule\end{tabular}\end{table}

For a full workload, replace the denominator of \doceqref{eq:inference-break-even} by the difference between complete per-query costs, including initial projection, selected field reconstruction, transfers, and storage. Output costs are assigned according to the requested number of reconstructed fields: the recorded NS workload saves 30 fields, and the ROM workload reconstructs three selected fields.

For basis generation over a geometry family, write $T_{\rm num}^{(g)}$ for a fresh numerical basis construction and $T_{\rm query}^{(g)}$ for the frozen neural map with the same postprocessing. For $G$ geometry queries indexed by $g$, training amortizes when
\begin{equation}
\sum_{g=1}^{G}\left(T_{\rm num}^{(g)}-T_{\rm query}^{(g)}\right)\geq T_{\rm train}.
\end{equation}
If the per-geometry costs are constant, write them as $T_{\rm num}$ and $T_{\rm query}$; then $G_\star=\lceil T_{\rm train}/(T_{\rm num}-T_{\rm query})\rceil$. This geometry-level threshold requires matched basis-generation timers. Each new geometry also incurs its own operator assembly when autonomous dynamics are requested. Arbitrarily dense neural spatial queries preserve the same learned continuous function; their cost scales with requested points and modes, and their accuracy is that of the learned representation.

\paragraph{Tensor-product and nonuniform parameter sweeps.}
A tensor-product sweep with $m$ values for each of $d$ independent query parameters contains $m^d$ cases. The per-axis break-even resolution $r_\star(d)$ and integer sampling count $m_\star(d)$ are
\begin{equation}
r_\star(d)=n_\star^{1/d},\qquad
m_\star(d)=\left\lceil r_\star(d)\right\rceil.
\label{eq:sweep-resolution}
\end{equation}
Here $d$ counts varied physical or query parameters and $n_\star$ is the integration-only reuse count in \doceqref{eq:inference-break-even}.

Unequal axis counts $m_1,\ldots,m_d$ meet the threshold when $\prod_{j=1}^d m_j\geq n_\star$; their geometric mean $(\prod_jm_j)^{1/d}$ is the mean per-axis resolution. Sparse, constrained, or correlated sweeps use the actual query count. With per-query savings $\Delta T_i$ and reusable cost $T_{\rm one\mbox{-}time}$, break-even requires $\sum_i\Delta T_i\geq T_{\rm one\mbox{-}time}$, including per-case preparation. Fixed-geometry initial-condition sweeps reuse the frozen basis and operators; geometry sweeps additionally charge geometry generation and required operator assembly.

\end{document}